%% file: Permut.tex
\documentclass[11pt]{book}

\usepackage[toc,page]{appendix}

\usepackage{caption}
\usepackage{subcaption}

\usepackage{mathptmx,amsthm,amsmath,tikz}
\usepackage{amsfonts, amssymb}
\usepackage{cite}
\usepackage{cmap}
\usepackage[T1]{fontenc}
\usepackage{subfig}

\usepackage{fancyhdr}

\include{header}

\begin{document}

\begin{titlepage}

\begin{center}
{\bf 
{\LARGE
Computational Complexity of Some Quantum Theories in $1+1$ Dimensions}
\\[2mm]
}
\end{center}
\vspace{4mm}
\begin{center}
{\bf Saeed Mehraban}\\
\vspace{0.7cm}
{\it Computer Science and Artificial Intelligence Laboratory}\\
{\it Massachusetts Institute of Technology, USA}\\

{\small \tt  mehraban@mit.edu}
\end{center}
\vspace{8mm}

While physical theories attempt to break down the observed structure and behavior of possibly large and complex systems to short descriptive axioms, the perspective of a computer scientist is to start with simple and believable set of rules to discover their large scale behaviors. Computer science and physics, however, can be combined into a new framework, wherein structures can be compared with each other according to observables like mass and temperature, and also complexity at the same time. For example, similar to saying that one object is heavier than the other, we can discuss which system is more complex. According to this point of view, a more complex system can be interpreted as the one which can be programmed to simulate the behavior of the others.

The aim of this thesis is to exemplify this point of view through an analysis of certain quantum theories in two dimensional space-time. In simple words, these models are quantum analogues of elastic scattering of colored balls moving on a line. Physical examples that motivate this are the factorized scattering matrix of quantum field theory, and the repulsive delta like collisions in $1+1$ dimensions. 

Classical intuition suggests that when two hard balls collide, they bounce off and remain in the same order. However, in the quantum setting, during a collision, either the balls bounce off, or otherwise they tunnel through each other and exchange their configurations. As a result, moving balls are put into a superposition of being in different configurations. Thereby, considering $n$ distinguishable balls, the Hilbert space is generated by orthonormal basis marked with the $n!$ possible permutations of an $n$-element set, and collisions act similar to local permuting quantum gates. We therefore study the space of unitary operators generated by these local permuting gates. 

First, quantum ball permuting model is defined as a generalized unitary model which simulates the discussed scattering models as its special case, and then the class of problems that are efficiently solvable by this model is partially pinned down within known complexity classes. We find that the complexity class essentially depends on the initial superposition of the balls. More precisely, if the balls start out from the identity permutation, additive approximation of the amplitudes in this model can be efficiently computed within $\DQC 1$, which is believed to be strictly weaker than the standard model of quantum computing. Similar result also applies to the integrable models of scattering, if no initial superposition is provided and the particles are considered to be distinguishable. On the other hand, if special initial superpositions are allowed, the result is that the quantum ball permuting model can efficiently sample from the output distribution of standard quantum computers. Then, we show how to use intermediate demolition measurements in the particle label basis to simulate the quantum ball permuting model with scattering amplitudes of repulsive delta interactions, nondeterministically. According to this result, using post-selection on the possibly exponentially small outcomes of these measurements, one obtains the original ball permuting model. Therefore, the post-selected analogue of repulsive delta interactions model can efficiently simulate standard quantum computers, when arbitrary initial superpositions are allowed. Using this observation, we formalize a scattering quantum computer based on delta-repulsive collisions and intermediate detections, and then we prove that the possibility of an efficient classical simulation for this model is ruled out, unless the polynomial hierarchy collapses to its third level.

A classical analogue of ball permutation is also defined as a model of computation, and its computational power is pinned down within the complexity classes below $\BPP$ and $\NP$. More specifically, two models are considered, deterministic and randomized ball permutation, both defined with $\AC^0$ pre-processing. An equivalence between deterministic ball permutation and $\LOGSPACE$ computation is demonstrated. For the randomized ball permutation, it is proved that the class of problems that are efficiently solvable by this model lies between $\BPL$ and $\Almost \L$. Moreover, we discuss a nondeterministic model of ball permutation, and show that with polynomial time pre-processing, the class of languages that are decidable by this model is equivalent to the class $\NP$. However, we demonstrate that if ball permutation is restricted to adjacent swaps only, then the class is contained in $\P$.

\vspace{50mm}
{\footnotesize
The material presented here is based on the author's Master's thesis, advised by Scott Aaronson, submitted to the department of electrical engineering and computer science at MIT on August 28, 2015. Editions and modifications has been made to the original thesis, also a new chapter, chapter \ref{ch4} is added. Chapter \ref{ch4} is the result of collaboration with Scott Aaronson. Sections \ref{trace} and \ref{classification} are the result of collaboration and discussions with Greg Kuperberg.}

\end{titlepage}

\vspace{4mm}

\vspace{-6mm} 
\tableofcontents 
\newpage

\chapter{Introduction}

\section{Motivating Lines}

\noindent What happens when computer science meets physics? The Church-Turing thesis states that \textit {all that is computable in the physical universe is also computable on a Turing machine} ~\cite{AarDem, Sipser}. More than a mathematical statement, this is a conjecture about theoretical physics. An outstanding discovery of computability theory was the existence of undecidable problems ~\cite{Turing}; problems that are not decidable by Turing machines. Therefore, the Church-Turing thesis can be falsified if there exists a physical system that can be programmed to decide an undecidable problem. The Church-Turing thesis was then further extended to another conjecture: \textit{all that is \textit{efficiently} computable in the physical universe is also \textit{efficiently} computable by a probabilistic Turing machine}. An efficient process is defined to be the one which answers a question reliably after time polynomial in the number of bits that specify the question. The extended Church-Turing thesis looks likely to be defeated by the laws of quantum physics, as the problem of factoring large numbers is efficiently solvable on a quantum computer \cite{Shor}, while yet no polynomial time probabilistic algorithm is not known for it. If this is true, then the revised thesis is that \textit{all that is efficiently computable in the physical universe, is also efficiently computable on a quantum computer}. As the preceding discussion illustrates, a natural approach is to classify the available physical theories with their computational power, both according to the notion of complexity and computability \cite{AarPhReality}. There are many examples that are known to be equivalent to classical Turing machines \cite{yao2003classical, AarPhReality, valiant2005classical, terhal2002classical}, and also other equivalents of the standard quantum computers exist \cite{lloyd1995almost, freedman2003topological}. Among the available computing models are some that are believed to be intermediate between classical polynomial time and quantum polynomial time \cite{knill1998power, aaronson2011computational, jordan2009permutational}. These are models that are probably strictly weaker than standard quantum computers, but they still can solve problems that are believed to be intractable on a classical computer.

In this thesis, we try to apply these ideas to some physical theories, involving scattering of particles. The goal is to figure out which problems these models can solve. More specifically, the aim is to find out if these models are equivalents of standard quantum computers, intermediate between quantum and classic computing, or if they can be efficiently simulated on computers. Scattering amplitudes are central to quantum field theory \cite{arkani2012scattering}. They relate the asymptotic initial states of a quantum system to the final states, and therefore they can be viewed as notable observables of quantum theory. While in general scattering amplitudes are sophisticated objects, in integrable theories of $1+1$ dimensions \cite{staudacher2012review, ahn2012review} they take simple forms, and can each be described by a single diagram with four-particle vertices. These diagrams encode the overall scattering matrix, whose effect can be placed in one-to-one correspondence with a permutation of a finite set \cite{yang1967some}. A crucial element of these integrable theories is a factorized scattering matrix\cite{zamolodchikov1978relativistic}. In this case, the scattering matrix can be decomposed as a product of local unitary scattering matrices. These local matrices satisfy the well-known Yang-Baxter \cite{yang1967some,baxter1972partition} relations, and it can be demonstrated that Yang-Baxter relations impose special symmetries on the diagrams in such a way that each diagram can be consistently assigned to a permutation. Such drastic simplification is directly related to the existence of an infinite family of conservation rules, first noticed by Zamolodchikov et. al. \cite{zamolodchikov1979factorized}.  In general, the set of unitary scattering matrices can form a manifold of dimension exponential in the number particles. However, it can be shown that the infinite family of conservation rules shrinks the number of degrees of freedom drastically, to linear in the number of particles.

\section{Methods and Summary of the Results}

Given these amazing features of the integrable quantum models, it is interesting to use tools from complexity theory to understand how hard these models are to simulate. Specifically, we analyze the situation where all the particles are distinguishable. In the language of quantum field theory, this is the situation where infinite colors are allowed.

The standard model of language in quantum computation is the class $\BQP$, which is the set of problems that are efficiently solvable by local quantum circuits, and a quantum model is called $\BQP$-universal if it can efficiently recognize the same set of languages. Bits of quantum information are called qubits, the states of a system which can take two states in a superposition. Therefore, in order to have a reference of comparison, we will try to relate the state space of the integrable quantum theory of scattering to bits and qubits. My approach is to define different variations of the scattering model, as new models, and demonstrate reductions between them one by one.

We define the ball permuting model as a quantum model with a Hilbert space consisting of permutations of a finite element set as its orthogonal basis. Then, the gates of ball permuting model act by permuting states like $|x,y\rangle $ according to $|x,y\rangle \rightarrow c |x,y\rangle + i s |y,x\rangle$, where $x$ and $y$ are members of a finite element set, and $c$ and $s$ are real numbers with $c^2+s^2=1$.  We prove that if the ball permuting model starts out with an initial state of $|123\ldots n \rangle$, then approximation of single amplitudes in this model, within additive error, can be obtaine within the so-called one-clean-qubit model, also known as the complexity class $\DQC 1$ \cite{knill1998power}. On the other hand, we demonstrate that if the model is allowed to start out from arbitrary initial states, then there is a way to simulate $\BQP$ within the ball permuting model. we also consider a variant of the ball permuting model, wherein the action of the gates are according to $|x,y\rangle \rightarrow c_{x,y} |x,y\rangle + i s_{x,y} |y,x\rangle$. Here $c_{x,y}$ and $s_{x,y}$ are real numbers that depend on the labels $x$ and $y$ only, and also $c^2_{x,y}+s^2_{x,y}=1$. We demonstrate that this model can directly simulate $\BQP$ on any initial state, including the identity $|123\dots n\rangle$. After that, we do a partial classification on the power of ball permuting model on different initial states. The classification is according to the Young-Yamanouchi orthonormal basis \cite{james1981representation}, which form the irreducible representations of the symmetric group. 

We provide evidence that although scattering matrices generated within the discussed $1+1$ dimensional integrable models correspond to unitary manifolds with linear dimensionality in the number of particles, it is hard to simulate them on a classical computer if we equip these models with arbitrary initial states and intermediate demolition measurements. For this purpose, we show that by postselecting \cite{aaronson2005quantum} on possibly exponentially-unlikely measurement outcomes, the model can efficiently solve any problem in the complexity class $\Post \BQP$. Then using same line of reasoning as in ~\cite{bremner2010classical}, one can infer that the existence of an efficient procedure to sample from the distribution of outcomes in the proposed model within multiplicative error implies the collapse of polynomial hierarchy to the third level.

In order to obtain a point of reference with classical computation, a model of ball permutation with classical balls is formalized. In this model, access to ball permutation is provided for an $\AC^0$ machine as an oracle, and the machine can make polynomially-long queries to the oracle. Ball permuting oracles are defined in two different ways; deterministic and randomized ones. Inputs to a deterministic ball permuting oracle are lists of swaps, and outputs are the permutations that are resulted from the application of swaps in order. A randomized ball permuting oracle also takes a list of probabilities as input, applies the swaps probabilistically and outputs the final permutation. The model corresponding to the deterministic ball permutation is proved to be equivalent to $\LOGSPACE$ Turing machines. The randomized ball permutation, on the other hand, can simulate $\BPL$ machines efficiently. However, it is proved that a machine from the class $\Almost \L$ can efficiently simulate randomized ball permutation. Further pinning down of the randomized ball permuting model between $\BPL$ and $\Almost\L$ is left as an open problem. Also, the relationship between the randomized ball permutation and polynomial-time computation is unknown. Other than deterministic and randomized models, a nondeterministic ball permuting model is defined to be the class of problems that are polynomial-time reducible to the problem of deciding if a target permutations in the randomized ball permutation can ever be generated. This class is proved to be equivalent to $\NP$, however, if the all of the queried swaps of the randomized computation are adjacent ones, a polynomial-time simulation is demonstrated for the model. 

\section{Summary of the Chapters}

The thesis consists of five chapters. In chapter \ref{ch2}, we review the essential background about computability and complexity. We start by defining alphabets, and proceed to the Turing machine as the well-accepted model of computation. After discussing some ingredients of computability theory, we talk about complexity theory, and bring the definitions for well-known complexity classes that are related to this thesis. Then we discuss circuits, which are essential ingredients of quantum computing.

In chapter \ref{ch3}, we quickly review quantum mechanics, and end my discussion with scattering amplitudes and quantum field theory. We then talk about quantum complexity theory. Finally, we review relevant integrable models in two dimensional space-time, both in quantum field theory, and quantum mechanics. 

Chapter \ref{ch4} and \ref{ch5} are dedicated to the results. The results of chapter $4$ are obtained with joint collaboration with Scott Aaronson. In this chapter a classical analogue of the ball permuting model is formalized and its computational power is pinned down within standard complexity classes. Three major complexity classes are defined. The first of these is the deterministic ball permuting model $\DBALL$, where an $\AC^0$ machine has access to a deterministic ball permuting oracle. Such an oracle takes as input a polynomially-long list of swaps and returns the permutation obtained by applying those swaps in order to the identity permutation. The result is an equivalence between $\DBALL$ and $\L$ ($\LOGSPACE$). The second model, $\RBALL$, is a randomized ball permuting model, where an $\AC^0$ machine has oracle access to a randomized ball permuting oracle. Such an oracle, along with the list of swaps, inputs a list of probabilities, and applies the swaps probabilistically. The major result is the containment of $\BPL$ (bounded error probabilistic $\LOGSPACE$) in $\RBALL$, and the containment of $\RBALL$ in $\Almost \L$ ($\LOGSPACE$ with access to a random oracle). The third model is $\NBALL$, which is the class of problems that are polynomial time reducible to the following problem: given a list of probabilistic swaps, decide if a target permutation can ever be generated. It turns out that $\NBALL=\NP$, and if the queried swaps are all adjacent ones, then $\NBALL \subseteq \P$. Also in order to show the relevance with the problem of ball (particle) scattering, a classical analogue of the Yang-Baxter equation is described.

The results of chapter \ref{ch4} are obtained with join collaboration with Greg Kuperberg. In this chapter, we formally define the languages and the variants of the quantum ball permuting model, and pin them down within the known complexity classes. Specifically, for the ball permuting model with the initial state $|123\ldots n\rangle$ we prove that single (permutation) amplitudes of this model can be approximated with rounds of $\DQC 1$ computation. Then we introduce a ball scattering computer based on the repulsive delta interactions model with intermediate demolition measurements. In order to partially classify ball permuting model on different initial states, we borrow tools from the decoherence free subspaces theory and representation theory of the symmetric group, which are reviewed when needed. After this classification, we demonstrate explicitly how to program the ball permuting model to simulate $\BQP$, if we are allowed to initialize the ball permuting model with any superposition that we want. Then in the end, we put everything together to demonstrate that the output distribution of the ball scattering computer cannot be simulated efficiently, unless the polynomial hierarchy collapses to the third level.

\section{Open Problems}

A detailed list of open problems and further directions is provided at the end of chapter \ref{ch4} and chapter \ref{ch5}. Here we include the major open problems and possible directions for further research:

\begin{itemize}
\item[] 1. As discussed above if the quantum ball permuting model starts out of the identity permutation quantum state, then the additive approximation of target permutation amplitudes of this model can be obtained within $\DQC 1$. However, it is tempting to see if there is a similar efficient sampling algorithm within $\DQC 1$ or any class that is believed to be below $\BQP$. Moreover, we do not know a lower-bound for the quantum ball permuting model in this case. Is there an efficient classical algorithm to approximate single amplitudes or to approximately sample from the output distribution? Also, it is left open to see if the additive approximation to the amplitudes of the ball permuting model is a reasonable one. A possible direction is to check if a similar approximation scheme exists for quantum models based on arbitrary group algebras.

\item[] 2. In chater \ref{ch5} it is proved that if the quantum ball permuting model has access to arbitrary initial states, then there is a way to efficiently sample from standard quantum circuits. The construction is based on encoding of qubits using superpositions over permutation states. More precisely $n$ qubits can be encoded using a superposition over permutations of $3n$ labels. However, a drawback of this construction is that it is not scalable, in the sense that the encoding for the tensor product of two qubit quantum states is not a tensor product of two permutation states.

\item[] 2. There is evidence suggesting that if the quantum gates in the ball permuting model satisfy the Yang-Baxter equation, the model generates a sparse subset of unitary operators acting on the Hilbert space of permutations. Moreover, it is argued that the simulation of the model is hard for a classical computer if intermediate measurements are done in the particle color basis. However, it is unknown if there an efficient classical simulation, for the model without intermediate measurements.

\item[] 3. For the randomized ball permuting model it is proved that $\BPL\subseteq \RBALL \subseteq \Almost \L$. Can $\RBALL$ be further pinned down within these classes? Moreover, with two adaptive queries to the randomized ball permuting oracle, $\RBALL$ can simulate $\Almost \L$. Can $\RBALL$ still simulate $\Almost \L$ with only one query?

\item[] 4. For the classical ball permuting model there is a restriction on the probability of swaps according to the Yang-Baxter equation. Is there a $\P$ simulation in this case?

\item[] 5. We do a partial classification on the computational power of this model on arbitrary initial states. The classification is based on the irreducible representations of the symmetric group. We prove that the unitary group generated by this model is as large as possible if the model starts from the initial states corresponding to Young diagrams with two rows or two columns. We conjecture that this result can be extended to arbitrary irreducible representations. However, we leave this for future work. 

\end{itemize}

\chapter{Computability Theory and Complexity Theory}
\label{ch2}

\section{Alphabets}

An alphabet $\Sigma$ is a finite set of symbols. Alphabets can be combined to create strings (sentences). For example, if $\Sigma = \{0,1\}$, then any combination like $01110$ is a string over this alphabet. The set of finite strings (sentences) over an alphabet $\Sigma$ is denoted by $\Sigma^\star := \{w_1 w_2 \ldots w_k : w_j \in \Sigma, k\geq 0\}$. The length of a string $w= w_1 w_2 \ldots w_k$  is simply the number of alphabets which construct the string, denoted by $|w|:= k$. Notice that in the definition of $\Sigma^\star$ the length of a string can be $0$ ($k=0$), which by definition corresponds to an empty string. An alphabet of length zero is called an empty alphabet, and the set of strings over this alphabet by definition only contains an empty string. The set of strings over alphabet $\Sigma$ of length $n$ is denoted by $\Sigma^n$. Clearly $|\Sigma^n|= |\Sigma|^n$. An alphabet of unit size is called a unary alphabet, and size $2$ alphabets are called binary. All nonempty alphabets are equivalent, in the sense that their corresponding sets of strings can count each other.  Also for nonzero $k$, $\Sigma^\star_k$ is isomorphic to natural numbers, $\Sigma_k^\star \cong \mathbb{N}$; where $\Sigma_k$ is any alphabet of size $k$. In order to see this, notice that unary sentences can trivially be placed in one-to-one correspondence with natural numbers. If $\Sigma_k$ is an alphabet of size $k >1$, then a bijection with natural numbers is obtains by just assigning a distinct number in $\{0,1,2,\ldots , k-1\}$ to each member of the alphabet, and viewing the members as base-$k$ representations of natural numbers. From this, for any $k>1$, $\Sigma_k^\star\cong \cup_{n\geq 0}0^n . \mathbb{N} \cong \mathbb{N}\times \mathbb{N} \cong \mathbb{N}$.

\section{Turing Machines}

A language $L$ over the alphabet $\Sigma$ is defined as a subset $L \subseteq \Sigma^\star$ \i.e. a language is a selection of sentences. A selection is in some sense a mechanical process. Given a language and any string, the aim is to distinguish if the string is contained in the language or not. A machine is thereby an abstract object which is responsible for this mechanical selection. A formal definition of such a computing machine was proposed by Alan Turing \cite{Turing}, when he introduced the Turing machine. A Turing machine is in some sense the model of a mathematician who is writing down a mathematical proof on a piece of paper. 

More formally, a Turing machine (TM) has a finite set of states (control) $Q$, a one dimensional unbounded tape, and a head (pointer) on the tape which can read/write from/on the tape. We can assume that the tape has a left-most cell and it is unbounded from the right. Initially, the machine is in a specific initial state $q_0\in Q$, and the tape head is in the left most cell, and the input is written on the tape and the rest of the tape is blank. An alphabet $\Sigma$ is specified for the input. However a different alphabet $\Gamma$ can be used for the tape. Clearly $\Gamma$ contains $\Sigma$. The machine evolves step by step according to a transition function $\delta$. A transition function inputs the current state ($p\in Q$) of the machine, and the content of the current tape cell $x$, and based on these, outputs $q \in Q$ as the next state, $y\in \Gamma$ as the content to be written on the current cell, and a direction $Left$ or $Right$ as the next direction of the head. For example $\delta(p,x)=(q,y,L)$, means that if the TM reads $x$ in state $p$, it will write $y$ instead of $x$ goes to the state $q$ and the head goes to the left on the tape. If at some point the machine enters a special state $q_y$, then it will halt with a yes (accepting) answer. There is another state $q_n$ to which if the machine enters, it will halt with a no (rejecting) answer. 

\begin{definition}
A (deterministic) standard Turing machine is a $7$-tuple $(\Sigma, \Gamma, Q, q_0, q_y, q_n, \delta, D:=\{L,R\})$. Where $\Sigma$ is the input alphabet, and $\Gamma\supseteq \Sigma$ is the tape alphabet. $Q$ is the finite set of states of the machine, $q_0\in Q$ is the unique starting state, $q_y$ and $q_n\in Q$ are the accepting and rejecting halting states, respectively. $\delta: Q\times \Gamma\rightarrow Q \times \Gamma \times D$ is the transition function.
\end{definition}

Any Turing machine corresponds to a languages, and informally an accessible (enumerable)  language is considered to be the one which has a corresponding Turing machine. In computer science this statement is recalled after \textit{Church} and \textit{Turing}:

\begin{center}
\textit{The Church Turing Thesis: "All that is computable in the Physical Universe is Computable by a Turing machine."}
\end{center}

Such a statement is a concern of scientific research in the sense that it is falsifiable, and can be rejected if one comes up with architecture of a physical computing device whose language corresponds to a language that is undecidable by Turing machines. Yet, still there is no counterexample to the \textit{Church-Turing Thesis}.

A language $L \subseteq \Sigma^\star$ over alphabet $\Sigma$ is called \textit{Turing recognizable} if there is a single tape standard Turing machine $M$ such that for every $x\in \Sigma^\star$, $x\in L$ if and only if $M$ accepts $x$. In this case we say that $M$ recognizes $L$. The language $L$ is called \textit{decidable} if there is a standard Turing machine $M$ which recognizes $L$ and moreover, for any $x\in \Sigma^\star$, $M$ halts. A function $f: \Sigma^\star \rightarrow \Sigma^\star$ is called computable if there is a Turing machine $M$ such that $M$ run on $x$ halts with $y$ on its tape iff $f(x)=y$.

We say a Turing machine halts on the input $x\in \Sigma^\star$, if after finite transitions the machine ends up in either state $q_y$ or $q_n$. We say the Turing machine accepts the input if it ends up in $q_y$, otherwise, if it does not halt or if it ends at $q_n$ the input is said to be rejected. We can make several conventions for the transition and the structure of a Turing machine. For example, we can think of a Turing machine with multiple finite tapes. Moreover, in the defined version of the Turing machine, we assumed that at each step the tape head either moves left or right. We can think of a Turing machine wherein the head can either move to left or right or stay put. Thereby, we define a transition \textit{stationary} if the tape head stays put, and define it \textit{moving} if it moves. Also the geometry of the tape itself can differ. Although these models each can give rise to different complexity classes, in terms of computability, we can show that all of these cases are equivalent.

One can immediately prove that there exists at least one language that is not decidable by Turing machines. The space of languages is according to $\{L\subseteq \Sigma^\star: \Sigma \text{  is finite}\}$. For any nonempty alphabet $\Sigma$, $\Sigma^\star$ counts the natural numbers; the following asserts that the set of languages cannot be counted by natural numbers:

\begin{proposition}
The set of subsets of any nonempty set cannot be counted by the original set.
\end{proposition}

\begin{proof}
This true for finite sets, since given any set of $n$ elements the set of subsets has $2^n$ elements. Suppose that $A$ is an infinite set, and suppose as a way of contradiction that $2^A$ (the set of subsets) can be counted by $A$. Then there is a bijection $f: A \rightarrow 2^A$. Given the existence of $f$ consider the subset of $A$, $P=\{x\in A: x \notin f(x)\}$, and the claim is that this subset cannot be counted by $A$, which is a contradiction. Suppose that $P$ has a pre-image, so $\exists a \in A, f(a)=P$. Then $a \in P$ if and only if $a\notin f(a)=P$.
\end{proof}

It suffices to prove that the space of Turing machines is countable by $\mathbb{N}$, and this implies that at least there exists a language that is not captured by Turing machines. For this purpose define the set of finite tuples by $\mathbb{N}^\mathbb{N} := \{(x_1, x_2,\ldots, x_n) : n\in \mathbb{N}, x_j \in \mathbb{N}, j \in [n] \}$. Then:

\begin{proposition}
(G\"odel\cite{godel1931formal}) $\mathbb{N}\cong \mathbb{N}^\mathbb{N}$, with a computable map.
\end{proposition}

\begin{proof}
A bijection $f: \mathbb{N}^\mathbb{N}\rightarrow \mathbb{N}$ is constructed. Given any $X=(x_1, x_2,\ldots, x_n)\in \mathbb{N}^\mathbb{N}$, construct $f(X)= p_1^{x_1} p_2^{x_2}\ldots p_n^{x_n}$. Clearly, non-equal tuples are mapped to non-equal natural numbers. Also the map is invertible since any natural number is uniquely decomposed into a prime factorization. The map is computable, since given any $n$-tuple one finds the first $n$ primes and constructs the image as multiplications.
\end{proof}

\begin{proposition}
The set of Turing machines ($TM$) can be counted by natural numbers.
\end{proposition}

\begin{proof}
It suffices to find a computable one-to-one embedding of $TM$ into $\mathbb{N}^\mathbb{N}$. Each Turing machine $M$ is a finite tuple of symbols. Give each symbol a natural number. These symbols correspond to the input and tape alphabets, name of the machine states, and the left/right (or possibly stay put) symbols. The following is one possible embedding:

\begin{eqnarray*}
M&\mapsto& 2^{|\Sigma|} 3^{|\Gamma|} 5^{|Q|} p^{\Sigma}\ldots p^{\Sigma} p^{\Gamma}\ldots  p^{\Gamma}\\
&& p^{q_0} p^{q_y}p^{q_n} p^D p^{q_1, \Gamma_1 , q_2, \Gamma_2, D}\ldots p^{q_1, \Gamma_1 , q_2, \Gamma_2, D}.
\end{eqnarray*}

we show the sequence of primes multiplied together with some encoding of the symbols for a set $A$ as multiplication powers $p^A \ldots p^A$ of consecutive primes. In order to avoid adding a new symbol for a delimiter, the sizes of $\Sigma, \Gamma$ and $Q$ are specified with the first three primes.
\end{proof}

\noindent As a corollary, the following statements should be true:

\begin{itemize}
\item[]- There is a language that is not recognizable by any Turing machine.
\item[]- There is a real number that is not computable.
\end{itemize}


\section{The Complexity Theory of Decidable Languages}

By definition, for any decidable language there exists a Turing machine which always halts in certain amount of time and space. One way of classifying these languages is based on the minimum amount of time (space) of a Turing machine that decides the language. In order to classify the languages, then we need a partial (or total order). For this purpose we use the inclusion $\subseteq$ as a partial order.

\begin{definition}
For any function $f: \mathbb{N}\rightarrow \mathbb{N}$, a language $L$ is in time $\TIME(f)$ ($\SPACE(f)$) if there is a standard Turing machine $M$ which on any input $x$  it halts using $O(f(|x|))$ time steps (tape cells of space) and $x \in L$ if and only if $M$ accepts $L$. We therefore define $\P:= \TIME (n^{O(1)})$ and $\PSPACE := \SPACE(n^{O(1)})$.
\end{definition}

\noindent In the above definition the so-called big-O notation is used: for two functions $f$ and $g : \mathbb{N}\rightarrow \mathbb{N}$, $f = O(g)$, means that there exists $n_0 \in \mathbb{N}$ and a constant $c>0$, such that for all $n \geq n_0$, $f(n)\leq g(n)$.

Next, we mention the concept of nondeterminism: 

\begin{definition}
For any $f: \mathbb{N}\rightarrow \mathbb{N}$ the language $L \in \NTIME(f)$ ($\NSPACE(f)$) if there is a Turing machine $M$ with the property that $x\in L$ if and only if there exists a string $y$ such that $M(x,y)=1$ (accepts), and that $M(x,y)$ halts in $O(f(|x|))$ for any string $y$. Define $\NP:=\NTIME (n^{O(1)})$ $\PSPACE:=\NSPACE(n^{O(1)})$, and $\NL:=\NSPACE(O(\log n))$.
\end{definition}

One can think of a nondeterministic version of a Turing machine in which the machine starts out of a unique initial state $q_0$ on some input $x\in \Sigma^\star$, and at each step the computation can branch according to a nondeterministic transition function. In other words, such a nondeterministic Turing machine can guess a transition, and the computation accepts, if among the guessed computations, at least one of them leads to an accepting state, and the computations rejects otherwise. The running time of a nondeterministic Turing machine is the greatest running time among the guessed computations (including rejecting paths).  Therefore, $\NP$ can be alternatively defined by the set of languages for which there is a polynomial-time nondeterministic Turing machine which accepts its inputs whenever they are contained in the language.

In computer science, sometimes we are interested in the class of languages that can be decided efficiently, if a certain language can be decided immediately by a black box. Such a black box access to a language can be formalized by an oracle. An oracle is the interface of  a language with a machine. More precisely, an oracle for language $A$, is a single tape machine which takes a string $x \in \{0,1\}^\star$ as its input, and the in one step of computation, clears its tape and writes a $1$ if $x\in A$ and otherwise writes a $0$. Oracles can also compute arbitrary functions $: \{0,1\}^\star \rightarrow \{0,1\}^\star$ in one step of computation. Given a class of computing machines $M$ which can make queries to an oracle $A$, define $M^A$ to be the class of languages that are decidable by these machines with query access to $A$. 

If we define a uniform probability distribution over the set of oracles then:

\begin{definition}
Let $M$ be a class of computing machines. $\Almost M$ is defined as the class of languages that are decidable with with bounded probability of error by a machine in the class of machines $M$ with access to a random oracle.
\end{definition}

By bounded probability of error it is meant that there are constants $1\geq c, c'>0$, such that if an input $x$ is in the language, then with probability greater than $c$ over the set of oracles, the machine accepts $x$, and if $x$ is not in the language, then with probability greater than $c'$, the machine rejects $x$.

Next, we mention reduction as a crucial element of theory of computing:

\begin{definition}
Given a class of machines $Q$, and two languages $L_1, L_2$, we say that $L_1 \leq^{m}_Q L_2$, or $L_1$ is mapping reducible to $L_2$ with $M$, if there is a function computable in $M$, such that $x \in L_1$ if and only if $f(x) \in L_2$. Also, alternatively, say $L_1$ is oracle reducible to $L_2$, $L_1 \leq^o_M {L_2}$, if $L_1 \subseteq M^{L_2}$.
\end{definition}

A reduction is a partial order on the set of languages, and when a language $A$ is reducible to another language $B$, intuitively, this means that $B$ is at least as hard as $A$. If $A$ is mapping reducible to $B$, then also $A$ is oracle reducible to $B$, but the converse is not necessarily true.

\begin{definition}
A language is called $\NP$ hard if all languages in $\NP$ are $\P$ reducible to it. A language is called $\NP$ complete if it is $\NP$ hard and is also contained in $\NP$.
\end{definition}

\begin{theorem}
(Cook-Levin \cite{cook1971complexity}) $\NP$ has a complete language.
\end{theorem}

\noindent A specific example of such a complete language is the boolean satisfiability problem, $\SAT$: given a boolean formula, decide if there is a satisfying assignment. The following is then immediate.

\begin{lemma}
A language is $\NP$ complete if there is a polynomial time reduction from $\SAT$ or any other $\NP$ complete language to it.  

An $\NP$ language is in $\P$ if there is a reduction from the language to a language in $\P$.

$\P=\NP$ if and only if $\SAT\in \P$ (also true for any other $\NP$ complete language)
\end{lemma}

Next a formal definition of counting classes is given:

\begin{definition}
$\#\P$ is the class of functions $f: \Sigma^\star\rightarrow \mathbb{N}$ which count the number of accepting branches of an $\NP$ machine. In other words, given a nondeterministic Turing machine, $M$, $f(x)$ is the number of accepting branches of $M$ when run on $x$. 

(Probabilistic polynomial time) $\PP$ is the class of problems $L$ for which there exists an $\NP$ machine $M$ such that $x\in L$ iff most of the branches of $M$ run on $x$ accept.
\end{definition}

Other than deterministic and nondeterministic models, an alternative model is randomized computation. In such scheme of computing a Turing machine has access to an unbounded read-once tape consisting of independent true random bits. The transition function can thereby depend on a random bit. Based on such machines we can define new complexity classes. 

\begin{definition}
A probabilistic standard Turing machine is defined similar to the deterministic version, with an extra unbounded read-once tape of random bits, as an $8$-tuple $(\Sigma, \Gamma, R, Q, q_0, q_y, q_n, \delta, D:=\{L,R\})$. Here $R$ is a finite alphabet of random bits, and each element of the alphabet is repeated with equal frequency (probability). The transition function $\delta: Q\times \Gamma \times R \rightarrow Q\times \Gamma \times D$.
\end{definition}

Therefore, the following two complexity classes are naturally defined as:

\begin{itemize}

\item (bounded error probabilistic polynomial time) $\BPP$ is the class of languages $L$ for which there is a probabilistic polynomial time Turing machine $M$ such that if $x\in L$, $\operatorname*{Pr}[M(x)=1]\geq 2/3$ and otherwise $\operatorname*{Pr} [M(x)=0]\geq 2/3$.

\item (probabilistic polynomial time) $\PP$ is the class of languages $L$ for which there is a probabilistic polynomial time Turing machine $M$ such that if $x\in L$ then $\operatorname*{Pr}[x=1]>1/2$ and otherwise $\operatorname*{Pr}[x=0]>1/2$.

\end{itemize}

The class $\PP$ is related to the counting classes by the following theorem:

\begin{theorem}
$\P^{\#\P}= \P^{\PP}$ \cite{arora2009computational}.
\end{theorem}

\section{The Polynomial Hierarchy}

The $\NP$ language can be equivalently formulated as the set of languages $L \subseteq \{0,1\}^\star$, for which there is a polynomial time Turing machine $M(\cdot, \cdot)$ and a polynomial $p:\mathbb{N} \rightarrow \mathbb{N}$, such that $x \in L$ if and only if there exists $y \in \{0,1\}^{p(|x|)}$ such that $M(x,y)$ accepts. We can just write:

$$
x \in L \leftrightarrow \exists y\hspace{2mm} M(x,y)=1
$$

The complement of $\NP$ is called $co\NP$ and is defined as the set of languages $L \subseteq \{0,1\}^\star$ for which there is a polynomial time Turing machine $M$ such that:

$$
x \in L \leftrightarrow \forall y \hspace{2mm}M(x,y)=1
$$

\noindent The relationship between $\NP$ and $co\NP$ is unknown, but we believe that they are in comparable as set of languages,\i.e. none of them is properly contained in the other. Define the notation $\Sigma^0_P = \Pi^0_P = \P$, and $\Sigma^1_\P =\NP$ and $\Pi^1_\P = co\NP$, then we can inductively extend the definitions to a hierarchy of complexity classes. Define $\Sigma^j_P$ to be the class of languages $L \subseteq \{0,1\}^\star$, for which there is a polynomial time Turing machine $M$ such that:

$$
x \in L \leftrightarrow \exists y_1 \hspace{2mm} \forall y_2 \hspace{2mm}\exists y_3 \hspace{1mm}\ldots\hspace{1mm} Q_i y_i\hspace{2mm} M(x,y_1, y_2, \ldots, y_i)=1.
$$

\noindent Here $Q_i$ is either a $\exists$ or $\forall$ quantifier depending on the parity of $i$. The complexity class $\Pi^i_\P$ is similarly defined as the class of languages for which there exists a polynomial time Turing machine $M$ such that:

$$
x \in L \leftrightarrow \forall y_1\hspace{2mm} \exists y_2 \hspace{2mm}\forall y_3 \hspace{1mm}\ldots \hspace{1mm}Q_i y_i \hspace{2mm} M(x,y_1, y_2, \ldots, y_i)=1.
$$

\noindent The complexity class polynomial hierarchy is then defined as the union $\PH:= \cup_{i\geq 0} \Sigma^i_\P$.

The hierarchy is conjectured to be infinite.  The relationship between the $\Pi^i$ and $\Sigma^i$ and also the different levels within $\PH$ is unknown. However we know that if $\Sigma^i_{\P}=\Pi^i{\P}$ or $\Sigma^i_{\P}=\Sigma^{i+1}_{\P}$ for $i>0$, then $\PH$ collapses to the $i$'th level, and as a result, the hierarchy will consist of finitely many levels \cite{arora2009computational}. Another interesting direction is the relationship between $\BPP$ and $\PH$. The relationship between $\BPP$ and $\NP$ is unknown, however according to Sipser et al. $\BPP \in \Sigma^2_\P$.

Infinite $\PH$ conjecture is sometimes used to make inference about the containments of complexity classes. For example, consider the following: the class $\BPP^{\NP}$ is known with approximate counting; in comparison, $\P^{\#\P}$ corresponds to exact counting. According to a theorem by Toda \cite{toda1991pp}, $\PH$ is contained in $\P^{\#\P}=\P^{\PP}$. $\BPP^{\NP}$ is contained in the third level of $\PH$. However, because of Toda's theorem $\P^{\#\P}\subseteq \BPP^{\NP}$ which implies $\PH \subseteq \P^{\#\P}\subseteq \Sigma^3_{\P}$, and then a collapse of $\PH$ to the third level. Therefore, we say that there are counting problems that are hard to even approximate within a multiplicative constant, unless $\PH$ collapses to the third level \cite{arora2009computational}.

\section{Reversible Turing Machines}

A Turing machine is called reversible if nodes of its infinite configuration space as a graph have in-degree and out-degree at most $1$. The following is crucial for the result of section \ref{ballpermutingoracles}.

\begin{theorem}
\begin{itemize}
\item[]
\item[] (Lange-McKenzie-Tapp \cite{lange2000reversible}) Any function that is computable in space (multi-head) $S(n)$, can be computed in the same space reversibly, (with possible exponential overhead in time).

\item[] Any language in $\L$ can be recognized by a reversible $\LOGSPACE$ ($\Rev\L$).
\end{itemize}
\label{Mck}
\end{theorem}

\begin{proof}
(Sketch) consider the directed configuration space for the computation of any function in space $S$, then the reversible algorithm is to take an Euler tour over an undirected graph constructed by doubling each edge of the configuration space. Thereby, the reversible machine is able to first compute the value of the function and then enumerate the possible pre-images of the function including the original input. In order to see $\Rev\L=\L$, just consider the action of $\LOGSPACE$ computation as an input saving computing mode, where there are two separate tapes one holding the original input and the other computes the function. From the first part of the theorem, any such computation can be captured by a symmetric computation in the same space, and sine the definition of $\L$ does not impose any constraint on time, the same class is equal to the symmetric (reversible) version.
\end{proof}

\section{Circuits}

Consider the set of functions $F$ of the form $\{0,1\}^\star\rightarrow \{0,1\}$, also known by Boolean functions. Given any language $L$, one can construct a function with $f(x)=1$ if and only if $x\in L$. In other words each such function represents a language. Most of the languages are undecidable, thereby most of these functions are not computable. We can think of a class of functions $F_n$ as the set of functions of the form $\{0,1\}^n\rightarrow \{0,1\}$. Represent each of these strings with an integer between $1$ and $2^n$. Given this ordering any function $F_n$ can be specified with a string of $2^n$ bits, and thereby $|F_n|=2^{2^n}$. Such an encoding of a function is formalized by a truth table: a truth table $tt_f$ of a function $f: \{0,1\}^n\rightarrow \{0,1\}$ is a subset of $\{0,1\}^n\times \{0,1\}$ for which $(x,s)\in tt_f$ if and only if $f(x)=s$.

Boolean can be described by Boolean variables ranging in $\{0,1\}$, and binary operations (AND) $.$ and  (OR) $+ : \Sigma \times \Sigma  \rightarrow \Sigma$ between them, and a single-bit operation called negation (NOT) $': \Sigma \rightarrow \Sigma$. Given $x,y \in \Sigma$, $x.y=1$ if $x=1$ and $y=1$, otherwise $x . y = 0$, and $x+y =0$ only if $x=0$ and $y=0$ and otherwise $x+y=1$. And $x'=0$ if $x=1$ and otherwise $x=1$. We can alternatively use the symbols $\wedge$, $\vee$ and $\neg$ for AND, OR, and NOT operations, respectively.

We can think of these operations as gates and variables as input (wires) to the gates, and the collection of these forms a model of computation called circuits. $AND, OR$ and $NOT$ gate-set is an alternative set of operations to capture Boolean functions. In general, circuits are compositions of local gates, where a local gate represents a Boolean function from constant number of inputs to constant number of outputs.

Circuits can be related to Turing machines by:

\begin{definition}
A family of circuits $\{C_{m,n}\}$, each with $m$ inputs and $n$ outputs, is called uniform if there is a Turing machine which on input $m,n$ outputs the description of $C_{m,n}$. The family is otherwise called nonuniform.
\end{definition}

Also the following definition is used in section \ref{ballpermutingoracles}:

\begin{definition}
$\AC^i$ is the class of decision problems that are solvable by a (possibly) nonuniform family of circuits composed of unbounded fanin $AND, OR$ and $NOT$ gates, and have depth growing like $O(\log^i n)$.
\end{definition}

An unbounded AND (OR) gate is a gate with unbounded input wires and unbounded output wires such that all output wires are a copy of the other, and their value is $1 (0)$ if and only if all the input wires are $1 (0)$, and otherwise the output wires take $0 (1)$. 

Any boolean function $f: \{0,1\}^n \rightarrow \{0,1\}$, can be constructed by a sequence of AND, OR, NOT and COPY. We then call the collection these operations a universal gate set. Therefore, any gate set which can simulate these operations is also universal for Boolean computing. Among these universal gate sets is the gate set consisting of NAND operation only. A NAND operation is the composition of NOT and AND from left to right. We are also interested in universal gate sets which are reversible. That is the gate sets that can generate subsets of invertible functions $f: \{0,1\}^\star \rightarrow \{0,1\}^\star$ only. A necessary condition is that each element of the gate set has equal number of inputs and outputs. Examples of reversible gates are controlled not $CNOT: \{0,1\}^2\rightarrow \{0,1\}^2$ which maps $(x, y) \mapsto (x, x\oplus y)$. That is $C$ flips the second bit if the first bit is a $1$. Here $\oplus$ is the addition of bits mod $2$. Notice that $CNOT$ is its own inverse. Circuits based on $CNOT$ can generate linear functions only and thereby, CNOT is not universal in this sense. However, if a gate operates as NOT controlled by two input bits, then we can come up with gates that are both reversible and universal. More precisely, let $T: \{0,1\}^3\rightarrow \{0,1\}^3$, be a Boolean gate with the map $(x, y, z) \mapsto (x, y, x.y \oplus z)$. Then $T$ is also its own inverse, and one can confirm that composition of $T$ gates can simulate a NAND gate. Notice that we need extra input and output bits to mediate the simulation. Such extra bits are called ancilla. The $T$ gate is also known as the Toffoli gate. As another example let $F: \{0,1\}^3\rightarrow \{0,1\}^3$, be a Boolean gate with the maps $(0, x, y ) \mapsto (0, x,y)$ and $(1, x, y) \mapsto (1, y, x)$. $F$ is also its own inverse and moreover it can be proved that $F$ is also universal. $F$ is also known as the Fredkin gate.

\chapter{Quantum Theory and Quantum Complexity Theory}
\label{ch3}

In this chapter, we go over some background in quantum mechanics, quantum computing, and quantum complexity theory complexity theory. After a short introduction to quantum mechanics, we discuss the integrable models in $1+1$ dimensions. Quantum computing and quantum complexity theory are discussed later in the second half of the chapter.

\section{Quantum Mechanics}

There are various interpretations and formulations of quantum mechanics. The following views quantum mechanics as a generalization of classical probability theory and classical mechanics. In that, a system is described as a quantum state, which is a complex vector in a vector space. These states encode the probability distribution over the possible outcomes of observables. Observables are Hermitian operators on the vector space. Like in classical probability theory, the states of the vector space should be normalized with respect to some norm, and the set of operators that map normalized states to normalized state are the legitimate evolution operators.

A vector space is called a Hilbert space $\mathcal{H}$, if it is complete and has an inner-product. A Hilbert space can have finite or infinite dimension. An inner-product is a function $\langle \cdot, \cdot\rangle: \mathcal{H} \times \mathcal{H} \rightarrow \C$, with conjugate symmetry, \i.e $\langle \phi_1| \phi_2\rangle^\star= \langle \phi_2| \phi_1\rangle, \forall \phi_1, \phi_2 \in \mathcal{H}$, positive definiteness, that is for all $\phi \in \mathcal{H}$, $\langle\phi|\phi \rangle\geq 0$, with equality if and only if $\phi =0$, and bilinearity $\langle \phi| a \phi_1+ b \phi_2 \rangle=a \langle \phi|\phi_1\rangle + b \langle \phi| \phi_2 \rangle$. Here $\bullet^\star$ is the complex conjugation of the $\C$-numbers. Complete means that any Cauchy sequence is convergent with respect to the norm inherent from inner product. We represent vectors $\phi \in \mathcal{H}$ with a ket notation $|\phi\rangle$. If $\mathcal{H}$ is finite dimensional with dimension $n$, then $\mathcal{H}\cong \C^n$, as a vector space. Otherwise, we denote an infinite dimensional Hilbert space with $\C^\infty$. We call $\{|e_j\rangle: j \in [n]\}$ an orthonormal basis of $\C^n$, if $\langle e_i|e_j\rangle=\delta_{ij}$. $\delta_{ij}$ is the Kronecker, which takes the value $1$ if $i=j$ and otherwise $0$. Let $|\phi\rangle = \sum_{j \in [n]} \phi_j |e_j\rangle$, and $|\psi\rangle = \sum_{j \in [n]} \psi_j |e_j\rangle$, be vectors in $\C^n$, we use the inner product:

$$
\langle \phi|\psi\rangle = \sum_{j\in [n]} \phi^\star_j \psi_j.
$$

\noindent Here $\langle \phi|:= \sum_{j\in [n]} \phi^\star_j \langle j |$, is the bra notation for the dual vectors. Where, $\langle e_j |$ act as $\langle e_j| e_i \rangle= \delta_{ij}$. More precisely, we call $\mathcal{H}^\star$ the dual of the Hilbert space $\mathcal{H}$, as the set of linear functions $: \mathcal{H}\rightarrow \C$. $\mathcal{H}^\star$ is also a vector space isomorphic to $\mathcal{H}$, and thereby has the same dimension as $\mathcal{H}$, and is spanned by $\langle e_j |$. We will not delve into the foundations of infinite dimensional Hilbert spaces. In simple words, such a Hilbert space corresponds to the space of square integrable functions $\phi: \mathbb{R}^m \rightarrow \C$, and we call this function square integrable if:

$$
\int_{z\in \mathbb{R}^n} d^n z |\phi(z)|^2
$$

\noindent exists, and the inner product is defined as:

$$
(\phi,\psi)=\int_{z\in \mathbb{R}^n} d^n z \phi^\star (z) \psi (z)
$$

\noindent A vector $|\phi\rangle$ in this Hilbert space is decomposed as $\int_{z\in \mathbb{R}^n} d^n z \phi(z)|z\rangle$, and a normalized state is the one which $\int_{z\in \mathbb{R}^n} d^n z |\phi(z)|^2=1$. 

Consider the Hilbert space $\C^n$. A quantum state $|\psi\rangle \in \C^n$ is therefore a normalized vector. Any orthonormal basis $|f_j\rangle, j \in [n]$ corresponds to a set of non-intersecting events. The amplitude of measuring the state $|\psi\rangle$ in the state $|f_j\rangle$ is the complex number $\langle f_j | \psi \rangle$, and is related to a probability with $P_j=|\langle f_j | \psi\rangle|^2$, where $\sum_{j\in [n]} P_j=1$.

An operator on the Hilbert space is any function $: \mathcal{H}\rightarrow \mathcal{H}$. Observables are therefore the linear Hermitian operators. Given a quantum state $|\psi\rangle$, and an observable $O$ with spectrum $\{a, |a\rangle\}$, measuring $|\psi\rangle$ with observable $O$ corresponds to observing the real value $a$ with probability $|\langle a | \psi \rangle|^2$, therefore the expected value of $O$ is $\langle \psi | O | \psi \rangle$. A Hamiltonian is the observable having the allowed energies of the system as its eigenvalues. A Hamiltonian encodes the dynamics of  a system.

A legitimate time evolution of a quantum system corresponds to an operator which maps the normalized states to normalized states. A linear operator $U$ with this property is called a unitary operator, and satisfies $U^\dagger U= I$. Each physical system can be described by a Hamiltonian $H$. The Hamiltonian is responsible for the unitary evolution in time. Such an evolution is described by a Schr\"odinger equation:

$$
i\dfrac{\partial}{\partial t} |\psi(t)\rangle = H|\psi(t)\rangle.
$$

\noindent Where $|\psi(t)\rangle$ is the state of the system at time $t$. If $H$ is time-independent, then the unitary evolution is $|\psi(t)\rangle=\exp(-i H t) |\psi(0)\rangle$.

\section{Some Quantum Models in $1+1$ Dimensions}
\label{models}

In this section, we review scattering the structure of amplitudes in some quantum models of two dimensional space-time. As it turns out, both relativistic and non-relativistic regimes pose similar structures. For the analysis of computational complexity it is sufficient to focus on one of these, and the same results immediately apply to the others. More specifically, these are integrable quantum models of $1+1$ dimensions \cite{faddeev1981two, ghoshal1994boundary, staudacher2012review}. Integrability is translated as a model which has an exact solution, that is, perturbation terms in the expression of scattering amplitudes amount to an expressible shape. For a brief review of scattering amplitudes see appendix \ref{scattering}. In order to understand this point, view each perturbation term as a piece among the pieces of a broken vase. While these pieces look unstructured and unrelated, in an integrable world, they can be glued together and integrated in a way that the whole thing amounts to a vase. However, the solution has a combinatorial structure in it, and the goal is to find out hardness for computation of these amplitudes.

The situation is that in far past, a number of free particles are initialized on a line, moving towards each other, and in far future, an experimenter measures the asymptotic wave-function that is resulted from scattering. In the following, we first review the factorized scattering matrix of quantum field theory. The structure of the interactions is described, and it is explained how entries of the scattering matrix are obtained. Next, we review the repulsive delta interactions model, as a non-relativistic model of scattering of free particles. In chapter \ref{ch5}, without loss of generality, we will focus on the second model throughout.

Zamolodchikov and Zamolodchikov \cite{zamolodchikov1978relativistic, zamolodchikov1979factorized} studied models of two dimensional quantum field theory that give rise to factorized scattering matrices. A scattering matrix is called factorized, if it is decomposable into the product of $2\rightarrow 2$ scattering matrices. They found that the factorization property is related to an infinite family of conservation rules for these theories. More specifically, suppose that the initial quantum state of $n$ particles with momenta  $p_1, p_2, \ldots, p_n$, and masses $m_1, m_2, \ldots m_n$ is related to an output state of $l$ particles with momenta $p'_1, p'_2, \ldots, p'_l$, and masses $m'_1, m'_2, \ldots, m'_l$, then an example of these conservation rules is according to:

$$
\sum_{j\in [n]}p^{2 N +1}_j=\sum_{j\in [l]} p'^{2 N +1}_j\hspace{1cm} N= 0, 1, 2, 3,\ldots
$$

\noindent and,

$$
\sum_{j\in [n]} p^{2 N}_j \sqrt{p^{2}_j + m^2_j}=\sum_{j\in [n]}  p'^{2 N}_j \sqrt{p'^{2}_j + m'^2_j}\hspace{1cm} N= 0,1, 2, 3, \ldots
$$

\noindent These equations directly impose selection rules on the scattering process. According to these selection rules, $n=l$, $\{m_1, m_2,\ldots, m_n\}=\{m_1, m_2,\ldots, m_l \}$, and that the particles of different mass do not interact, and the output momenta among the particles of the same mass are permutations of the input momenta. In this case, the conservation rules put drastic constraints on the structure of the scattering amplitudes, and this directly imply factorization of the scattering matrix, and thereby integrability of the scattering matrix. Indeed, particles do not actually interact, and instead they only exchange their internal degrees of freedom and their momenta. Thereby, the process resembles pairwise elastic collisions quantum hard balls. Candidates for factorized scattering matrix include the quantum sine-Gordon \cite{schroer1976towards}, the massive Thirring model, and quantum chiral field \cite{zamolodchikov1979factorized}. All of these models pose an $O(n)$ isotopic symmetry.

In the following, we sketch the general structure of a factorized relativistic scattering matrix. Suppose that $n$ particles of the same mass $m$ are placed on a line, where each one is initialized with a two momentum $(p^0, p^1)=: m (\cosh \theta, \sinh \theta)$, and an internal degree of freedom with a label in $[n]$.Where $\theta$is a real parameter, called rapidity, which is related to $p$ as $p=(p^0,p^1)= m (\cosh \theta, \sinh \theta)$. We can mark the entries of the scattering matrix, $S$, by $n$ discrete labels $i_1, i_2, i_3, \ldots, i_n$ each ranging in $[n]$, and rapidities $\theta_1, \theta_2, \ldots, \theta_n$. Let $\pi$ be the permutation for which $\theta_{\pi(1)}\geq \theta_{\pi(2)}\geq \ldots \geq \theta_{\pi(n)}$. Given the conservation rules, the entries corresponding to $I:=i_1, i_2, i_3, \ldots, i_n\rightarrow J:= j_1, j_2, j_3, \ldots, j_n$ and $\tilde{\theta}:= \theta_1, \theta_2, \theta_3, \ldots, \theta_n\rightarrow \tilde{\theta'}:= \theta'_1, \theta'_2, \theta'_3, \ldots, \theta'_n$ of $S$ has the following form:

$$
S^{\tilde{\theta},\tilde{\theta'}}_{I,J}=\delta(\tilde{\theta'}-\pi(\tilde{\theta})) \mathcal{A}_{I,J}.
$$

\noindent Where $\pi (\theta_1, \theta_2, \ldots, \theta_n)=(\theta_{\pi(1)}, \theta_{\pi(2)}, \ldots, \theta_{\pi(n)})$. That is, the only nonzero entries are the ones where the rapidities are reordered in a non-ascending order.

The goal is to compute the amplitudes $\mathcal{A}_{I,J}$. For this purpose, Zamalodchikov et. al. invented an algebra to compute the amplitudes of a factorized model. This is now known as the Zamalodchikov algebra. The algebra is generated by non-commutative symbols that encode initial rapidities and labels of the particles before scattering. Suppose that $n$ particles are initialized with rapidities $\theta_1, \theta_2, \ldots, \theta_n$, and labels $i_1, i_2, \ldots, i_n \in [n]$.

$\mathcal{A}_{i_1 \ldots, i_n\rightarrow j_1 \ldots, j_n}$ can be computed by the following: define a symbol $A_{i_j} (\theta_j)$ for each particle $j$. Here, the labels $i_j$ can be possibly repeated. The multiplication rules between these symbols are according to:

$$
A_i(\theta) A_j (\phi)=\alpha(\theta, \phi)A_i(\phi) A_j (\theta)+\beta (\theta, \phi)A_j(\phi) A_i (\theta)
$$

\noindent for $i\neq j$. This case corresponds to the scatterings $i+j \rightarrow i+j$ and $i+j \rightarrow j+i$, where intuitively either the particles bounce off or pass through each other. Here $\alpha$ and $\beta$ are complex numbers that depend on the rapidities, and the number of particles. For $i=j$ the replacement rule is an annihilation-creation type $i+i \rightarrow j+j$:

$$
A_i(\theta) A_i (\phi)=e^{i\phi(\theta, \phi)}\sum_{j\in [n]}A_j(\phi) A_j (\theta),
$$

\noindent and the overall process obtains a global phase. Now in order to compute the amplitude $\mathcal{A}_{i_1 \ldots, i_n\rightarrow j_1 \ldots, j_n}$, multiply the symbols according to:

$$
A_{i_1} (\theta_1) A_{i_2} (\theta_2)\ldots A_{i_n} (\theta_n),
$$

\noindent and read the coefficient of $A_{j_1} (\theta_{\pi(1)}) A_{j_2} (\theta_{\pi(2)})\ldots A_{j_n} (\theta_{\pi(n)})$ as the desired amplitude. In order for the calculation to make senes the final amplitude should be independent of the order the symbols are multiplied. This can be also translated as the associativity of the algebra:

$$
\ldots A_{i_1} (\theta_1)\Big (A_{i_2} (\theta_2) A_{i_3} (\theta_3)\Big )\ldots=\ldots\Big (A_{i_1} (\theta_1) A_{i_2} (\theta_2)\Big ) A_{i_3} (\theta_3)\ldots
$$

\noindent This is also known as the factorization condition, also known as the Yang-Baxter equation \cite{yang1967some, baxter1972partition}. We are going to describe this point in detail in the context of the non-relativistic repulsive model, and later in section ~\ref{YangBaxter} of chapter \ref{ch5}. See Figure \ref{fig21}: the total scattering process is decomposed into pairwise four-particle interactions, as demonstrated with the red circles. The reconfiguration of the rapidities is represented by lines flowing upwards. In this case, the scattering matrix is the product of smaller matrices, each corresponding to one of the red circles.

The scattering matrix $S$ has a separate block corresponding to the matrix entries indexed by all distinct labels as permutations of $1,2,3,\ldots, n$. In this case, scattering processes like $1+1\rightarrow 2+2$ do not occur. This separate block has dimension $n!$. We are specifically interested in computational complexity of finding matrix elements of this block.

\begin{figure}[tp]
\begin{center}
\includegraphics[height=4.5in]{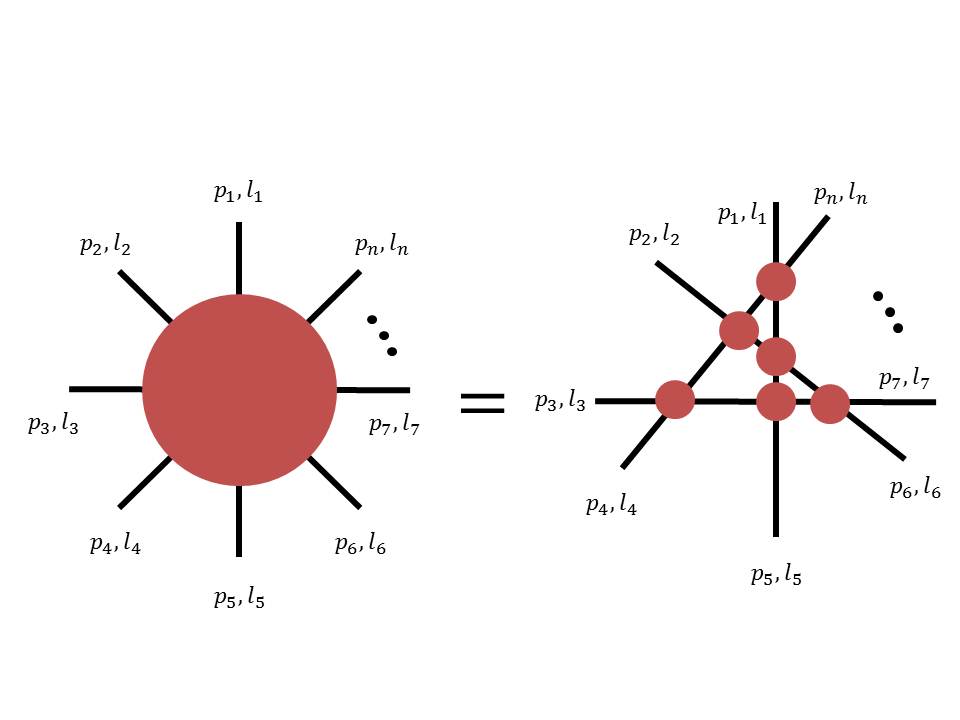}
\caption[Factorization of the scattering matrix]{Factorization of the S-Matrix into two-particle interactions. Any such nonzero amplitude diagram has even number of legs, and for $2 n$ legs, the factorization is according to intersections of $n$ straight lines.
}
\label{fig21}
\end{center}
\end{figure}

\begin{figure}[tp]
\begin{center}
\includegraphics[height=4.0in]{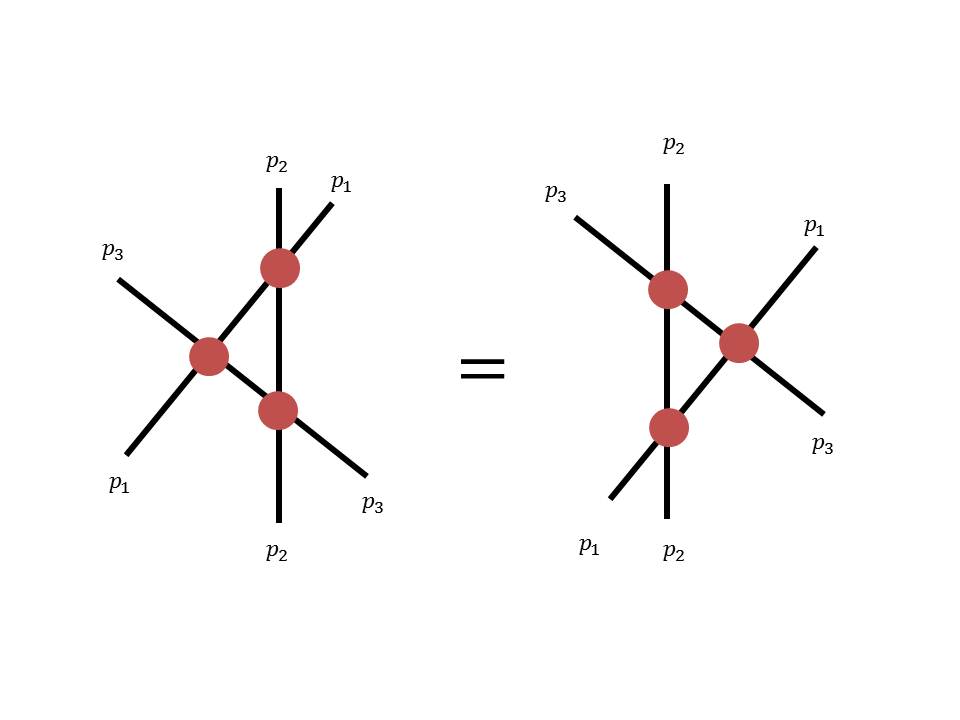}
\caption
[Line representation of the Yang-Baxter equation]
{Three particles satisfying the Yang-Baxter equation.}
\label{fig22}
\end{center}
\end{figure}

\begin{figure}[tp]
\centering
\includegraphics[height=3.5in]{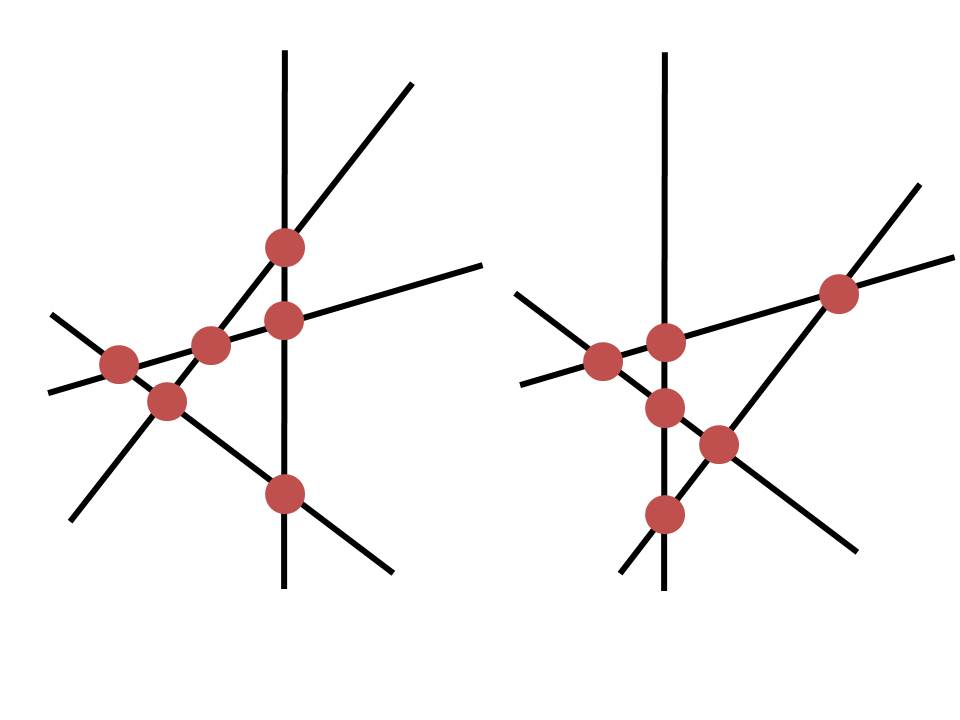}
\caption[Examples of equivalent factorized diagrams]{Two equivalent diagrams with amplitudes satisfying the Yang-Baxter equation. The right hand diagram can be obtained from the left hand diagram by moving some of the straigth lines parallel to themselves. 
}
\label{fig23}
\end{figure}

The non-relativistic analogue of the factorized scattering is given by the repulsive delta interactions model ~\cite{yang1967some} of quantum mechanics. In this model, also elastic hard balls with known momenta are scattered from each other. The set of conserved rules are closely related to the relativistic models. If we denote the initial momenta of the balls with $p_1, p_2, \ldots, p_n$, then the conserved quantities are $\sum_j p^{2 k+1}_j$, and $\sum_j p_j^2/m_j$, for $k\geq 0$. Where $m_j$ are the mass of the balls. Again, the selection rules assert that balls with different mass do not interact with each other, and the final momenta among the balls of same mass are permutation of the initial momenta.

In the repulsive delta interactions model, $n$ asymptotically free balls interact and scatter on a line. Asymptotic freedom means that except for a trivially small spatial range of interactions between each two balls, they move freely and do not interact until reaching to the short range of contact. Denote the position of these balls by $x_1, x_2, \ldots, x_n$ and the range of interaction as $r_0$, then for the asymptotic free regime, we assume $|x_i-x_j|\gg r_0$. The interaction consists of at most $\dfrac{n(n-1)}{2}$ terms, one for each pair of balls. For each pair of balls, the interaction is modeled by the delta function of the relative distance between them. If no balls are in contact, then the action of such Hamiltonian is just a free Hamiltonian, and each contact is penalized by a delta function. The functional form of the Schr\"odinger's equation is written as:

$$
i\hbar\dfrac{\partial}{\partial t} \psi(x_1, x_2, \ldots, x_n, t) = \Big [\sum_{j=1}^n-\dfrac{\partial ^2}{\partial x_j^2}+2 c\sum_{i<j}\delta(x_i-x_j)\Big ]\psi(x_1, x_2, \ldots, x_n, t)
$$

\noindent Here $c>0$ is the strength of the interactions. As the species of unequal mass do not interact, only balls of same mass are considered. The Hilbert space is indeed $(\mathbb{C}^\infty)^{\otimes n} \times \mathbb{R}$. Using the Bethe ansatz ~\cite{staudacher2012review} for spin chain models, a solution for the eigenfunction with the following form is considered:

$$
\psi(x_1,\ldots,x_n)=\sum_{\tau,\pi \in S_n} \mathcal{A}^\tau_\pi \theta_\tau (x_1, \ldots, x_n) \exp{[i(x_{\tau_1} p_{\pi_1}+\ldots+x_{\tau_n} p_{\pi_n})]}
$$

\noindent $\theta_\pi : \mathbb{R}^n \rightarrow \{0,1\}$ is an indicator function which is set to $1$ whenever its input $x$ satisfies $x_{\pi(1)}<x_{\pi(2)}<\ldots <x_{\pi(n)}$, and otherwise zero. $p_j$ for $j\in [n]$ are constant parameters, and can be viewed as the momenta. The proposed solution must be a continuous function of the positions and also one can impose a boundary condition for $x_{\pi(t)}=x_{\pi(t+1)}$ for $t\in [n]$ on the derivative of the wave-function. Applying these boundary conditions, one can get linear relations between the amplitudes:

$$
\mathcal{A}^{t \circ \tau}_{\pi}=\dfrac{-ic \mathcal{A}^\tau_\pi + V_{\tau, t} \mathcal{A}^\tau_{t \circ \pi}}{ic + V_{\tau, t}}
$$

\noindent Here $t \circ \pi$ is a new permutation resulted from the swapping of the $t$ and $t+1$'th labels in the permutation $\pi$. $V_{\tau, t}=p_{\tau(t)}-p_{\tau(t+1)}$. The above linear map has a simple interpretation: two balls with relative velocity $V$ collide with each other with amplitude $\dfrac{-ic}{ic+V}$, they reflect from each other, or otherwise, with amplitude $\dfrac{V}{ic+V}$ they tunnel through without any interaction. In any case, the higher momentum passes through the lower momentum and starting from a configuration $x_1 < x_2 < \ldots <x_n$ for $n$ balls with momenta are in a decreasing order $p_1>p_2>\ldots > p_n$, the wave-function will end up in a configuration with momenta in the increasing order $p_n< p_{n-1}< \ldots < p_1$. 

Each of these pairwise scatterings can be viewed as a local quantum gate, and the collection of scatterings as a quantum circuit. In order to see this, consider an $n!$ dimensional Hilbert space for $n$ particles with orthonormal basis $\{ |\sigma \rangle : \sigma \in S_n \}$. Assume an initial state of $|1,2,3,\ldots, n\rangle$, with defined momenta $p_1, p_2, \ldots, p_n$. These momenta and the initial distance between the particles specify in what order the particles will collide. It is instructive to view the trajectory of the particles as $n$ straight lines for each of these particles in an $x-t$ plane. Time goes upwards and the intersection between each two lines is a collision. In each collision, either the label of the two colliding balls is swapped or otherwise left unchanged. The tangent of each line with the time axis is proportional to the momentum of the ball that the line is assigned to in the first place. Balls with zero relative velocity do not interact, as lines with equal slope do not intersect. Suppose that the first collision corresponds to the intersection of line $t$ with line $t+1$. Such a collision occurs when $p_t>p_{t+1}$. Then the initial state is mapped to:

$$
\Big |1,2, \ldots, n\Big \rangle \rightarrow \dfrac{-ic}{ic + V_{t,t+1}}\Big |1,2, \ldots t, t+1, \ldots, n\Big \rangle + \dfrac{V_{t,t+1}}{ic + V_{t,t+1}}  \Big |1,2, \ldots t+1, t, \ldots, n\Big \rangle
$$

Where, $V_t= p_t - p_{t+1}$. This map can be viewed as a $n!\times n!$ unitary matrix:

$$
H(p_t- p_{t+1}, t):=H(p_t, p_{t+1},t):=\dfrac{-ic}{ic + V_{t, t+1}} I + \dfrac {V_{t,t+1}}{ic + V_{t, t+1}} L_{(t, t+1)}
$$

\noindent Here $I$ is the $n! \times n!$ identity matrix, $L_{(t,t+1)}$ is the $n!\times n!$ matrix which transposes the $t$ and the $t+1$'th labels of the basis states. $H(u,t )$ acts only on the $t$ and the $t+1$'th balls only. $u$ is the velocity of the $t$'th ball relative to the $t+1$'th balls. From now on we refer to these unitary gates as the ball permuting gates. One can check that these gates are unitary:

$$
H(u,t ) H^\dagger(u,t)=H (u,t) H(-u,t)=I.
$$

Given $n$ particles, with labels $i_1, i_2, \ldots, i_n$, and momenta $p_1, p_2, \ldots, p_n$, we can obtain a quantum circuit with gates $H(u_1,t_1), H(u_2,t_2),\ldots, H(u_m,t_m)$, one for each intersection of the straight lines. The scattering matrix in this theory is then given by the product $S=H(u_m,t_m), H(u_{m-1},t_{m-1}),\ldots, H(u_1,t_1)$. 

In general, the label of the balls can be repeated, and the matrix $S$ has a block diagonal form. For each tuple $I=i_1, i_2, \ldots, i_n$, assign a vector $X_I=(x_1, x_2, \ldots, x_n)$, where $x_j$ is the number of times that the index $j$ appears in $I$. Clearly, $\sum_j x_j=n$. Given this description, the blocks of $S$ are marked by vectors $X$, that is, the block $X= (x_1, x_2, \ldots, x_n)$, consists of basis entries for which the index $1$ appears for $x_1$ times, the index $2$ for $x_2$ times and so on. $S$ is an $n^n \times n^n$ matrix, and the dimension of the block $X$ is given by $\dfrac{n!}{x_1 ! x_2 ! \ldots x_n !}$. For the purpose of this thesis, We are interested in the block $(1,1,1,\ldots, 1)$, where the entries of the $S$ matrix are marked by permutations of the numbers $\{1,2,3,\ldots, n\}$. The product of symbols with distinct labels $A_1(\theta_1) A_1(\theta_1)\ldots A_n(\theta_n)$ can be formulated similarly using product of two-local unitary gates.

An important ingredient of these quantum gates is the so-called Yang-Baxter equation ~\cite{yang1967some, baxter1972partition}, which is essentially the factorization condition, and is the analogue of the associativity of Zamalodchikov algebra. The Yang-Baxter equation is a three ball condition, and is according to:

$$
H(u,t) H(u+v, t+1) H(v,t)= H(v, t+1) H(u+v, t) H(u, t+1).
$$

Basically, the Yang-Baxter equation asserts that the continuous degrees of freedom like the initial position of the particles does not change the outcome of a quantum process, and all that matters is the relative configuration of them. In order to see the line representation of the Yang-Baxter equation, see Figure \ref{fig22}. Also, the Yang-Baxter equation imposes overall symmetries on the larger diagrams, see Figure \ref{fig23} for an example. Consider three balls with labels $1, 2, 3$, initialized with velocities $+u$, $0$ and $-u$, respectively. If we place the middle ball very close to the left one, the order of collisions would be $1-2\rightarrow 2-3 \rightarrow 1-2$. However, if the middle one is placed very close to the third ball, the order would be $2-3\rightarrow 1-2 \rightarrow 2-3$. The Yang-Baxter equation asserts that the output of the collisions is the same for the two cases. Therefore, the only defining parameters are the relative configurations, and the relationships between the initial velocities. The Yang-Baxter equation has an important role in many disciplines ~\cite{YBrev}, these range from star-triangle relations in analog circuits to lattice models of statistical mechanics. Also it can be related to the braiding of $n$ tangles, in the sense that a collision corresponds to the braiding of two adjacent tangles. Braid~\cite{Braid} group is defined by $B_n$ generated by elements $b_j$ for $j\in [n-1]$. The defining feature of the braid group is the two conditions: $[b_i, b_j]=0$ for $|i-j|\geq 2$, and $b_i b_{i+1} b_i=b_{i+1}b_{i}b_{i+1}$ for $i\in[n-2]$. The first property is readily satisfied for the ball permuting gates, and the second property corresponds somehow to the Yang-Baxter equation.

\subsection{Semi-classical Model}

It is not conventional to do a measurement at the middle of a scattering process. However, in the discussed integrable models, it sounds that the two particle interactions occur independently from each other, and the scattering matrix is a product of smaller scattering matrices. Moreover, in the regime that we are going to consider, no particle creation or annihilation occurs. Therefore, it sounds reasonable to assume that at the middle of interactions the particles (balls) are independent from each other, and no interactions occur, unless two particles collide. So, in the following sections, we assume that balls start out from far distances and the nondeterminism in the momentum variable is small. Therefore, a semi-classical model is considered, where the balls move according to actual trajectories, and it is possible to track and measure them in between and stop the process whenever we want at the middle of collisions. According to this assumption quantum effects occur only at the collisions and measurements.


\section{Quantum Complexity Theory}

As discussed we compare the complexity of the models using reductions; this is translated in the question of which system can efficiently simulate the other ones. Therefore, in this subsection we review $\BQP$, the standard complexity class for quantum computing, and will use this model and its variations as the point of reference in reductions.

A qubit as the extension of a bit to quantum systems, is a quantum state in $\C^2$. Let $|0\rangle$ and $|1\rangle$ be an orthonormal basis state for $\C^2$, and we assume that an experimenter can measure the qubit in these basis. Such a basis state is called the computational basis. The extension of strings to quantum computing is given by quantum superpositions over $(\C^2)^{\tensor n}$, for some $n\geq 1$. Therefore, a quantum computer can create quantum probability distributions over strings of qubits, $\sum_{x\in \{0,1\}^n} \alpha_x |x\rangle$. $\alpha_x$ are complex numbers, amounting to $\sum_{x \in \{0,1\}^n} |\alpha_x|^2=1$. A quantum algorithm then is a way of preparing a quantum superposition from which a measurement reveals nontrivial information about the output of a computing task. Therefore, we use a quantum circuit to produce such a superposition.

To compare unitary operators with each other, there are variety of definitions for the state norms and operator norms; however, in the context of this research, they all give similar results. More specifically, if $|\psi\rangle$ is a vector, the $L_2$ norm of this vector is defined as:

$$
\| |\psi\rangle \|_2 :=\left( \sum_j |\psi_j|^2 \right)^{1/2}
$$
\noindent A valid distance between two operators $U$ and $V$ is then defined as \cite{ nielsen2010quantum}:

$$
d(U,V) := sup_{\| |\psi\rangle \| =1 } \| (U-V)|\psi\rangle \|_2
$$

A local quantum gate set is a set of unitary operators $G$, each of which affecting a constant number of qubits at a time. A quantum circuit on $n$ qubits is then a way of composing the gates in $G$ on $n$ qubits. $G$ as gate-set is called dense or $\BQP$-universal if for any $n>0$, for any unitary operator $U$ on $n$ qubits and any $\epsilon >0$, there is a quantum circuit in $G$ which amounts to a unitary that is $\epsilon$-close to $U$.

\begin{definition}
A group $(G,\cdot)$ is a set $G$ and a binary operation $\cdot$ with the associative map map $(g_1, g_2) \mapsto g_1 \cdot g_2$, with the following structures: $1) G$ is closed under $\cdot$, $2)$ there is an element $e \in G$ with $e\cdot g= g\cdot e =g, \forall g \in G$, and $3)$ for all $g \in G$ there exists $g^{-1}$ such that $g^{-1}\cdot g = g \cdot g^{-1} =e$. 
\end{definition}

The set of $n\times n$ real and invertible matrices with the matrix multiplication create a group, which we call it the general linear group, $GL(n, \mathbb{R})$. Let $O(n)$ be the set of orthogonal matrices as a subset of $GL(n, \mathbb{R})$. These are matrices with orthonormal columns and rows. The determinant of an orthogonal matrix is either $1$ or $-1$. Determinant of a matrix is a homomorphism with respect to matrix multiplication, thereby the subset of $O(n)$ corresponding to determinant $1$ is a subgroup called the special orthogonal group $SO(n)$. Similarly, we can define the same groups with matrices over the field of complex numbers. These are $GL(n,\C)$, $U(n)$ and $SU(n)$. The determinant of a unitary matrix is a phase, \i.e., a complex number of the form $e^{i\phi}, \phi \in \mathbb{R}$. $SU(n)$ is thereby the (connected) proper subgroup of $U(n)$ with determinant $1$. See the containment relations in Figure \ref{groups}.

From a computing perspective, we are interested in programming a quantum gate set into the unitary group. This can be achieved by compositions of gates which approximate every element of $U(n)$. Denseness of the gate set in $U(n)$ is then a sufficient condition. However, this is not a necessary condition for universal computing. As a first observation, the overall phase of a unitary matrix is not an observable in the output probability distribution. Therefore, a gate set that is dense in $SU(n)$ would suffice for universal computation.

Let $\phi : GL(n, \C) \rightarrow GL(2n,\mathbb{R})$, be the map which replace each entry $M_{ij}= m e^{i\theta}$ of $M \in GL(n,C)$ with a $2\times 2$ real matrix:

$$
m
\begin{pmatrix} 
\cos\theta & -\sin \theta \\
\sin \theta & \cos \theta
\end{pmatrix}.
$$

$\phi$ is a homomorphism and respects the group action. Let $\{ |j\rangle : j \in [n]\}$ and $\{|j_1\rangle, |j_2\rangle : j \in [n]\}$ be basis for $GL(n,\C)$, and $GL(2n, \mathbb{R})$, respectively. Then if $M \in GL(n,\C)$ maps $\sum_{j\in [n]}\alpha_j |j\rangle$ to $\sum_{j\in [n]}\beta_j |j\rangle$, then $\phi(M)$ maps $\sum_{j\in [n]}\Re {\alpha_j} |j_1\rangle+\Im {\alpha_j} |j_2\rangle$ to $\sum_{j\in [n]}\Re {\beta_j} |j_1\rangle+\Im {\beta_j} |j_2\rangle$. If $M$ is a unitary matrix, $\phi(M)$ is an orthogonal matrix. Moreover, the determinant of $\phi(M)$ is $1$. In order to see this write $M=V D V^{\dagger}$, where $D$ is a diagonal matrix consisting of phases only, and $V$ is a unitary matrix. Then $\phi(M)= \phi (V) \phi (D) \phi (V^\dagger)=\phi (V) \phi (D) \phi (V)^T$. Thereby, $\det (M)=\det(\phi (V))^2 \det(\phi (D))=\det (\phi(D))$. $\phi(D)$ has a block diagonal structure, and the determinant of each block is individually a $1$. Therefore $\phi$ sends $U(n)$ to a subset of $SO(2n)$. From this we conclude that denseness in $SO(n)$ is a more relaxed sufficient condition for universal quantum computing. See Figure \ref{groups} for the relationship between these.

\begin{figure}[tp]
\centering
\includegraphics[height=3.0in]{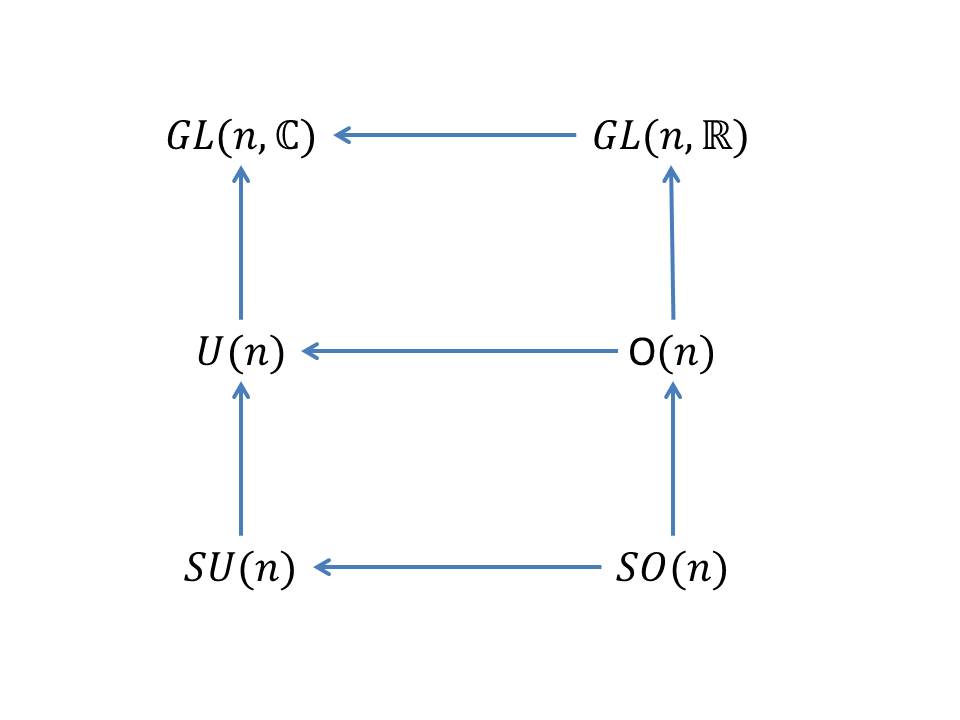}
\caption[Relationships between general linear, unitary and orthogonal groups]{Containment relations between general linear, unitary, orthogonal, special unitary and special orthogonal groups groups.
}
\label{groups}
\end{figure}

As a first step we need universal gate-sets that act on $Q=\C^2$. A qubit is a normal vector in $\C^2$, therefore, any such complex vector can be specified using three real parameters:

$$
|(\psi, \theta, \phi)\rangle = \exp(i \psi)( \cos \theta/2 |0\rangle + \sin \theta_2 e^{i\phi} |1\rangle)
$$

The overall phase $\psi$ is unobservable and we can drop it, and qubits can be represented by $|(\theta, \phi)\rangle$. In other words, we take two states equivalent if and only if they are equal modulo a global phase. $|(\theta+ 2\pi, \phi)\rangle = - |(\theta, \phi)\rangle$ are projectively equivalent. The choice of $\theta/2$ is important, and with this choice the space of qubits (modulo overall phase) is isomorphic to points $(\theta, \phi)$ on a unit $2$-sphere, corresponding to the surface $(\cos \theta \cos \phi, \cos \theta \sin \phi, \cos \theta)$. This sphere is referred to as the Bloch Sphere. Therefore, a qubit is programmable if given any two points on the Bloch sphere there is a way to output a unitary operator which maps one point to the other.

Define the Pauli operators on the Hilbert space $\C^2$ as:

\begin{equation}
\sigma_x :=
\begin{pmatrix}
0&1\\
1&0
\end{pmatrix},
\hspace{1cm}
\sigma_y :=
\begin{pmatrix}
0&-i\\
i&0
\end{pmatrix},
\hspace{1cm}
\sigma_z :=
\begin{pmatrix}
1&0\\
0&-1
\end{pmatrix}
\end{equation}

\vspace{1cm}

\noindent
These operators are Hermitian, unitary, traceless and have determinant equal to $-1$. They anti-commute with each other and each of them squares to the identity operator. We say that two operators $A, B$ anti-commute if $\{A,B\}=A B + B A=0$, or in other words $A B= - B A$. Moreover they satisfy the commutation relation:

$$
[\dfrac{1}{2}\sigma_{i}, \dfrac{1}{2}\sigma_{j}]=i \dfrac{1}{2} \sum_{k\in \{x,y,z\}}\varepsilon_{ijk}\sigma_k, \hspace{0.1cm} \forall. i, j \in \{x,y,z\}
$$

\noindent
$[\cdot, \cdot]$, is the commutator operator and maps $A, B \mapsto A B- B A$. $\varepsilon_{ijk}$ is the Levi Civita symbol, amounts to zero if any pair in $i, j, k$ are equal, otherwise gives a $1$ if the order of $(i,j,k)$ is right-handed, and otherwise takes the value $-1$. A triple $(i,j,k)$ is called right-handed if it is equal to $(1,2,3)$, and is called left-handed if the order is $(2,1,3)$, modulo cyclic rotation. We usually drop the summation for simplicity. In short the Pauli operators satisfy $\sigma_i \sigma_j = \delta_{ij} + i \varepsilon_{ijk} \sigma_k$. 

Any $2\times 2$ unitary operator with unit determinant, can be decomposed as $R(v)=v_0+ i (v_1 \sigma_x + v_2 \sigma_y + v_3 \sigma_z)$, where $v:=(v_0, v_1, v_2, v_3)=: (v_0,\bf{v}) \in \mathbb{R}^4$, and pose the structure $v^2_0+v^2_1+v^2_2+v^2_3=1$. We can thereby use equivalent parameterization $(v_0, v_1, v_2, v_3)= (\cos \theta , \sin \theta \bf{n})$, where $\bf {n} \in \mathbb{R}^3$ is a unit vector.

\noindent
If we define the exponential map as the limit of $M \mapsto \exp(M):= \sum_{j=0}^\infty \dfrac{M^j}{j!}$, then $R(v)=\exp(i\theta \bf{n} . \bf{\sigma})$. Where, $\bf{\sigma}= (\sigma_x, \sigma_y, \sigma_z)$, and $\bf{n}. \sigma$ is the usual inner product of the two objects. The object sitting in the argument of the exponential map has the structure of a vector space, with $\sigma_x, \sigma_y, \sigma_z$ as its linearly independent basis. Along with the commutation relation as the vector-vector action it has the structure of an algebra. This algebra is called the Lie algebra $\su(2)$. The exponential map is an isomorphism between $SU(2)$ and $\su(2)$. The element $R_j(\theta):=\exp(i \theta \sigma_j)$ is called the single qubit rotation along the $j$ axis, for $j\in \{x,y,z\}$. Indeed, any element in $SU(2)$, can be decomposed as a composition of two rotations $R_x, R_y$, only.

We can extend this to larger dimensions. Any unitary matrix $A$ with unit determinant can is related to a traceless Hermitian operator $H$ with the exponential map $\exp(i H)$. Let $\su(n)$ the vector space over $\mathbb{R}$, with $n\times n$ traceless Hermitian matrices as its linearly independent basis. Again the exponential map is an isomorphism between $SU(n)$ and $\su(n)$. $\su(n)$ as a vector space has dimension $n^2-1$. Also, we know that elements of the set of $n\times n$ unitary matrices with unit determinant can be specified with $n^2-1$ real parameters. 

Other well known qubit operations are Hadamard $H$ and $\pi/8$ gate P:

\begin{equation}
H :=
\dfrac{1}{\sqrt{2}}\begin{pmatrix}
1&1\\
1&-1
\end{pmatrix},
\hspace{1cm}
P :=
\begin{pmatrix}
e^{-i\dfrac{\pi}{8}}&0\\
0&e^{i\dfrac{\pi}{8}}
\end{pmatrix}
\end{equation}

\vspace{1cm}

The importance of a Hadamard gate is that its action $H^{\tensor n}$ in parallel maps $|0\rangle^{\tensor n}$ to an equal superposition over $n$ bit strings, \i.e, $\dfrac{1}{2^{n/2}}\sum_{x \in \{0,1\}^n}|x\rangle$.

Let $\C^d$ be a Hilbert space, with orthonormal basis $\{|e_j\rangle\}_{j\in[d]}$. Consider the Lie algebra $g_{ij}$ generated by the operators:

$$
|e_i\rangle \langle e_j |+|e_j\rangle \langle e_i|, \hspace{0.5cm}- i |e_i\rangle \langle e_j |+ i|e_j\rangle \langle e_i|, \hspace{0.5cm} |e_i\rangle \langle e_i |-|e_j\rangle \langle e_j|,
$$

\noindent
for $i < j$. Clearly, $g_{ij}$ is closed under Lie commutation, and is isomorphic to $\su(2)$. Its image under the exponential map, $G_{ij}$, is isomorphic to $SU(2)$, and corresponds to quantum operators (with unit determinant) that impose rotations on the subspace spanned by $|e_i\rangle$ and $|e_j\rangle$, and acts as identity of the rest of the Hilbert space. Such set of operation is called a two level gate. If we allow operations from $G_{ij}$ for all $i<j$, then the corresponding gate set is called a two-level system. A two level system is universal, and is dense in $SU(d)$:

\begin{theorem}
Let $g$ be the vector space generated by $\cup_{i<j} g_{ij}$ then $g= \su(d)$.
\end{theorem}

\begin{proof}
$g \subseteq \su(d)$, since elements of $g_{ij}$ are traceless and Hermitian. Pick any Hermitian matrix $M$ with vanishing trace. Then:

\begin{eqnarray*}
M&=&\sum_{i<j\in [d]} m_{ij} |e_i\rangle \langle e_j |+m^\star_{ij} |e_j\rangle \langle e_i |+ \sum_{i\in [d]} m_{ii}|e_i\rangle \langle e_i|\\
&=&\sum_{i<j\in [d]} \Re m_{ij}( |e_i\rangle \langle e_j | + |e_j\rangle \langle e_i |) + \Im m_{ij}( i |e_i\rangle \langle e_j | - i |e_j\rangle \langle e_i |)\\
&&+ \sum_{i\in [d]} m_{ii}|e_i\rangle \langle e_i|\\
\end{eqnarray*}

\noindent with $\sum_{i\in [d]} m_{ii}=0$. The off-diagonal terms are manifestly constructible with $g_{ij}$ basis. The last term is also constructible with $g$ basis:

\begin{eqnarray*}
\sum_{i \in [d]} m_{ii} |e_i\rangle \langle e_i|&=&\sum_{i \in [d-1]} m_{ii} |e_i\rangle \langle e_i|+ m_{dd}|e_d\rangle \langle e_d|\\
&=&\sum_{i \in [d-1]} m_{ii} (|e_i\rangle \langle e_i|-|e_d\rangle \langle e_d|)
\end{eqnarray*}
\end{proof}

\begin{corollary}
A two level system on $\C^d$ can generate elements of $SU(d)$.
\end{corollary}

\begin{corollary}
The following $d^2-1$ elements create a linearly independent basis for $\su(d)$:

$$
|e_i\rangle \langle e_j |+|e_j\rangle \langle e_i|,
$$

$$
- i |e_i\rangle \langle e_j |+ i|e_j\rangle \langle e_i|,
$$

\noindent for all $i<j\in [d]$, and:

$$
|e_d\rangle \langle e_d| - |e_i\rangle \langle e_i|,
$$

\noindent for $i\in [d-1]$.
\end{corollary}

Any $d\times d$ unitary matrix with unit determinant can be decomposed as the composition of $\dfrac{d(d-1)}{2}$ two level gates. For computation purpose we want to program a gate set to act on multi-qubit systems, that is the Hilbert space $Q^{\tensor n}$, which has dimension $2^n$. Therefore a two-level system on $Q^{\tensor n}$ consists of exponentially many elements, and is not an efficient choice for computing. Therefore, we are looking for gates that act on constant number of qubits at a time and generate a dense subgroup of $SU(Q^{\tensor n})$.

An important two qubit gate is the controlled not (CNOT) gate \cite{nielsen2010quantum}, which maps the basis $|x,y\rangle$ to $|x,x\oplus y\rangle$, for $x, y \in \{0,1\}$. That is it flips the second bit if the first bit is set to $1$. Indeed, we can discuss a controlled-U gate as $|0\rangle \langle 0| \tensor I + |1\rangle \langle 1| \tensor U$ for any unitary $U$. Therefore, a CNOT gate is the controlled $\sigma_x$ gate. A controlled phase gate is the one with $U=\sigma_z$. CNOT is also related to the classical reversible circuits, that is a boolean gate which flips the second bit conditioned on the status of the first bit. Among these, CNOT has a classical circuit analogue as we discussed previously. Also inspired by the classical reversible gates, we can discuss three qubit gates. Among these are the quantum Fredkin and Toffoli gates. A Fredkin gate is a swap controlled on two qubits controlled by a third qubit, that is the maps $|0,x,y\rangle\rightarrow |0,x,y\rangle$ and $|1,x,y\rangle\rightarrow |1,y,x\rangle$. A Toffoli or controlled-controlled not or CCNOT is the gate which acts as $\sigma_x$ a qubit controlled on the status of two other bits; it gives the map $|x,y,z\rangle\rightarrow |x,y,x.y\oplus z\rangle$. Here $x.y$ is the logical AND of $x$ and $y$.

As described two level systems are universal, but not efficiently programmable. There are a number of well known universal gate sets. CNOT with arbitrary qubit rotation are $\BQP$ universal by simulating the two level systems. Moreover, for any unitary $U$, there is a way of assigning real angles to the rotations, such that the composition of gates in this gate set simulates $U$ exactly. However, given finite rotation gates with irrational rotation angles, along with CNOT generate a dense subset of $SU(Q^\tensor(n))$. CNOT and $X$ rotations can simulate any special orthogonal matrix, and by the discussion of embedding complex matrices into real ones these are still universal for quantum computing. There are other known $\BQP$-universal gate set. For example, Hadamard with Toffoli generate a dense subset of the orthogonal group (see \cite{aharonov2003simple} for a proof), and also $\pi/8$ phase gate, Hadamard with CNOT are also universal for $\BQP$. However, the composition of Hadamard and CNOT gates generate a sparse subset of unitary matrices and the output of any such quantum circuit can be simulated in polynomial time.

A quantum computer works in three steps, initialization, evolution and measurement. The initialization is due to a polynomial time Turing machine which on the size of the input outputs the description of a quantum circuit in some universal gate set. Evolution is simply the action of the quantum circuit on the input $|x\rangle\tensor |00\ldots 0\rangle$. The $|x\rangle$ part is the string in computational basis and the state $|00\ldots 0\rangle$ is a number of extra bits which mediate the computation. These extra bits are called ancilla qubits. Measurement is basically sampling from the output distribution of $C|x 00\ldots 0\rangle$ in the computational basis. Moreover, for a decision problem, we can deform $C$ in such a way that measuring the first bit is sufficient to obtain a nontrivial answer.

\begin{definition}
(Bounded-error quantum polynomial time \cite{bernstein1993quantum}) a language $L\subset \{0,1\}^\star$ is contained in $\BQP$, if there a polynomial time Turing machine $M$, which on the input $1^|x|$, outputs the description of a quantum circuit $C$ in some universal gate set, such that if $x \in L$, the probability of measuring a $1$ in the first qubit is $\geq 2/3$ and otherwise $\leq 1/3$.
\end{definition}

One important result is the Solovay-Kitaev theorem, which asserts that denseness of a gate set implies efficiency:

\begin{theorem}
(Solovay-Kitaev\cite{kitaev1997quantum}) all $\BQP$-universal gate sets are equivalent: suppose that $G$ is a gate set consisting of $SU(d)$ gates. If $G$ is dense in $SU(d)$ and is furthermore closed under inverses, then for any element $U$ in $SU(d)$, there is a quantum circuit $C'$ of size $\log^{O(1)}(1/\epsilon)$ composed of $G$ gates that is $\epsilon$-close to $C$.
\label{solkit}
\end{theorem}

\section{The One-Clean-Qubit Model}

While state of a quantum system is a pure vector in a Hilbert space, most of the time the actual quantum state is unknown; instead, all we know is a classical probability distribution over different quantum states, \i.e., the given quantum state is either $|\psi_1\rangle$ with probability $p_1$, or $|\psi_2\rangle$ with probability $p_2$, and so on. In other words, the state is an ensemble of quantum states $\{(p_1, |\psi_1\rangle) ,(p_2, |\psi_2\rangle),\ldots, (p_n, |\psi_n\rangle) \}$, for $p_1 + p_2 + \ldots p_n =1$. Such an ensemble is a mixture of quantum probability and classical probability distributions at the same time, it is also called a mixed state, and is described by a density matrix $\rho$:

$$
\rho = \sum_{j \in [n]} p_j |\psi_j \rangle \langle \psi_j |.
$$

\noindent A density matrix is a Hermitian operator, with nonnegative eigenvalues and unit trace. A quantum state is called pure, if it has a density matrix of the form $|\psi\rangle \langle \psi |$. In other words, a quantum state $\rho$ is pure if and only if $tr(\rho^2)=1$.

If some quantum state is initially prepared in the mixed state $\rho_0$, then given a unitary evolution $U$, the state is mapped to $U\rho_0 U^\dagger$. Let $\{|j\rangle: j \in [n]\}$ be some orthonormal basis of a Hilbert space. The maximally mixed state of this Hilbert space has the form $\dfrac{1}{n} \sum_j  |j\rangle  \langle j|= \dfrac{I}{n}$, is a quantum state which contains zero quantum information in it. That is, the outcome of any measurement can be simulated by a uniform probability distribution on $n$ numbers. Also, a maximally mixed state is independent of the selection of the orthonormal basis. Quantum computing on a maximally mixed state is hopeless, since $\dfrac{I}{n}$ is stable under any unitary evolution.

Consider the situation where we can prepare a pure qubit along with $n$ maximally mixed qubits to get $|0\rangle \langle 0| \tensor \dfrac{I}{n}$. The state $|0\rangle \langle 0|$ is also referred to as a clean qubit. In this case, the quantum state has one bit of quantum information in it. It is also believed that there are problems in $\DQC1$ that are not contained in the polynomial time. One example of such problem, is the problem of deciding if the trace of a unitary matrix is large or small. No polynomial time algorithm is known for this problem. We are going to point out to the trace computing problem later in section ~\ref{trace}. Moreover, if we consider the version of $\DQC1$ where we are allowed to measure more than one qubits, then it is shown that there is no efficient classical simulation in this case, unless the polynomial Hierarchy collapses to the third level.  In the version of my definition, since we used a polynomial Turing machine as a pre-processor, $\DQC1$ immediately contains $\P$. Pre-processing can be tricky for the one-clean-qubit model. For example, as it appears, if instead of $\P$, we used $\NC^1$, the class $\P$ and $\DQC1$ are incomparable.

\begin{definition}
(The one-clean-qubit model \cite{knill1998power}) $\DQC1$ is the class of decision problems that are efficiently solvable with bounded probability of error using a one-clean-qubit and arbitrary amount of maximally mixed qubits. More formally, it is the class of languages $L\subseteq \{0,1\}^\star$, for which there is a polynomial time Turing machine $M$, which on any input $x \in \{0,1\}^\star$, outputs the description of a unitary matrix $\langle U \rangle$ with the following property: if $x \in L$, the probability of measuring a $|0\rangle$ on the first qubit of $U |0\rangle \langle 0| \tensor \dfrac{I}{n} U^\dagger$ is $\geq 2/3$, and otherwise it is $\leq 1/3$. Here $U$ is a $2 n \times 2 n$ unitary matrix.
\end{definition}

Notice that if we allow intermediate measurements we will obtain the original $\BQP$; just measure all qubits in $\{|0\rangle, |1\rangle \}$ basis, and continue on a $\BQP$ computation. Clearly, $\DQC1$ is contained in $\BQP$; in order to see this, just use Hadamrds and intermediate measurements to prepare the maximally mixed state, and continue on a $\DQC1$ computation. It is unknown whether $\BQP \subseteq \DQC_1$, however, we believe that this should not be true.

\section{Complexity Classes with Postselection}

Here we define the complexity classes with postselection. Intuitively, these are the complexity classes with efficient verifiers with free retries. That is an algorithm which runs on the input, and in the end will tell you whether the computation has been successful or not. The probability of successful computation can be exponentially small.

\begin{definition}
Fix an alphabet $\Sigma$. $\Post\BQP$ ($\Post \BPP$) is the class of languages $L \subset \Sigma^\star$ for which there is a polynomial time quantum (randomized) algorithm $\mathcal{A}: \Sigma^\star \rightarrow \{0,1\}^2$, which takes a string $x\in \Sigma^\star$ as an input and outputs two bits, $\mathcal{A}(x)= (y_1, y_2)$, such that:

\begin{itemize}

\item[] 1) $\forall x \in \Sigma^\star, \operatorname*{Pr}(y_1 (x)=1)>0$.

\item[] 2) If $x \in L$ then $\operatorname*{Pr}(y_2 (x)=1| y_1(x)=1)\geq \dfrac{2}{3}$

\item[] 3) If $x \notin L$ then $\operatorname*{Pr}(y_2 (x)=1| y_1(x)=1)\leq \dfrac{1}{3}$
\end{itemize}
\end{definition}

\noindent Here $y_1$ is the bit which tells you if the computation has been successful or not, and $y_2$ is the actual answer bit. The conditions $2)$ and $3)$ say that the answer bit $y_2$ is reliable only if $y_1=1$. In this work, we are interested in the class $\Post\BQP$. However, $\Post \BPP$ is interesting on its own right, and is equal to the class $\BPP_{path}$, which a modified definition of $\BPP$, where the computation paths do not need to have identical lengths. $\Post \BPP$ is believed to be stronger than $\BPP$, and is contained in $\BPP^{\NP}$, which is the class of problems that are decidable on a $\BPP$ machine with oracle access to $\NP$ or equivalently $\SAT$.

Due to a seminal result by Aaronson, $\Post \BQP$ is related to the complexity class $\PP$:

\begin{theorem}
(Aaronson \cite{aaronson2005quantum}) $\Post\BQP=\PP$.
\end{theorem}

\noindent Firstly, because of $\P^{\PP}=\P^{\#\P}$ as a corollary to the theorem, with oracle access to $\Post\BQP$, $\P$ can solve intricate counting tasks, like counting the number of solutions to an $\NP$ complete problem. The implication of this result for the current work is that if a quantum model, combined with postselection is able to efficiently sample from the output distribution of a $\Post\BQP$ computation, then the existence of a  randomized scheme for approximating the output distribution of the model within constant multiplicative factor is ruled out unless $\PH$ collapses to the third level. This point is going to be examined in section \ref{phcol}.

\chapter{Computational Complexity of the Classical Ball Permuting Model}
\label{ch4}

The results of this chapter has been obtained in joint collaboration with Scott Aaronson. 

In this chapter, a classical analogue for the ball permuting model is formalized. The computational power of this model is then partially pinned down within the known complexity classes.

\section{Classical Computation with Probabilistic Swaps}

Suppose that $n$ distinct colored balls are placed on a line. Label the initial configuration of the balls from left to right with ordinary numbers $1, 2, 3, \ldots, n$. Each configuration of the balls is then represented by a permutation of the set $\{1, 2, 3,\ldots, n\}=[n]$. Therefore, the ball permuting model of $n$-balls has $n!$ distinct states, and the transition rules are given by the set of bijections $: [n]\rightarrow [n]$. These bijections along with their compositions correspond to the well-known symmetric group. In the following some notations and background about the symmetric group is established. This notation is also going to be used in section \ref{dfs} of this chapter.

Given a finite element set $X$ of $n$ elements, let $S(X)\cong S([n])$ be the set of bijections $:X\rightarrow X$. The bijections of $S[n]$ along with $\circ : S[n] \times S[n] \rightarrow S[n]$ as the composition of functions of functions from right-to-left, create a group $S_n:=(S[n],\circ, e)$, with identity $e$ as the identity function.

$S_n$ is a group because each element of $S[n]$ is a bijection, and thereby is invertible. Also, the composition of functions $\circ$ is associative, and the set of bijections are closed under compositions. For any $\pi \in S_n$ construct the structure $\Big \{ \{1, \pi(1), \pi \circ \pi(1), \ldots \}, \{2, \pi(2), \pi \circ \pi(2), \ldots \}, \ldots, \{n, \pi(n), \pi \circ \pi(n), \ldots \}\Big \}=:\{C_{\pi}(1), C_{\pi}(2),\ldots, C_{\pi}(n)\}=: C_{\pi}$. This structure is called the set of cycles of permutation $\pi$, and each element $C_\pi (j):= \{\pi(j), \pi \circ \pi(j), \ldots \}$ is the cycle corresponding to element $j$. The number of cycles in $C_{\pi}$ can vary from $1$ to $n!$ (corresponding to $e$). The size of a cycle is the minimum $|C_\pi(j)|=\min \{k>0 : \pi^k (j)=j\}$. The set $[n]$ is therefore partitioned into the union of disjoint cycles $C_1, C_2, \ldots, C_N$. Let $\lambda_1 \geq \lambda_2\geq\ldots\geq \lambda_N$ be the size of these cycles, clearly $\lambda_1 + \lambda_2+\ldots+\lambda_N=n$.

\begin{definition}
For each positive integer $n$, a partition of $n$ is a non-ascending list of positive integers $\lambda=(\lambda_1 \geq \lambda_2\geq\ldots\geq \lambda_N)$ such that $\lambda_1 + \lambda_2+\ldots+\lambda_N=n$. Denote $\lambda !$ with $\lambda_1 ! \lambda_2 !\ldots \lambda_N !$.
\end{definition}

For any cycle structure of partition (size of cycles) $\lambda_1 \geq \lambda_2\geq\ldots\geq \lambda_N$ there is a subgroup of $S_n$ that is isomorphic to $S[\lambda]:=S_{\lambda_1}\times S_{\lambda_2}\times \ldots \times S_{\lambda_n}$. In other words the subgroup consists of the product of $N$ permutations each acting on a distinct cycle. Any such subgroup has $\lambda ! := \lambda_1 ! \lambda_2 ! \ldots \lambda_N !$ elements.

In order to address each permutation $\pi$ uniquely, we can use an alternative representation for the cycles as an ordered list (tuple) $(y_1, y_2, \ldots, y_k)$, with $\pi (y_t)=y_{t+1\mod k}$. Therefore, each permutation $\pi$ can be represented by its cycles $\pi \cong (C_1, C_2,\ldots, C_k)$. We can define two different actions of the symmetric group $S_n$ on an ordered list $(x_1, x_2, \ldots, x_n)$. The first one of these is called the left action where a permutation $\pi$ acts on a tuple $X=(x_1, x_2, \ldots, x_n)$ as $(\pi, X)\mapsto (x_{\pi(1)}, x_{\pi(2)}, \ldots, x_{\pi(n)})=: l(\pi)$. The left action rearranges the location of the symbols. The right action is defined by $(\pi, X)\mapsto (\pi(x_1), \pi(x_2), \ldots, \pi(x_n))=: r(\pi)$. Throughout by permutation we mean left action, unless it is specified otherwise.

A permutation is called a swap of $i, j$ if it has a nontrivial cycle $(i, j)$ and acts trivially on the other labels. A swap is called a transposition $(i)$ for $i\in [n-1]$ if it affects two adjacent labels $i, i+1$ with a cycle $(i,i+1)$. There are $n-1$ transpositions in $S_n$ and the group is generated by them. We use the notations $(i,j)$ or $b_{i,j}$ for swaps and $b_i$ and $(i)$ for transpositions, interchangeably. Any permutation can be generated by at most $\dfrac{n(n-1)}{2}$ transpositions or $n$ swaps. Indeed, we can define a distance between permutations as $d: S_n\times S_n \rightarrow \mathbb{N}$, with $d(\pi,\tau)$ being the minimum number of transpositions whose application on $\pi$ constructs $\tau$. Given a distance on the permutations we can classify the permutations into even and odd ones.

\begin{definition}
A sign of a permutation $sgn: S_n\rightarrow \{-1,1\}$ is a homomorphism, with the map $\sigma\mapsto 1$ if $d(\sigma, e)$ is even and gives $-1$ otherwise.
\end{definition}

It is not hard to see that sign is a homomorphism. Given this, we can think of $E_n \leq S_n $ as a subgroup of $S_n$ consisting of $\sigma \in S_n$ with $sgn (\sigma)=1$. Clearly, $E_n$ includes the identity permutation, inverse and closure.

\begin{proposition}
For $n\geq 2$, any group $(G, . , e)$ generated by $b_1, b_2, \ldots, b_{n-1}$ is isomorphic to $S_n$ if and only if the following is satisfied:

\begin{itemize}
\item $b_i . b_j = b_j . b_i $ for $|i-j| >1$
\item $b_i^2=e$ for $i \in [n-1]$
\item $b_i b_{i+1} b_i = b_{i+1} b_i b_{i+1}$ for $i \in [n-2]$
\end{itemize}

\label{sym}
\end{proposition}

\subsection{The Classical Ball Permuting Model}
\label{classicalball}

Consider the following problem:

\textit{"($\BALL$) The input is given as a target permutation $\pi\in S_n$ along with a polynomially-long list of swaps $(i_1, j_1), (i_2, j_2), \ldots, (i_m, j_m)$, and a list of independent probabilities $p_1, p_2, \ldots, p_m$. Apply the swaps to the list $(1, 2, 3, \ldots, n)$ in order, each with its corresponding  probability. That is, first apply $(i_1, j_1)$ with probability $p_1$ (and with probability $1-p_1$ do nothing), then apply $(i_2, j_2)$ with probability $p_2$ and so on. 
\begin{itemize}
\item[] - 1) Given the promise that the probability of target permutation is $\geq 2/3$ or $\leq 1/3$, decide which one is the case.
\item[] - 2) Decide if there is a probability support on $\pi$.
\item[] - 3) Compute the probability that $\pi$ appears at the end of the process.
\end{itemize}
"}

The above process creates a probability distribution over the permutations of $S_n$, and three variations are phrased for the problem. How hard is it to fulfill these tasks? Is the complexity of the problem different if all the input swaps are adjacent ones? 

As is justified in appendix \ref{sharpp}, all of the problems can be solved in polynomial time if oracle access to $\# \P$ is provided. The first problem can be solved within $\BPP$, by just using random bits to produce a probability distribution on the set of permutations of $n$ labels. In other words, sampling from the probability distribution $(p_\pi)_{\pi \in S_n}$ can be done withing $\BPP$. Also deciding if the $p_\pi >0$ (the second problem) lies within $\NP$, since the witness for nonzero probability is a subset $S\subseteq [m]$, for which $\pi_S= \pi$. In the next section it is proved that if all of the swaps are adjacent ones, then the problem can be decided in polynomial time, but if the swaps are general ones, the problem is complete for $\NP$.

Lastly, it is worth to mention that finding the marginal probability distribution of the location of each color in the end can be done in classical polynomial time:

\begin{theorem}
Starting with the list $(1,2,3,\ldots,n)$, given a list of swaps $(i_1, j_1), (i_2, j_2),\ldots, (i_m, j_m)$ with corresponding probabilities $p_1, p_2, \ldots, p_m$, there is a polynomial time (in $n$ and $m$) procedure to compute the marginal probability distribution over the locations for each (ball) color $\in [n]$.
\end{theorem}

\begin{proof}
Let $V$ be a vector of size $n$, whose entries $V_j$ for $j\in [n]$ is the probability of finding the first ball in the $j$'th location in the end of ball permutation. The vector is initialized at $(1, 0,0,0,\ldots, 0)$. For steps $t=2, 3, \ldots, m$, let $V_k \leftarrow V_k$ if $k \neq i_t, j_t$, and otherwise map:

$$
V_{i_t}\leftarrow p_t V_{j_t} +(1-p_t)V_{i_t}
$$

and,

$$
V_{j_t}\leftarrow p_t V_{i_t} +(1-p_t)V_{j_t}
$$

This can be done in $O(m.n)$ number of operations.
\end{proof}

\subsection{The Classical Yang-Baxter Equation}

The classical ball permuting model can be alternatively viewed as a stochastic matrices. Consider the set of probability distributions on permutations $\mathbb{V}_n:=\{V\in \mathbb{R^{+}}^{n!}, \sum_j V_j=1\}$. Each entry of $V$ corresponds to a permutation, and its content is the probability that the permutation appears in the process. Denote the basis $\{|\sigma\rangle: \sigma \in S_n\}$ for this vector space. The basis $|\sigma\rangle$ has probability support $1$ on $\sigma$ and $0$ elsewhere. Now a permutation can be viewed as a doubly stochastic $n! \times n!$ matrix with the map $L(\pi)|\sigma\rangle= |\pi\circ \sigma\rangle$. Denote the swap matrix corresponding to $(i,j)$ by $L_{i,j}:=L(i,j)$. A probabilistic swap $(i,j)$ with probability $p$ is therefore given by the convex combination:

$$
R_{i,j} (p)=(1-p) I+p L_{i,j}
$$

One can then formalize the Yang-Baxter equation for three labels (balls). Yang-Baxter (YB) equation is a restriction on the swap probabilities in such a way that the swaps of order $(1,2), (2,3) , (1,2)$ give the same probability distribution as the swaps $(2,3), (1,2), (2,3)$:

\begin{equation}
R_{1} (p_1) R_{2} (p_2)  R_{1} (p_3) =  R_{2} (p'_1) R_{1} (p'_2)  R_{2} (p'_3)
\label{YB}
\end{equation}

We want to solve this equation for $1\geq p_1, p_2, \ldots, p'_6 \geq 0$. If we rewrite $R_{i,j} (p)$ with the parameter $x\in [0,\infty)$:

$$
R_{i,j} (x)=\dfrac{1+x L_{i,j}}{1+x},
$$

\noindent then: 

\begin{theorem}
The following is a solution to equation ~\ref{YB}:

$$
p_1=\frac{x}{1+x} \hspace{1cm} p_2=\dfrac{x+y}{1+x+y} \hspace{1cm}p_1=\dfrac{y}{1+y} \hspace{1cm}
$$
$$
p'_1=\dfrac{y}{1+y} \hspace{1cm} p'_2=\dfrac{x+y}{1+x+y} \hspace{1cm} p'_3=\frac{x}{1+x} \hspace{1cm}
$$

\end{theorem}

\begin{proof}
First we need to show that these are indeed a solution to YB equation. Expanding the equation $R_1 (x) R_2 (x+y) R_1 (y)$
we get:

\begin{equation}
\dfrac{1+xy+(x+y)(L_1+L_2)+(x+y)(x L_1 L_2 + y L_2 L_1) + xy (x+y) L_1 L_2 L_1}{(1+x)(1+x+y)(1+y)}
\label{expand}
\end{equation}

Exchanging $L_1$ with $L_2$ and $x$ with $y$ in the left hand side of equation ~\ref{expand} gives the right hand side. Using the identity $L_1 L_2 L_1 = L_2 L_1 L_2$, however equation ~\ref{expand} is invariant under these exchanges.
\end{proof}

\subsection{Ball Permuting Languages}
\label{ballpermutingoracles}

In this section we formalize the class of problems that are solvable by classical ball permutation, and the aim is to pin down the resulting complexity class among standard complexity classes.

Assume that a configuration balls on a line is given as an oracle which is able to permute them upon demand. We can define different versions of such an oracle. For example, an oracle can permute the balls deterministically or probabilistically. Input to such an oracle is the list of local swaps and the probabilities, and the output is the final permutations obtained from the action of swaps. Thereby classes of languages with ball permutation are defined as those which are reducible to the language these ball permuting oracles. The class $\AC^0$ is used for the reduction:

\begin{definition}
\begin{itemize}
\item[]
\item A (deterministic ball permuting) $\DBall(n)$ oracle takes as input an ordered sequence of labels $(i_1, j_1), (i_2, j_2), \ldots, (i_m, j_m)$, applies them in order (from left) to the identity permutation $(1,2,3,\ldots, n)$, and outputs the resulting permutation.

\item A (randomized ball permuting) $\RBall(n)$ oracle inputs an ordered sequence of labels $(i_1, j_1), (i_2, j_2), \ldots, (i_m, j_m)$, along with a list of independent probabilities $p_1, p_2, \ldots, p_n$, applies the swaps in order (from left) to the identity permutation $(1,2,3,\ldots, n)$ each with its corresponding, \i.e. at step $t$ it applies the swap $(i_t, j_t)$ with independent probability $p_t$ or otherwise applies nothing with probability $1-p_t$. The oracle outputs the resulting permutation.

\item A (nondeterministic ball permuting) $\NBall(n)$ oracle inputs a target permutation $\tau \in S_n$, an ordered sequence of labels $(i_1, j_1), (i_2, j_2), \ldots, (i_m, j_m)$ and an instruction string of $m$-bits $z_1 z_2 \ldots z_n$. The oracle starts with the identity permutation. At each step $t$ if the instruction bit $z_t$ is equal to $1$, it nondeterministically either swaps the permutation according to $(i_t, j_t)$ or otherwise does nothing. Otherwise if the instruction bit is a $0$ it deterministically applies the swap. The oracle outputs the bit $1$ if the target permutation is generated nondeterministically in one of the permuting branches and otherwise outputs $0$.

\item A (counting ball permuting) $\#\Ball(n)$ oracle inputs an ordered sequence of labels $(i_1, j_1), (i_2, j_2), \ldots, (i_m, j_m)$, a target permutation $\tau \in S_n$ and an instruction string of $m$-bits $z_1 z_2 \ldots z_n$. The oracle outputs the number of way that a nondeterministic ball permuting oracle outputs the target permutation and if the $\NBall$ oracle rejects it outputs $0$.

\item Finally we define the complexity classes:
\begin{itemize}
\item $\DBALL$ is the class of languages that are decidable by an $\AC^0$ machine with single query access to a $\DBall$ oracle, and $\DBALL_{adj}$ is the subset where all queried swaps are adjacent transpositions.

\item $\RBALL$ is the class of languages that are decidable by bounded probability of error with an $\AC^0$ machine which makes a single polynomial size query to an $\RBall$ oracle. $\RBALL_{adj}$ is the subset where all the queried swaps are adjacent ones, and $\RBALL^\star_{adj}$ as the subset of $\RBALL_{adj}$ where all probabilities are nonzero.

\item $\NBALL$ is the class of languages that are polynomial time reducible to an $\NBall$ oracle with single query. $\NBALL_{adj}$ is the subset where all swaps are adjacent, and $\NBALL^\star_{adj}$ is the subset of $\NBALL_{adj}$ where the queried instruction string is all ones.

\item $\#\BALL$ is the class of functions that are computable by a polynomial time Turing Machine which makes a single query to a $\#\Ball$ oracle. $\#\BALL^\star_{adj}$ is the subset of $\#\BALL_{adj}$ where the queried instruction string is all ones.

\end{itemize}

\end{itemize}
\end {definition}

The next step is to pin down the defined complexity classes:

\begin{theorem}
$\DBALL=\DBALL_{adj}=\L=\Rev\L$ 
\end{theorem}

\begin{proof}
We first prove the direction $\DBALL=\DBALL_{adj}$. Clearly, $\DBALL_{adj}\subseteq \DBALL$ as a special case. Any nonadjacent swap $(i,j)$ for $i<j$ can be obtained by a sequence of adjacent transpositions $(i,j)= (i) \circ (i+1) \circ \ldots \circ (j-2) \circ (j-1) \circ (j-2) \circ \ldots \circ (i+1) \circ (i)$. Therefore, any sequence of $m$ nonadjacent swaps on $n$ labels can be simulated by $O(m.n)$ number of adjacent transpositions. Thereby $\DBALL \subseteq \DBALL_{adj}$.

In order to prove the direction $\Rev \L \subseteq \DBALL$ we observe from before that the evolution of a reversible computation is according to a configuration space wherein all configuration nodes have in-degree and out-degree at most $1$, and thereby the map which evolves the current configuration of a Turing machine to the next one is a bijection between configurations. A $\LOGSPACE$ (reversible) Turing Machine that uses $c. \log(n)$ space has a configuration space of size $c. n \log (n) n^c=n^{O(1)}$. Notice that such Turing Machine runs for at most polynomial amount of time before looping. Now given any $\LOGSPACE$ machine, consider a $\DBall$ oracle of size $N=n^{O(1)}$. Also without loss of generality we can assume that the Turing Machine on any input runs in a fixed time $T=n^{O(1)}$ for all of its inputs.  Given the description of the Turing machine and its input, we can encode the description of each configuration of the Turing machine with numbers $1, 2, \ldots, N$. Now each step of the computation corresponds to a permutation $:[N]\rightarrow [N]$, and each permutation can be decomposed into $N^{O(1)}$ pairwise permutations (swaps). Therefore an $AC^0$ machine encodes the evolution of the Turing Machine as sequence of swaps. And the evolution of the machine for $T$ steps as the list of swaps repeated for $T$ times. Let $1$ be the encoded initial state. The oracle then applies these swaps in order and in the end we look at the location of the symbol $1$ in the final permutation. If the final location corresponds to an accepting state the $\DBALL$ computer accepts, and otherwise rejects. Now from lemma ~\ref{Mck} we observe that $\Rev \L=\L$ and thereby the direction $\Rev \L=\L\subseteq \DBALL$ is derived.

In order to prove the direction $\DBALL\subseteq \L$, consider any $\DBall$ oracle queried with a list of swaps as the input. We can devise a $\LOGSPACE$ computation in which the final permutation is computed. Notice that the logarithmic space is unable to store the full description of the final permutation. Therefore, instead we design the computation in a way that the location of each symbol $s$ in the final permutation appears on the read/write tape. For this purpose, the machine just keeps track of the current location of $s$ on the read/write tape. If at some step a swap $(i,j)$ is queried, the machine updates the tape if and only if the location of $s$ is either $i$ or $j$. Consider this program as a subroutine. Given the list of swaps and a target permutation as the input, a $\LOGSPACE$ machine implements this subroutine for each symbol and compares the location of the symbol in the final permutation with the input and rejects if they do not match. Then the whole computation accepts if all the tests succeed.
\end {proof}

\begin{theorem}
$\L \subseteq \BPL \subseteq \RBALL=\RBALL_{adj}\subseteq \Almost \L\subseteq \BPP$. However, if we let $\RBALL (2)$ to the class where the $\AC$ machine is allowed to make two adaptive queries, then $\RBALL(2)=\Almost\L$.
\end{theorem}

\begin{proof}
In order to observe $\RBALL=\RBALL_{adj}$, we use the fact that any nonadjacent swap can be produced as application of polynomially many transpositions, thereby $\RBALL_{adj}$ simulates $\RBALL$ by simulating each swap with a sequence of adjacent swaps. More precisely, suppose that $\RBALL$ queries the swap $(i,j)$ with probability $p$. Then $\RBALL_{adj}$ computer first queries $(i) , (i+1)  ,\ldots , (j-2)$ each with probability $1$, then queries $(j-1)$ with probability $p$, and finally queries $ (j-2), \ldots, (i+1), (i)$ each with probability $1$. This proof works independent of the number of queries.

Now suppose that single queries are allowed. As the simulation of $\DBALL$ in $\L$ requires repeated use of a subroutine for each label over and over, any machine to simulate $\RBALL$ with this scheme needs consistent access to the random bits for each subroutine. Given this, we observe that the machines of class $\Almost\L$ provide such consistent access to these random bits. Therefore, an $\Almost\L$ machine runs $\RBALL$ in the $\L$ simulation for each label, and whenever a probabilistic swap is queried, the machine uses the random oracle to decide whether to make it or not. Notice that for this simulation a random oracle is required rather than an ephemeral stream of random bits. That is because the simulation needs to use the same random bits over and over. Suppose that the $\Almost\L$ machine requires $N=n^{O(1)}$ random bits in each $\L$ simulation. Then it picks a lexicographic convention (hardwired to the machine's transition function) on finite strings, and in order to obtain the $j \leq N$'th random bit $b_j$, it queries the $j$'th lexicographic string to the oracle and lets $b_j=1$ if the oracle accepts, and $0$ otherwise.

In order to see the direction $\BPL \subseteq \RBALL$, we amalgamate the configuration space with a single tape cell which is a random bit. Thereby we double the configuration space, by adding $C0$ and $C1$ for each configuration $C$ of the original $\L$ machine. Also we map all configurations of the form $C0$ to odd numbers $1, 3, 5, \ldots$ and all of those with the form $C1$ to even numbers $2, 4, 6,\ldots$. Thereby we just run $\L$ in the $\DBALL$ simulation and whenever the machine needs a random bit on the random bit cell the $\AC^0$ machine queries the swaps $(1,2), (3,4), \ldots$ each with probability $1/2$ to the $\RBall$ oracle.

Now suppose that $\RBALL$ computer is allowed to adaptively query the Ball-Permuting oracle twice. First of all, in this case also $\RBALL\subseteq \Almost\L$. In order to see this suppose that the simulation of the first query to $\RBall$ needs $N_1$ random bits and the second one requires $N_2$ bits. Then the $\Almost \L$ machine first queries the first $N_1$ strings to its random oracle, and when it is done with the first round of $\RBall$ queries, it uses the result to simulate the $AC^0$ machine to design the second query to the oracle. Notice that the $\log n$ read/write tape might not be sufficient to store the whole content of the second query, and thereby it suffices for the machine to just store one bit of the second query at a time and repeat the whole computation to obtain the next bit.

In order to see the $\Almost\L\subseteq \RBALL$ direction, suppose that an $\Almost\L$ machine queries $N$ distinct strings to its oracle. Then an $\RBALL$ computer first queries $N$ transpositions $(1,2), (3,4),\ldots,(2N-1,2N)$ each with probability $1/2$, and uses the result to design the second query to simulate the running of $\Almost\L$ on the queried strings.
\end{proof}

Next we wish to pin down the power of $\NBALL$. In the following proposition we prove that indeed $\NBALL$ can decide all problems in $\NP$. We prove this by a reduction from word problem for the product of permutations ($\WPPP$) which is known to be $\NP$ complete.

\begin{definition}
($\WPPP$) Given the set $\{1,2,3, \ldots, n\}=: [n]$, an ordered list of subsets $S_1, S_2, \ldots, S_m \subseteq [n]$ with $m=\Poly(n)$ and a target permutation $\tau$ on $[n]$, the problem is to decide whether there exist permutations $\pi_1, \pi_2, \ldots, \pi_m$, with each $\pi_j$ acting on the labels of $S_j$ only and identity on the others, such that the combination $\pi_1\circ \pi_2\circ \ldots\circ \pi_m=\tau$.
\end{definition}

\begin{theorem}
(Garey, Johnson, Miller, and Papadimitriou \cite{garey1980complexity}) $\WPPP$ is $\NP$ complete.
\end{theorem}

\begin{proof} 
(Sketch) There is a polynomial time reduction from another $\NP$ complete problem, the ($\VDP$) vertex disjoint path problem: given an undirected graph with source and sink nodes on the boundary, decide if there are paths disjoint in vertex, from source nodes to the sink nodes.
\end{proof}

\begin {theorem}
$\NBALL=\NBALL_{adj}=\NP$
\end{theorem}

\begin {proof}
We introduce a polynomial time reduction from the word problem of permutations $\WPPP$ which is known to be $\NP$ complete. Any instance of $\WPPP$ is given by an ordered list of subsets of $[n]$, thereby for each subset $S =\{i_1, i_2, \ldots, i\}$ we add $ O(k^3)$ swaps. We need to choose the list in such a way that for each permutation on $S$, there is a nondeterministic choice swaps which produces the permutation. Each permutation on $k$ elements can be produced by $O(k)$ swaps. Therefore we list all the swaps on elements of $S$ in some list $L$ and repeat the list for $k$ times. 
\end{proof}

\begin{definition}
(The edge disjoint path problem $\EDP$) Given a directed graph $G$, with source and sink nodes $(s_1, t_1), (s_2, t_2),\ldots,(s_m, t_m),$, decide if there are paths from $s_i$ to $t_i$ for all $i \in [m]$, such that all the paths are disjoint in edge.
\end{definition}

\begin{theorem}
(Wagner and Weihe\cite{wagner1995linear}) $\EDP$ for the case of planar graphs is decidable in linear time.
\label{ww}
\end{theorem}

\begin{theorem}
$\NBALL^\star_{adj}\subseteq\P$
\label{nballstar}
\end{theorem}

\begin{proof}
This can be done by a reduction to the problem of edge disjoint path for directed planar graphs, which is contained in linear time. Given the list of $m$ swaps as input to an $\NBall$ oracle with $n$ balls and a target permutation, we construct a directed planar graph with $m$ nodes, $n$ source edges (nodes) and $n$ sink edges (nodes), according to the following: Initially add $n$ source nodes with $n$ outgoing edges. Number these edges and nodes with $1,2,3, \ldots, n$. For each transposition $(i)$, merge the edge $i$ and $i+1$ in a vertex with two outgoing edges, update the numbering of the edges accordingly, \i.e. name one of the edges to be $i$ and another $i+1$. Continue this for all transpositions. In the end add $n$ sink nodes each taking one of the edges as an input edge, and number them according to the target permutation. Clearly, the resulting graph is planar. This can be seen by induction on the steps of the construction algorithm. Initially the graph is planar. At each step we only merge two adjacent edges and the graph in the next step remains planar. We need to prove that there is a nondeterministic computation of target permutation in the $\NBall$ oracle if and only if there is an edge disjoint path between the source nodes and sink nodes. Suppose that there is an edge disjoint path in the graph, we construct the list of swaps that create the target permutation. Sort the vertices in an ascending order by the distance from the source nodes. Each vertex is mapped to a transposition in the list with the ascending order, \i.e., if some vertex inputs the edges $i$ and $i+1$, then the corresponding transposition is $i$. Two paths are incident to each vertex. Suppose that the input edge of the first path takes the label $i$ and the second one $i+1$. If the output edge of the first path is $i+1$, we include the swap $(i)$ in the instruction list, and otherwise we don't. Each path thereby maps the initial label of its corresponding source node to the desired permuted label in the sink node. Now if there is a nondeterministic computation of the target permutation, then there is a way of choosing the transpositions to construct the target permutation. We then construct the resulting planar graph according to the algorithm and the edge disjoint paths by the following: for each source node pick the first outgoing edge. For each edge in this path then its endpoint corresponds to a vertex which corresponds to a swap in the list. If the edge is labeled by $i$ choose the edge $i+1$ as the next edge if the corresponding swap is active, and otherwise pick the edge $i$.
\end{proof}

\begin{corollary}
The problem of deciding if the probability of a target permutation in $\BALL$ is nonzero is $\NP$ complete; but if the queried swaps are all adjacent ones, then the problem has a polynomial time algorithm.
\end{corollary}

\begin{proof}
There are linear time reductions in both ways between $\NBALL$ and $\BALL$. These are the same problems in two disguises.
\end{proof}

Or in other words: Given a list of swaps and a target permutation, it is in general $\NP-Hard$ to decide if there is a way of constructing the target permutation out of the given swaps. If all the swaps are transpositions, the problem is decidable in linear time.

\begin{corollary}
The edge disjoint path problem in the non-planar case is $\NP$ complete.
\end{corollary} 

\begin{proof}
There is a reduction from $\BALL$ to the edge disjoint path. The reduction is similar to the one given in the proof of lemma ~\ref{nballstar}. The graph instance of $\EDP$ is non-planar if and only if the swaps of $\BALL$ are adjacent. In short we have:

$$
\VDP \leq_\P \WPPP \leq_\P \BALL \leq_\P  \EDP
$$

That means that a polynomial time algorithm for the general $\EDP$ implies $\P=\NP$.
\end{proof}

\begin{theorem}
There is a polynomial time reduction from $\#\BALL$ to the problem of computing the probabilities of $\BALL$.
\end{theorem}

\begin{proof}
Given an instance of $\#\BALL$, if all the bits of the instruction string are ones, output a version of $\BALL$ with the same order of swaps, target permutation $\tau$ and with probabilities all equal to $1/2$. Otherwise, produce an equivalent instance of $\#\BALL$ where all the instruction string bits are ones, and follow the last step. In order to do this, consider a portion of the instruction bit with a sequence of ones followed by zeros and then by a sequence of ones $11\ldots 100\ldots 011\ldots 1$, corresponds to a sequence of swaps $(i_1,j_1)\ldots (i_m,j_m)$ for the first collection of ones, $(i'_1,j'_1)\ldots (i'_{m'},j'_{m'})$ for zeros, and $(i''_1,j''_1)\ldots (i''_{m''},j''_{m''})$ for the second collection of ones. The composition of $(i'_1,j'_1)\circ\ldots \circ(i'_{m'},j'_{m'})$ gives a pure permutation $\pi$. Then there is an equivalent list of swaps with all ones instruction bits according to $(i_1,j_1)\ldots (i_m,j_m) (\pi(i''_1),\pi(j''_1))\ldots (\pi(i''_{m''}),\pi(j''_{m''}))$. We can then iterate over the list until all the zeros are removed. Clearly, the desired number is $2^m p_{\tau}$, where $m$ is the number of (nonzero) instruction bits, and $p_\tau$ is the probability of the target permutation, in the $\BALL$ instance.
\end{proof}

\section{Open Problems}

\begin{itemize}

\item[] 1. We proved that $\BPL \subseteq \RBALL \subseteq \Almost \L \subseteq \BPP$. The class $\Almost \L$ is not known to be contained in $\P$. It remains open to see if $\RBALL \subseteq \P$.

\item[] 2. We proved that with two adaptive queries to the $\RBall$ oracle, $\RBALL = \Almost \L$. Can $\Almost \L$ be simulated with $\RBALL$ using one query?

\item[] 3. What is the complexity of $\#\BALL$ and $\#\BALL^\star_{adj}$ and exact computation of probabilities in $\BALL$? Since $\NBALL=\NP$ there is no approximation within multiplicative error for $\#\BALL$, unless $\P= \NP$. Is this the case that $\# \BALL = \#\P$? Although deciding if some target permutation has nonzero probability is decidable in polynomial time, the counting version still seems like a hard task. 

\item[] 4. We defined the ball permuting oracles with all distinguishable balls. What is the complexity class if the balls are labeled with $0$ and $1$ labels only? Clearly, the original oracles can simulate the binary balls, and in the case of $\DBALL$ there are reductions in both ways. However, it is not clear if we recover $\RBALL$ and $\NBALL$ in the case of binary balls.

\item[] 5. Is there a $\P$ simulation for the case where the probabilities are given by the Yang-Baxter equation? If we pin down this class, how can we program the velocities of the balls to decide problems?

\item[] 6. The problem of deciding if a target permutation has nonzero probability in $\RBALL^\star_{adj}$ is in $\P$, however the same problem for $\RBALL$ is $\NP$ complete. Is $\RBALL^\star_{adj}$ itself a weaker class than $\RBALL$?
\end{itemize}

\chapter{Computational Complexity of Particle Scattering and Ball Permuting Models}
\label{ch5}

The results of this chapter has been obtained in joint collaboration with Greg Kuperberg. 

Using tools from pervious chapters, the computational complexity of the scattering processes of two dimensional integrable models is examined. Quantum ball permuting model is defined as a model of quantum computation, and its relationship with the scattering problem of the integrable model is demonstrated, when all distinguishable particles are considered. Complexity classes corresponding to these models are then partially pinned down within the standard complexity classes. In particular, it is argued that under certain complexity theoretic assumptions, the complexity class is different for different initial quantum states. In order to investigate the scattering problem, a model of quantum computation with ball scattering with intermediate demolition measurements is formally defined. The main result of this chapter is that under the assumption that the polynomial hierarchy is infinite, there is no polynomial time randomized algorithm to sample from the output distribution of the ball scattering computer within multiplicative error.

\section{General Hilbert Space of Permutations}

In section \ref{models}, several quantum models were discussed. In these models the evolution is governed by local gates that are responsible for the exchange of internal degrees of freedom. The approach is to introduce a formal model, the quantum ball permuting model, which captures all of the others. In this model, all degrees of freedom are distinguishable, and the basis of the Hilbert space is marked by the permutations of a finite-element set. We can think of these labels, as colors, distinguishable particles (balls), or internal degrees of freedom like qudits.

Let $\mathcal{H}_n=\C S_n$ be an $n!$ dimensional Hilbert space, with permutations of $n$ symbols as its orthonormal basis.We are interested in the set of local gates according to:

$$
X(\theta, k)=\cos\theta I + i \sin \theta L_{(k,k+1)},
$$

\noindent 
where $\theta$ is a free parameter, and $I$ is the identity operator on the Hilbert space. $L_{(k,k+1)}$ is called the left transposition with the map:

$$
|x_1, x_2, \ldots, x_k, x_{k+1},\ldots, x_n\rangle\mapsto |x_1, x_2, \ldots, x_{k+1}, x_{k},\ldots, x_n\rangle
$$

\noindent for any permutation $x_1, x_2, \ldots, x_n$ of the labels $1,2,3,\ldots, n$.

For example, consider $\C S_2$, spanned by the basis $|12\rangle$ and $|21\rangle$. The matrix form of $X(\theta,1)$ on these basis according to:

$$
X(\theta, 1) =
\begin{pmatrix}
\vspace{2mm}\cos\theta & i\sin\theta\\
i\sin\theta & \cos\theta
\end{pmatrix}.
$$

\noindent This has one free parameter, and moreover, each column is a permutation of the other.

%
%

The operator $L$ can be thought of as a homomorphism $:S_n\rightarrow U(\C S_n)$, we call it a left action because it lets a member $\tau$ of $S_n$ act on a basis state $|\sigma\rangle,  \sigma \in S_n$ according to:

$$
L(\tau)|\sigma\rangle = |\tau \circ \sigma\rangle
$$
%
%
%

We can also talk about a right action $:S_n\rightarrow U(\C S_n)$ with a corresponding map:

$$
R(\tau)|\sigma\rangle = |\sigma\circ \tau\rangle
$$

\noindent While a left action rearranges the physical location of the labels, a right action relabels them, and as it is mentioned in a later section, the unique commutant of the associative algebra generated by left actions is the algebra of right actions. Therefore, it is worthwhile to introduce a right permuting version of a local gate:

$$
Y(\theta, k)=\cos\theta I + i \sin \theta R_{(k,k+1)}
$$

A $Y$ operator simply exchanges the labels $k$ with $k+1$, independent of their actual location, and is not local in this sense. 

The following asserts that the set of unitary operators generated with these gates is small compared to the set of unitary operators on an $n!$ dimensional Hilbert space:

\begin{theorem}
Let $U=X(\theta_m, k_m) \ldots X(\theta_2, k_2) X(\theta_1, k_1)$ be any composition of the $X$ operators, then columns of $U$ as a matrix in $S_n$ basis, are obtainable by permuting the entries of the top-most column.
\label{ind}
\end{theorem}

\begin{proof}
Consider the first column of $U$ spanned by:

$$
U|123\ldots n\rangle =\sum_{\sigma\in S_n} \alpha_\sigma |\sigma\rangle
$$

Where $\alpha_\sigma$'s are the amplitudes of the superposition. Now consider any other column marked by $\pi$:

$$
U|\pi\rangle=\sum_{\sigma\in S_n} \beta_\sigma |\sigma\rangle
$$

Clearly, $|\pi\rangle=R(\pi)|123\ldots n\rangle$, and since $[U,R(\pi)]=0$:

$$
\sum_{\sigma\in S_n} \beta_\sigma |\sigma\rangle= \sum_{\sigma\in S_n} \alpha_\sigma |\sigma\circ \pi\rangle,
$$

\noindent which is the desired permutation of columns, and in terms of entries $\beta_{\sigma\circ \pi}=\alpha_\sigma$.

\end{proof}

The same conclusion can be made for the composition of $Y$ operators. Let $G$ be the group of unitary operators that can ever be generated by the compositions of $X$ operators. While the unitary group $U(\mathcal{H})$ is a Lie group of dimension $n!^2$, as the corollary of the above theorem $G\subset U(\mathcal{H})$ as a Lie group, has dimension $n!$ which is polynomially smaller than $n!^2$, suggesting that $G$ is not a dense subgroup of the unitary group.

In the definition of $X (\cdot, \cdot)$ and $Y(\cdot, \cdot)$ operators, the angle $\theta$ is independent of the labels that are being swapped. In fact the property observed in theorem ~\ref{ind} was a consequence of this independence. Therefore, we introduce another local unitary, $Z(\tilde{\theta},k)$, wherein the transposition angles depend on the color of the labels. Here $\tilde {\theta} = \{\theta_{ij}\}$ is a list of angles, one element per each $i\neq j\in [n]$. By definition $Z(\tilde{\theta},k)$ acts on the labels $|ab\rangle$ in the locations $k$ and $k+1$ with the following map:

$$
Z(\tilde{\theta},k)|ab\rangle =\cos \theta_{ab} I + i\sin \theta_{ab} L_{(a,a+1)}
$$

If we assume real valued angles with $\theta_{ij}=\theta_{ji}$, then the operator $Z$ becomes unitary. Clearly, the $X$ operators are the special case of the $Z$ operators. In order to see this, consider any basis $|\sigma\rangle, \sigma \in S_n$, and suppose $\sigma(k)=a, \sigma(k+1)=b$ then:

\begin{eqnarray*}
&Z&^\dagger (\tilde{\theta},k)Z(\tilde{\theta},k)|\sigma\rangle=\\
&(&\cos \theta_{ab}- i \sin \theta_{ab} L_{(k,k+1)})(\cos \theta_{ab}+ i \sin \theta_{ab} L_{(k,k+1)})|\sigma\rangle=|\sigma\rangle
\end{eqnarray*}

Now the composition of $Z$ operators give rise to a subgroup whose dimension can exceed $n!$. As is going to be clarified later in section \ref{lowerbounds}, this variant of of the ball permuting gates leads to a possibly more powerful model of computation.

 We can also define $W(\tilde{\theta},k)$ as an analogue of the $Z$ operators. Such a $W$ map is according to the following:

$$
|\ldots \overset{a}{k}\ldots \overset{b}{k+1}\ldots \rangle \rightarrow\cos \theta_{a b}|\ldots \overset{a}{k}\ldots \overset{b}{k+1}\ldots \rangle +i\sin \theta_{a b} |\ldots \overset{a}{k+1}\ldots \overset{b}{k}\ldots \rangle.
$$

The superscripts demonstrate the location of the labels. However, $W$ is not a dual to $Z$, like $X$ and $Y$'s, since these operators do not commute in general. In general, we can confirm the following properties for the discussed operators:

\begin{itemize}
\item $X^\dagger (\theta, k)=X^{-1} (\theta, k)=X (-\theta, k)$, and the similar relations are true for the $Y$ and $Z$ operators.
\item $[X(\theta, k), X(\theta',k')]\neq 0$ if and only if $|k-k'|=1$, this is also true for the $Y$ operators. 
\item $[Z(\tilde{\theta}, k), Z(\tilde{\theta'},k')]\neq 0$ if and only if $|k-k'|=1$, this is also true for the $W$ operators.
\item $[X(\theta, k), Y(\theta',k')] = 0$ for all values of $k, k'$ and $\theta, \theta'$.
\item $[Z(\tilde{\theta}, k), R(\pi)]$ can be nonzero. Therefore if $U$ is a composition of $Z$ operators, the columns of can be different modulo permutation.
\end{itemize}

\subsection{The Quantum Yang-Baxter Equation}
\label{YangBaxter}

We can discuss the restriction on the angles of the $X$ operators in such a way that they respect Yang-Baxter equation (YBE) of three particles. Therefore, this restricted version can capture the scattering matrix formalism of particles on a line. A solution to the YBE in the scattering models is based on the amplitudes which depend only on the initial state and momenta of the particles. The aim of this section is to justify that the only non-trivial solution to the YBE on the Hilbert space of permutations is the one in which the amplitudes are selected according to the velocity parameters.

Given a vector space $\mathbb{V}^{\otimes n}$, let $H_{ij}$ for $i<j \in [n]$ be a family of two-local operators in $GL(\mathbb{V}^{\otimes n})$ such that each $H_{ij}$ only affects the $ij$ slot of the tensor product, and acts trivially on the rest of the space. Then, $H$ is said to satisfy the parameter independent YBE if they are constant and:

$$
H_{ij} H_{jk} H_{ij}= H_{jk} H_{ij} H_{jk}
$$

Sometimes, we refer to the following as the YBE:

$$
(H\otimes I) (I\otimes H) (H\otimes I)=(I\otimes H) (H\otimes I) (I\otimes H)
$$

Both sides of the equation act on the space $\mathbb{V}\otimes \mathbb{V}\otimes \mathbb{V}$, and $H\otimes I$ acts effectively on the first two slots, and trivially on the other one. Similarly, we can define a parameter dependent version of the YBE, wherein the operator $H:\C \rightarrow GL(\mathbb{V}\otimes \mathbb{V})$ depends on a scalar parameter, and $H$ is said to be a solution to the parameter dependent YBE is according to:

$$
(H(z_1)\otimes I) (I\otimes H(z_2)) (H(z_3)\otimes I)=(I\otimes H(z'_1)) (H(z'_2)\otimes I) (I\otimes H(z'_3))
$$

for some $z_1, z_2, \ldots, z'_3$. We are interested in a solution of parameter dependent YBE with $X(\cdot, \cdot )$ operators. For simplicity of notations, in this part, we use the following operator:

$$
H(z , k):=\dfrac{1}{\sqrt{1+z^2}}+\dfrac{i z}{\sqrt{1+z^2}} L_{(k, k+1)} = X(\tan^{-1} (z), k)
$$

instead of the $X$ operators. The following theorem specifies a solution to the parameter dependent YBE:

\begin{theorem}
Constraint to $z_1 z_2 \ldots z'_3\neq 0$, the following is the unique class of solutions to the parameter dependent YBE, with the $H (\cdot, \cdot) = X (\tan^{-1}(\cdot), \cdot )$ operators:

$$
H(x,1) H(x+y,2) H(y,1)=H(y,2) H(x+y,1) H(x,2),
$$

for all $x,y \in \mathbb{R}$.

\end{theorem}

\begin{proof}
We wish to find the class parameters $z_1, z_2, \ldots, z'_3$ such that the following equation is satisfied:

\begin{equation}
H(z_1,1) H(z_2,2) H(z_3,1)=H(z'_1,2) H(z'_2,1) H(z'_3,2)
\label{PDYBE}
\end{equation}

It is straightforward to check that if $z_1=z'_3$, $z_3=z'_1$ and $z_2=z'_2=z_1+z_2$, then the equation is satisfied. We need to prove that this is indeed the only solution. Let:

$\Gamma:= \sqrt{\dfrac{(1+z'^2_1)(1+z'^2_2)(1+z'^2_3)}{(1+z^2_1)(1+z^2_2)(1+z^2_3)}}.$

If equation ~\ref{PDYBE} is satisfied, then the following are equalities inferred:

\begin{eqnarray*}
&1)&\hspace{3mm}\Gamma . (1- z_1 z_3) = (1- z'_1 z'_3)\\
&2)&\hspace{3mm}\Gamma . (z_1 + z_3) = z'_2\\
&3)&\hspace{3mm}\Gamma . z_2 = (z'_1+z'_3)\\
&4)&\hspace{3mm}\Gamma . z_1 z_2 = z'_2 z'_3\\
&5)&\hspace{3mm}\Gamma . z_2 z_3 = z'_1 z'_2\\
&6)&\hspace{3mm}\Gamma . z_1 z_2 z_3 = z'_1 z'_2 z'_3\\
\end{eqnarray*}

Suppose for now that all of the parameters are nonzero; We will take care of these special cases later. If so, dividing $6)$ by $5)$ and $6)$ by $4)$ reveals:

\begin{eqnarray}
\nonumber
z_1&=&z'_3\\
z_3&=&z'_1.
\label{blahh}
\end{eqnarray}

Again suppose that $z_1 z_3\neq 1$ and $z'_1 z'_3\neq 1$. Then using the equivalences of ~\ref{blahh} in $2)$, one gets $\Gamma=1$, from $2)$ and $3)$:

$$
z_2=z'_2=z_1+z_3,
$$

\noindent which is the desired solution. Now suppose that $z_1 z_3 = 1$. This implies also $z'_1 z'_3 = 1$. Using these in $6)$  one finds $\Gamma . z_2 =z'_2$ and substituting this in $2)$ and $3)$ reveals $\Gamma=1$ as the only solution, and inferring from equations $2), 3), \ldots, 6)$ reveals the desired solution.
\end{proof}

If one of the parameters is indeed $0$, we can find other solutions too, but all of these are trivial solutions. The following is the list of such solutions:

\begin{itemize}
\item If $z_1=0$, then:
\begin{itemize}
\item either $z'_3=0$, which implies $z_3=z'_2$ and $z_2=z'_1$
\item or $z'_2=0$ that implies $z_3=0$ and $\tan^{-1}(z_2)=\tan^{-1}(z'_1)+\tan^{-1}(z'_3)$.
\end{itemize}
\item If $z_2=0$ then $z'_1=z'_3=0$ and $\tan^{-1}(z,_2)=\tan^{-1}(z_1)+\tan^{-1}(z_3)$.
\item If $z_3=0$, then:
\begin{itemize}
\item either $z'_1=0$, which implies $z_1=z'_2$ and $z_2=z'_3$
\item or $z'_2=0$ that implies $z_1=0$ and $\tan^{-1}(z_2)=\tan^{-1}(z'_1)+\tan^{-1}(z'_3)$.
\end{itemize}
\end{itemize}

The solutions corresponding to $z'_j=0$ are similar, and we can obtain them by replacing the primed rapidities with the unprimed rapidities in the above table. There is another corresponding to the limit $z_j \rightarrow \infty$:

$$
X(0,1) X(0, 2) X(0, 1)=X(0,2) X(0, 1) X(0, 2)
$$

Which corresponds to the property $L_{(1,2)}L_{(2,3)}L_{(1,2)}=L_{(2,3)}L_{(1,2)}L_{(2,3)}$ of the symmetric group. From now on, we use the following form of the $H$-matrices:

$$
H(v_1, v_2, k) := H(v_1 -v_2, k),
$$

for real parameters $v_1, v_2$, and the YBE is according to:

$$
H(v_1, v_2, 1) H(v_1, v_3, 2) H(v_2, v_3, 1)=H(v_2, v_3, 2) H(v_1, v_3, 1)H(v_1, v_2, 2).
$$

The parameters $v_j$ can be interpreted as velocities in the scattering model. We can now extend the three label Yang-Baxter circuit to larger Hilbert spaces. 

\begin{definition}
An $m$ gate Yang-Baxter circuit over $n$ labels is a collection of $n$ smooth curves $(x_1(s), s), (x_2(s), s)\ldots (x_n(s), s)$ where $s\in[0,1]$, with $m$ intersections, inside the square $[0,1]^2$, such that, $0<x_1(0)< x_2(0) <\ldots <x_n (0)<1$, and $x_i(1)$ are pairwise non-equal.

If $\sigma \in S_n$, and $x_{\sigma(1)}(1)<x_{\sigma(2)}(1)<\ldots<x_{\sigma(n)}(1)$, then $\sigma$ is called the permutation signature of the circuit.

We say a Yang-Baxter circuit consists of line trajectories if all of the smooth curves are straight lines.
\end{definition}

Each such Yang-Baxter circuit can be equivalently represented by a set of adjacent permutations. When only line trajectories are considered, the circuit is related to the particle scattering models discussed in section \ref{models}. The permutation signature in this case is obtained by the momenta of the particles.

\begin{definition}
Let $C$ be a Yang-Baxter circuit of $m$ gates, each corresponding to a transposition $(k_t, k_t +1), t \in [m]$, and the permutation signature $\pi_t, t \in [m]$ at each of these gates. Then if one assigns a real velocity $v_j$ to each line, then the Yang-Baxter quantum circuit for $C$ is a composition of $H(\cdot, \cdot, \cdot)$ operators:

\begin{center}
$$
\hspace{-0.8cm}
 H(v_{\pi_m(k_m)}-v_{\pi_m(k_m)+1}, k_m) \ldots H(v_{\pi_2(k_1)}-v_{\pi_2(k_1+1)}, k_2) H(v_{k_1}-v_{k_1+1}, k_1).
$$
\end{center}

Each of these unitary $H$-matrices is a quantum gate.
\end{definition}

\subsection{Quantum Ball-Permuting Complexity Classes}
\label{qcomplexity}

Now that we have specified the quantum gate sets, we formalize these models according to classes of languages they recognize. In general, we are interested in a form of quantum computing where one starts with some initial state in $\C S_n$, applies a polynomial size sequence of ball permuting gates, and then in the end samples from the resulting probability distribution in the permutation basis of $S_n$. An initial state of the form $|123\ldots n\rangle$ sounds natural, however, for the reasons that we are going to see later, the computational power of the model seems to depend critically on the initial states. Therefore, we study the case where the model has access to arbitrary initial states separately:

\begin{definition}

Let $\XQBALL$ be the class of languages $L \subseteq \{0,1\}^\star$ for which there exists a polynomial time Turing machine $M$ which on any input $x \in \{0,1\}^\star$, outputs the description of a ball permuting quantum circuit $C$ as a composition of $X$ operators, and the description a subset $P \subseteq S_n$ of permutations, such that if $x\in L$, then the probability that the sampled permutation from $C$ is in $P$ at least $1/2 + \dfrac{1}{\Poly(n)}$, and otherwise it is at most $1/2- \dfrac{1}{\Poly(n)}$. Also, define $\YQBALL$ and $\ZQBALL$ similarly with the ball permuting circuits as the composition of $Y$ and $Z$ operators, respectively. Define $\HQBALL$ in the same way with the $H$ operators according to the Yang-Baxter circuits.
\end{definition}

In the definitions we did not specify the initial state of the ball permuting circuits, and we will specify this point whenever we mention the complexity classes.

\section{Upper-bounds}
\label{Upper-bounds}

Some of these models are the special cases of the others and therefore following containments are immediate:

$$
\HQBALL\subseteq \XQBALL \subseteq \ZQBALL
$$

Clearly, models with arbitrary initial states immediately contain their corresponding model with initial state $|123\ldots n\rangle$. In order to see all the containments in $\BQP$, it is sufficient to prove that $\ZQBALL \subseteq \BQP$.

\begin{theorem}
$\ZQBALL \subseteq \BQP$.
\end{theorem}

\begin{proof}
Let $L\in \ZQBALL$. Then, on any input $x\in \{0,1\}^\star$, there exists a polynomial time Turing machine that outputs the description of a $Z$ ball permuting circuit. We simulate the Hilbert space of permutations with bits, by just representing each label of $[n]$ with its $\lceil \log n \rceil$ long binary representation. Therefore, we use $n \lceil \log n \rceil$ bits to encode the permutations of $S_n$. The computation consists of three steps: at first we should simulate the initial state quantum states over binary bits, then we need to simulate the $Z$ operators, and in the end we need to demonstrate how to sample from the output states.

$1)$ Initialization: the $\BQP$ quantum circuit first applies enough not gates to the $|0\rangle^{\tensor n\lceil \log n \rceil}$ to prepare the encoded initial state $|123\ldots n\rangle$ with binary representations.

$2)$ Evolution: it is sufficient to show how to simulate one of the $Z(\tilde{\theta})$ operators on two labels. The list $\tilde{\theta}$ consists of coefficients $\theta_{ij}$. So for each pair of indices $i<j$ we add a control ancilla bit. We initialize all of the ancilla bits with zeros. We first apply enough controlled operations to the binary encodings of the $k$ and $k+1$ slots, to flip the $i,j$ control bit if and only if the contents of $k$ and $k+1$ slots are $i$ and $j$, then controlled with the $i,j$ control bit we apply the unitary $\cos \theta_{ij}+i \sin \theta_{ij} S$, where $S$ is the operator which swaps all the bits in the $k$ slot with all the bits in $k+1$. Notice the operator $S$ acts on at most $O(\log n)$ qubits, and because of the Solovay-Kitaev theorem, it can be efficiently approximated by a quantum computer. We continue this for all of the indices $i<j$. Since we are using several ancilla bits, we need to uncompute their contents, at the end of each $Z$ simulation. Therefore, a $Z$ ball permuting quantum circuit can be simulated by a qubit quantum circuit with polynomial (in $n$) blow-up in its size, and a Hilbert space consisting of $O(n\log n + m n^2)= O(m. n^2)$, where $m$ is the size of the original circuit.

$3)$ Measurement: at the end of the computation we only need to measure the output bits, and interpret them as a permutation. Therefore the output of the $Z$ ball permuting circuit can be sampled efficiently. We can also use enough controlled operations to flip a single bit if and only if the $n \lceil \log n\rceil$ bits encode the identity permutation.
\end{proof}

This readily demonstrates that all of the discussed models are contained in $\BQP$. 

Next, we use the following theorem to prove that the models with $X$ and $Y$ operators are equivalent:

\begin{theorem}
Let $G$ and $G'$ be the unitary groups generated by $X$ and $Y$ operators, respectively. Then $G\cong G'$.
\label{isom}
\end{theorem}

\begin{proof}
We show an isomorphism $T: G\rightarrow G'$, as a linear map, with $T (L_\sigma) = R_{\sigma^{-1}}$, and the (linear) inverse $T^{-1}: G'\rightarrow G$, with $T^{-1} (R_\sigma) = L_{\sigma^{-1}}$. Let $U$ be any element in $G$, then $U$ can be decomposed as a sequence of $X$ operators $U_m=X(\theta_m, k_m) X(\theta_{m-1}, k_{m-1})\ldots X(\theta_1, k_1)=:\sum_{\sigma\in S_n} \alpha_{\sigma} L_\sigma$, we need to prove that $T(U_m) \in G'$. We use induction to show that $T(U)$ is simply $Y(\theta_m, k_m) Y(\theta_{m-1}, k_{m-1})\ldots Y(\theta_1, k_1)\in B$. Clearly, $T(X(\theta_1, k_1))=Y(\theta_1, k_1)$, since the inverse of each transposition is the same transposition. For $t<m$, let $ U_t= X(\theta_t, k_t) X(\theta_{t-1}, k_{t-1}) $ $ \ldots X(\theta_1, k_1))=:\sum_{\sigma\in S_n} \alpha'_{\sigma} L_\sigma$. By induction hypothesis suppose that $T(U_t)=: V_t =Y(\theta_t, k_t)$ $ Y(\theta_{t-1}, k_{t-1})\ldots Y(\theta_1, k_1)=\sum_{\sigma\in S_n} \alpha'_{\sigma} R_{\sigma^{-1}}$. For simplicity let $c:= \cos \theta_{t+1}$ and $s:= \sin \theta_{t+1}$ and the transposition $k:= (k_t, k_t +1)$, then: 

\begin{eqnarray*}
U_{t+1}&=&X(\theta_{t+1}, k_{t+1}) U_t= \sum_{\sigma\in S_n} c \alpha'_{\sigma} L_{\sigma} + is \alpha'_{\sigma} L_{k\circ \sigma}\\
&=&\sum_{\sigma\in S_n} (c \alpha'_{\sigma} + is \alpha'_{k\circ \sigma}) L_{\sigma}
\end{eqnarray*}

and,

\begin{eqnarray*}
V_{t+1}&=&Y(\theta_{t+1}, k_{t+1}) V_t= \sum_{\sigma\in S_n} c \alpha'_{\sigma} R_{\sigma^{-1}} + is \alpha'_{\sigma} R_{\sigma^{-1}\circ k^{-1}}\\
&=&\sum_{\sigma\in S_n} c \alpha'_{\sigma} R_{\sigma^{-1}} + is \alpha'_{\sigma} R_{(k\circ\sigma)^{-1}}=\sum_{\sigma\in S_n} (c \alpha'_{\sigma} + is \alpha'_{k\circ \sigma}) R_{\sigma^{-1}}\\
&=& T(U_{t+1})
\end{eqnarray*}
\end{proof}

As a corollary the following is understood:

\begin{corollary}
$\XQBALL=\YQBALL$.
\end{corollary}

\begin{proof}
We use the isomorphism result of theorem ~\ref{isom}: suppose that after application of several gates $X(\theta_m, k_m) X(\theta_{m-1}, k_{m-1})\ldots X(\theta_1, k_1)$ to the state $|123\ldots n\rangle$ one obtains the quantum state $|\psi\rangle =\sum_{\sigma\in S_n} \alpha_\sigma |\sigma\rangle$, then the application of the corresponding $Y$ operators $Y(\theta_m, k_m) Y(\theta_{m-1}, k_{m-1})\ldots Y(\theta_1, k_1)$ to the same initial state, one obtains $|\psi\rangle =\sum_{\sigma\in S_n} \alpha_\sigma^{-1} |\sigma\rangle$, where $\sigma^{-1}$ is the inverse of the permutation $\sigma$. Any $X$ computation can be deformed in a way that in the end the computation just reads the amplitude corresponding to the identity permutation, \i.e., $\alpha_{123\ldots n}$. Since the inverse of the identity permutation is identity itself, $X$ computation can be simulated by a $Y$ computation, by just applying the same quantum circuit with $Y$ operators and read the identity amplitude in the end. A similar reduction also works from $X$ to $Y$ computations.
\end{proof}

The last observation of this section is that if we constraint the gates from ball permuting circuit to satisfy the Yang-Baxter equation, then the set of unitary gates that can be ever generated constitute a small dimensional manifold.

\begin{theorem}
Let $Q_n$ be the Lie group generated by \textit{planar} Yang-Baxter quantum circuits over $n$ labels, then $Q_n$ as a manifold is isomorphic to the union of $n!$ manifolds, each with dimension at most $n$.
\label{YBnonu}
\end{theorem}

\begin{proof}
Fix the velocities $v_1, v_2, \ldots, v_n$. The idea is to demonstrate an embedding of the group generated with these fixed velocities into the symmetric group $S_n$.  Consider any two planar Yang-Baxter quantum circuits $C$ and $C'$, with permutation signatures $\sigma$ and $\tau$, respectively. We show that if $\sigma=\tau$, then $C=C'$.

The underlying circuit of $C$ corresponds to a sequence of transpositions $k_1, k_2, \ldots , k_M$ and $C'$ corresponds to another sequence $l_1, l_2, \ldots , l_N$, such that $k_M\circ \ldots \circ k_2\circ k_1= \sigma$, and $l_N\circ \ldots \circ l_2\circ l_1= \tau$. Then the unitary operators $C$ and $C'$ can be written as a sequence of $H$ operators:

$$
C= H(z_M, k_M) \ldots H(z_2, k_2) H(z_1, k_1)
$$

and,

$$
C'= H(z'_N, l_N) \ldots H(z'_2, l_2) H(z'_1, l_1).
$$

Where the $z$ parameters are the suitable rapidities assigned to each two-particle gate based on the velocities $v_1, v_2, \ldots, v_n$, and the underlying Yang-Baxter circuits. From proposition ~\ref{sym} if two sequence of transpositions $k_M\circ \ldots \circ k_2\circ k_1$ and $l_N\circ \ldots \circ l_2\circ l_1$ amount to the same permutation, then there is a sequence of substitution rules among:

1) $b_i ^2 \Leftrightarrow e$

2) $b_i b_j \Leftrightarrow b_j b_i$  if $|i-j|>1$

3) $b_i b_i+1 b_i \Leftrightarrow b_i+1 b_i b_i+1$, for all $i\in [n-1]$

Such that if we start with the string $k_M\circ \ldots \circ k_2\circ k_1$ and apply a sequence of substitution rules, we end up with $l_N\circ \ldots \circ l_2\circ l_1$. All we need to do is to prove that the sequences of unitary gates are invariant under each of the substitution rules. The invariance under each rule is given in the below:

$1)$ If we apply two successive quantum transpositions on the labels $i, i+1$ we will end up with the identity operator. This follows from unitarity $H(z) H(-z)=I, \forall z\in \mathbb{R}$, and planarity of the circuits.

$2)$ Clearly $H(\cdot, i) H(\cdot, j) = H(\cdot, j) H(\cdot, i)$ for $|i-j|>1$, since these are $2$-local gates.

$3)$ This part also follows from the Yang-Baxter equation.

We can then start with the unitary $C= H(z_M, k_M) \ldots H(z_2, k_2) H(z_1, k_1)$  and apply the same substitution rules and end up with $C= H(z'_N, l_N) \ldots H(z'_2, l_2) H(z'_1, l_1)$.

Now let $Q_n (\sigma)$ be the Lie group corresponding to all Yang-Baxter quantum circuits with permutation signature $\sigma$. For each choice of velocities, there is exactly one unitary in this group, so $Q_n (\sigma)$ is locally diffeomorphic to $\mathbb{R}^n$, and $Q_n=\cup_{\sigma\in S_n} Q_n(\sigma)$.
\end{proof}

In the following we show that even with postselection in the end, a quantum planar Yang-Baxter circuit still generates a sparse subset of unitary group. In other words any attempt to prove post-selected universality for the particle scattering model without intermediate measurements will probably fail.

\begin{theorem}
The set of unitary operators generated by $\HQBALL$ with postselection in particle label basis in the end of computation, correspond to the union of (discrete) $n!^{O(1)}$ manifolds, each with linear dimension.
\label{YBnonu1}
\end{theorem}

\begin{proof}
We follow the proof of theorem ~\ref{YBnonu}. Consider the planar YB circuits on $n$ labels. If the input velocities are fixed, then the unitary operators generated by the model constitute a finite set of size at most $n!$. There are finite $n!^{O(1)}$ to do a postselection on the output labels of each circuit. So for each fixed set of velocities, the unitary matrices obtained by postselection and proper normalization still constitute a set of size $n!^{O(1)}$. Therefore, label the manifolds with the permutation signature of the circuits and the type of final postselection. Then the points in each of these manifolds are uniquely specified by $n$ velocity parameters, which is an upper-bound on the dimension for each of them.
\end{proof}

\noindent Notice that result of these theorems still hold if we allow the circuit models to start with arbitrary initial states.

\section{$\DQC 1$ Algorithm to Approximate Single Amplitudes of the Ball Permuting Model with Separable Initial State}
\label{trace}

The result of this section was obtained in collaboration with Greg Kuperberg. We find a $\DQC 1$ algorithm to compute the amplitudes of the $\XQBALL$ model in ball color basis within additive error.

The main result is the following:

\begin{theorem}
There is an efficient $\DQC 1$ algorithm which takes the description of a $\Poly(n)$ size ball permuting circuit $C$ over $\C S_n$ as its input, and outputs a complex number $\alpha$ such that $|\alpha - \langle 123 \ldots n | C | 123 \ldots n\rangle| \leq \dfrac{1}{\Poly(n)}$, with high probability.
\label{mainDQC1}
\end{theorem}

Here, efficient means polynomial time in $n$ in terms of the preprocessing which outputs the description of the circuit, and the size of output quantum circuit. Clearly, there are efficient reductions from the approximation of any other (permutation) amplitudes to the computation of $\langle 123 \ldots n | C | 123 \ldots n\rangle$. Also, for scattering models of $1+1$ dimensions, in the case where all particles are distinguishable there is an efficient $\DQC 1$ computation to obtain additive error approximation to the amplitudes:

\begin{corollary}
Additive approximation to the amplitudes in the models of section \ref{models} can be obtained within $\DQC 1$.
\end{corollary}

Usually, basing quantum computation on single amplitudes is hopeless, and one seeks for efficient sampling from the output distributions instead. But the $\DQC 1$ computation in theorem \ref{mainDQC1} only approximates single amplitudes, and this does not immediately imply efficient sampling. The existence of an efficient sampling from ball permuting circuits with $\DQC 1$ computation is unknown. Moreover, additive approximations can be problematic: if the amplitudes are exponentially small for generic circuits, outputting $0$ all the time will provide a good additive approximation. Whether or not this situation happens in this case is unknown.
 
The theorem is proved in three steps. First, in lemma \ref{lem1} it is motivated that for ball permuting circuits the computation of single amplitudes can be reduced to the computation of (normalized) traces. Next, we borrow a result of \cite{shor2008estimating} which provides a reduction from additive approximation of traces for unitary matrices to $\DQC 1$ computations. Finally, in the third step, by some careful analysis it is shown that the $\DQC 1$ reduction of the second step is an efficient one. The main idea for this step is to use a compressed encoding permutations with binary bits.

The amplitudes in ball permuting circuits are related to traces according to:

\begin{lemma}
For any ball permuting quantum circuit $C$, the trace $Tr(C)= n! \langle 123\ldots n | C | 123\ldots n \rangle$.
\label{lem1}
\end{lemma}

\begin{proof}
A quantum ball permuting circuit, by definition, consists of left permuting actions only which commute with right actions $R(\sigma)$ (relabeling) for any $\sigma \in S_n$. Thereby $\langle 123\ldots n | C | 123\ldots n \rangle = \langle 123\ldots n | R^{-1}(\sigma) C R(\sigma)| 123\ldots n \rangle= \langle \sigma | C | \sigma \rangle$. From this, $Tr(C)= \sum_{\sigma \in S_n} \langle \sigma | C | \sigma \rangle = n! \langle 123\ldots n | C | 123\ldots n \rangle$.
\end{proof}

Next, we formally mention the problem of trace approximation:

\begin{definition}
($\Trace$) given as input the $\Poly(n)$ size description of a unitary circuit $U$ as a composition of gates from a universal gate set over $n$ qubits, compute a complex number $t$ such that $|t-\dfrac{1}{2^n} Tr(U)| \leq \dfrac{1}{\Poly(n)}$, with high probability.
\end{definition}

The following theorem provides an efficient $\DQC 1$ algorithm for $\Trace$:

\begin{theorem}
(Jordan-Shor \cite{shor2008estimating}) $\Trace \in \DQC1$. \footnote{Moreover, the authors show that $\Trace$ is a complete problem for this class.}
\label{JordanShor}
\end{theorem}

Indeed, this theorem can be reformulated as: given an $n$ qubit unitary $U$, there is a round of $\DQC 1$ computation which reveals a coin which gives heads with probability $\dfrac{1}{2} + \dfrac{1}{2}\dfrac{\Re Tr (U)}{2^n}$. Also, there is another similar computation which gives a coin with bias according to the imaginary part of the normalized trace.

Using these observations, we are ready to present the proof of the main theorem:

\begin{proof}
(of theorem \ref {mainDQC1}) The objective is find an efficient algorithm which given a ball permuting circuit $C$ over $n$ labels, outputs the description of a unitary $U$ over $m=\Poly(n)$ qubits such that $\dfrac{1}{2^m}Tr(U) =\dfrac{1}{\Delta(n)} \langle 123\ldots n | C | 123 \ldots n \rangle$, with $\Delta(n)=\Poly(n)$. Given this reduction using theorem \ref{JordanShor} we deduce that the additive approximation of the amplitude can be obtained by rounds of $\DQC 1$ computation. 

The idea is to encode permutations with strings of bit. For the reasons that we are going to discuss later in this section, we need the encoding to be as compressed as possible. More precisely, we need an encoding that uses $O(\log (n! \Poly(n))$ number of bits. Moreover, in order to provide efficient quantum circuits, the code needs to be local, in the sense that in order to apply a swap, we just need to alter few ($O(\log n)$) bits. Otherwise, it is not clear if it is possible to implement quantum circuits efficiently.

Inspired by factorial number system and Lehmer code, we use an encoding of permutations that uses $\log n! + O(n)$ number of bits. Moreover, we confirm that in order to apply a swap, only $O( \log n)$ bits of the code should be altered.

The encoding of each permutation, $\sigma(1), \sigma(2),\ldots, \sigma(n)$ ($\sigma \in S_n$), is accomplished by a walk from root to each leaf of the following tree, $T_n$: consider a tree with its root located at node $0$, as we mark it to be distinct. Let node $0$ have degree $n$, with its children marked with numbers $1,2,3,\ldots, n$, from left to right. Denote these nodes by layer $1$. Let each node of layer $1$ have $n-1$ children, and label each child of node $i$ in layer $1$, by numbers $[n]-\{i\}$, in an increasing order from left to right. Construct the tree inductively, layer by layer: each node $k$ in layer $j$ have $n-j$ children, and the children labeled with numbers $[n]-L_k$. Where $L_k$ is the set of labels located on the path from node $0$ to node $k$. Therefore, nodes of layer $n$ have no children. The number of leaves of the tree is $n!$. For each leaf there is a unique path from root down to the leaf, and the indexes from top to down represent a permutation. This is because the indexes of each path are different from each other. Also each permutation $\sigma$ is mapped to a unique path in this tree: start from node $0$, pick the child with index $\sigma(1)$, then among the children of $\sigma(1)$, pick the child with index $\sigma(2)$ and so on. Therefore, this establishes a one-to-one map between the paths on $T_n$ and permutations of labels in $[n]$.

The next step is to provide a one-to-one mapping from the paths on the graph to bit strings of length $\log n! + O(n)$. First, label the edges of $T_n$ by the following. For each node of degree $p$, with children labeled with $x_0< x_1< \ldots< x_{p-1}$, label the edge incident to $x_0$ by $0$, the edge incident to $x_1$ by $1$, and so on. Given these edge labels, The construction is simple: represent each path with the bit string $a_n a_{n-1} \ldots a_0$, where $a_j$ is a bit string of length $\lceil \log j\rceil =\log j + O(1)$, is the binary representation of the label of the edge used in the $j$'th walk.

The final step is to show that in order to apply a swap on this encoding one needs to alter only $O(\log n)$ bits. Suppose that the permutations $\sigma=\sigma(1), \sigma(2), \ldots, \sigma(k), \sigma(k+1), \ldots, \sigma(n)$ and $\pi=\sigma(1), \sigma(2), \ldots, \sigma(k+1), \sigma(k), \ldots, \sigma(n)$ are represented by the binary encoding $X=a_1, a_2, \ldots a_k, a_{k+1}\ldots, a_n$ and $Y=b_1, b_2, \ldots b_k, b_{k+1}\ldots, b_n$, respectively. Clearly, $\pi$ can be obtained from $\sigma$ by swapping the element $k$ and $k+1$. Notice that $a_1=b_1, a_2=b_2, \ldots , a_{k-1}=b_{k-1}$. This is because the corresponding path representations of the two permutations on $T_n$ walk through the same node at the $k-1$'th walk. Also $a_{k+2}=b_{k+2}, \ldots , a_n=b_n$. This is because the subtrees behind the $k+2$'th layer nodes in the two paths are two copies of the same tree, since their nodes consist of same index sets. Therefore, $X$ and $Y$ differ only at $a_k, a_{k+1}$ and $b_k, b_{k+1}$ substrings. As a consequence of these observations, the bit-string codes for two permutations that differ in adjacent labels only, are different in $O (\log n)$ bits.

If in the description of $C$ nonadjacent swaps are implemented, we simulate these swaps by adjacent ones. We construct $U$ by approximating each adjacent $X$ gate in $C$.  Each such gate alters $O(\log n)$ bits and because of the Solovay-Kitaev theorem, there exists a $\Poly(n, \log 1/\epsilon)$ size circuit that approximates each $X$ gate within error $\epsilon$.
\end{proof}

In the proof above we mentioned that a compressed encoding of permutations is necessary to establish the result with this approach. Here, we mention an example of a slightly less compressed encoding which makes the situation hopeless: represent each number in $[n]$ with $\log n$ bits. Simulate each $X$ gate in $C$ with a qubit quantum circuit which swaps the encoded numbers in a superposition. Such a quantum circuit can be efficiently obtained from a universal gate set. Again the reason for the existence of such efficient circuit is the Solovay-Kitaev theorem. Let $U$ be the composition of these unitary circuits. The objective is to do a $\DQC 1$ computation to obtain an approximation to $Tr(U)/D$, where $D$ is the dimension of the Hilbert space that $U$ is acting on. Among the summands of $Tr(U)$ there are terms like $\langle b | U | b\rangle$, where $b$ is a string of bits with repeated labels (for example $|1 1 2 3 4\rangle$). In order to avoid the contribution of these terms, use $\dfrac{n(n-1)}{2}$ more (flag register) qubits, $f_{ij}, i< j \in [n]$. Then we add another term $T$ to the quantum circuit to obtain $U T$. The role of $T$ is simply to modify the flag registers in a way that the contribution of unwanted terms in the trace becomes zero: for each $i<j \in [n]$, using sequences of $CNOT$ gates, $T$ compares the qubits $(i-1) \lceil \log n\rceil +1$ to $i.\lceil \log n\rceil$ with the qubits $(i-1) \lceil \log n\rceil +1$ to $i.\lceil \log n\rceil$, bit by bit, and applies $NOT$ to the register $f_{i,j}$ if the corresponding bits are all equal to each other. Then $U T$ is fed into the $\Trace$ computation. Let's see what approximation to $ \langle 123\ldots n | C| 123 \ldots n \rangle$ we get in this case. Let $N:=n\lceil \log n\rceil+n(n-1)/2$. The trace $Tr(U)=\sum_{x\in \{0,1\}^N} \langle x|U|x \rangle$. Given the described construction, the term $\langle x | U | x\rangle= \langle \sigma | C |\sigma \rangle$, if and only if the label part of $x$ is the correct encoding of the permutation $\sigma$, and if $x$ is not a correct encoding of a permutation it gives $0$. There are $2^{n(n-1)/2}$ strings like $x$ which encode $\sigma$ correctly, therefore:

$$
1/2^N Tr(U)= \dfrac{2^{n(n-1)/2}}{2^N} Tr (C)= \dfrac{n!}{2^{n \lceil \log n \rceil}} \langle 123\ldots n | C | 123 \ldots n \rangle.
$$

\noindent This is problematic, since the coefficient $\dfrac{n!}{2^{n \lceil \log n \rceil}}$ can be exponentially small, and thereby polynomial iterations of the $\DQC1$ computation wouldn't reveal any information about the desired amplitude. Taking a close look at the this coefficient, it is observed that for any encoding of permutations with bit-strings, the proportionality constant appears as:

$$
\dfrac{n!}{dim{V}}
$$

\noindent where $V$ is the dimension of the Hilbert space that is used to encode permutations in it. In the latter example, we used $O (n \log n )$ bits to encode permutations of $n$ labels, which is less compressed than $\log n! + O(n)$, used in the proof of theorem \ref{mainDQC1}.

We believe that this result can be generalized to a wide variety of quantum models based on group algebras. More precisely, consider a group $G$, with identity element $e$. Then construct the Hilbert space $\mathcal{H}_G$ with orthonormal basis $\{ |g\rangle : g \in G \}$. Let $\C G$ be the (left) group algebra, and $x \in \C G$. Then $\langle e | x | e \rangle= \dfrac{1}{|G|} Tr (x)$, which is a reduction to the computation of normalized trace. The only issue is with the encoding of the bases of the Hilbert space with local binary strings.

\section{Some Lower-bounds}
\label {lowerbounds}
\subsection{$\ZQBALL=\BQP$}

In this part we use a simple encodings of qubits using labels $1,2,3\ldots, n$, and the $Z$ operators to operate on them as single and two qubit gates. More specifically, we prove that using a sequence of $Z$ operators, one can encode any element in the special orthogonal group. For an example of encoded universality see \cite{gottesman2001encoding, bacon2001encoded}. We encode each qubit using two labels. Given two labels $a < b$ we define the encoded (logical) qubits as:

$$
|0\rangle := | a b \rangle
$$

and,

$$
|1\rangle := i |b a\rangle.
$$

\noindent Using simple $X(\theta,1)$ we can apply arbitrary rotation of the following form:

$$
|0\rangle \rightarrow \cos \theta |0\rangle+\sin \theta |1\rangle
$$

and,

$$
|1\rangle \rightarrow \cos \theta |1\rangle-\sin \theta |0\rangle.
$$

\noindent We are dealing with orthogonal matrices which are represented over the field or real numbers. Using the $Z$ operators, we can discuss a controlled swap of the form:

$$
S(i,j,k,l) := Z(\pi/2 \delta_{i,j}, k , l).
$$

In simple words, $S(i,j,k,l)$ applies the swap $i L_(k,l)$, on the $k$ and $l$'th labels if and only if the content of these label locations are $i$ and $j$ ( $j$ and $i$). We can also extend it to the following form:

$$
S(\{(i_1,j_1)^{s_1}, (i_2,j_2)^{s_2}, \ldots, (i_t,j_t)^{s_t}\},k,l) := Z(\pi/2 \delta_{i,j}, k , l).
$$

\noindent Where $s_m$ can be a symbol $\star$ or nothing. Given $(i_m, j_m)^\star$ in the list means that the swap $(i L_{(k, l)})^\dagger= - i L_{(k, l)}$ is applied if the content of $k$ and $l$ are $i_m$ and $j_m$. And given plain $(i_m, j_m)$ in the list means $i L_{(k, l)}$ if the content of $k$ and $l$ are $i_m$ and $j_m$.

Suppose that one encodes one qubit with labels $a < b$ and another one with $x < y$, we wish to find a unitary operator which applies a \textit{controlled not} on the two qubits, that is the following map:

\begin{eqnarray*}
|00\rangle &:=& | a,b, x,y\rangle \rightarrow  | a,b, x,y\rangle=|00\rangle\\
|01\rangle &:=& i| a,b, y,x\rangle \rightarrow  i| a,b, y,x\rangle=|01\rangle\\
|10\rangle &:=& i| b,a, x,y\rangle \rightarrow -| b,a, y,x\rangle=|11\rangle\\
|10\rangle &:=& -| b,a, y,x\rangle \rightarrow i | b,a, x,y\rangle =|10\rangle
\end{eqnarray*}

\noindent It can be confirmed that the following operator can do this:

$$
\hspace{-1cm}
C:= S (\{(a , x), (a , y)^\star\}, 1, 3) S (\{(a , x), (a , y)\}, 2, 3) S (\{(a , x), (a , y)\}, 1, 2)
$$

\noindent Given these two operators, one can simulate special orthogonal two-level systems, that is for each orthonormal $|\psi\rangle$ and $|\phi\rangle$ in the computational basis of $n$ qubits we can apply an operator which acts as:

$$
|\psi\rangle \rightarrow \cos \theta |\psi\rangle + \sin \theta |\phi\rangle 
$$

and,

$$
|\phi\rangle \rightarrow \cos \theta |\phi\rangle - \sin \theta |\psi\rangle 
$$

\section{Quantum Ball Permuting Model on Arbitrary Initial States}

As mentioned in theorem ~\ref{ind}, the columns of a ball permuting operator as a unitary matrix are all permutations of each other. Moreover, one can observe the following property:

\begin{lemma}
If $C$ is any composition of $X$ ball permuting operators over $\C S_n$, then $C|123\ldots n\rangle = |123\ldots n \rangle$ if and only if $C=I$.
\end{lemma}

\begin{proof}
One direction is clear, that is if $C=I$ then $C|123\ldots n\rangle = |123\ldots n \rangle$. In order to see the other direction, suppose that $C|123\ldots n\rangle = |123\ldots n \rangle$, then act $R(\sigma)$ for all $\sigma \in S_n$ to the both sides to obtain $R(\sigma) C|123\ldots n\rangle = R(\sigma) |123\ldots n \rangle$. Since any ball permuting circuit commutes with a right (relabeling) action, we get $C|\sigma\rangle = |\sigma \rangle, \forall \sigma \in S_n$, which readily implies $C=I$.
\end{proof}

The lemma states that the group of ball permuting operators that stabilize $|123\ldots n\rangle$ is trivial. Also, one can extend this to all states like $|\sigma\rangle$ for $\sigma \in S_n$, and indeed all states that are reachable by an $X$ ball permuting circuit from $|123\ldots n \rangle$. This suggests that the set of ball permuting gates is contained in a subgroup $SU(n!)$ as manifold of lower dimension. However, although this does not rule out encoded universality, because of the following corollary there is some evidence, supporting that even encoded $\BQP$ universality and therefore $\BQP$ universality is impossible for ball permuting circuits in this case:

\begin{corollary}
Let $|0_L\rangle=C_0 |123\ldots n\rangle$ and $|1_L\rangle=C_1 |123\ldots n\rangle$, with $C_0$ and $C_1$ being ball permuting circuits, be any logical encoding of a qubit. Then for any $N\geq 1$ and any string $x \in \{0,1\}^N$, a ball permuting unitary $U$ is a stabilizer of the encoded $|x\rangle$, if and only if $U=I$.
\label{nonu}
\end{corollary}

\begin{proof}
Again, one direction of the proof is clear, that is if $U=I$, then $U|x\rangle =|x\rangle$. Now suppose that $U|x\rangle=|x\rangle$. Then $U C_{x_1} \otimes C_{x_2} \otimes \ldots \otimes C_{x_N} |12,3,\ldots, n N\rangle=C_{x_1} \otimes C_{x_2} \otimes \ldots \otimes C_{x_N} |12,3,\ldots, n N\rangle$. Since $C_{x_j}$ are unitary ball permuting circuits, they have inverses, and $G_x := C_{x_1} \otimes C_{x_2} \otimes \ldots C_{x_N}$ is also a ball permuting gate and has an inverse $G^{-1}_x$ which is also a ball permuting circuit. Therefore, $G^{-1}_x U G_x |123\ldots n N \rangle=|123\ldots n N \rangle$. From lemma ~\ref{nonu}, $G^{-1}_x U G_x=I$ which implies $U=I$.
\end{proof}

\noindent In the proof, we used two facts: first that $C_0$ and $C_1$ have inverses, and that they commute with the relabeling $R$ operators. Indeed, the lemma applies to non-unitary $C_{0}$ and $C_{1}$, as long as they have the commuting and the inverse properties. we finally conclude the following general non-universality criterion:

\begin{corollary}
The result of corollary ~\ref{nonu} remains true if $C_0$ and $C_1$ have inverses and commute with the relabeling operators.
\end{corollary}

While computing with the $|123\ldots n\rangle$ initial state, results in a presumably week model, one can use other initial states to break the conditions of lemma ~\ref{nonu}. For example, if one allows the initial state $|\psi\rangle = \dfrac{1}{n!} \sum_{\sigma \in S_n} |\sigma \rangle$ then the application of any operator of the form $X(\theta_1, \cdot )X(\theta_2, \cdot )\ldots X(\theta_p, \cdot )$ results in the state $\exp(i(\theta_1+\theta_2+\ldots+ \theta_p)) |\psi \rangle$, so as long as the angles sum up to zero the operator stabilizes $|\psi\rangle$. Indeed, the projection $\dfrac{1}{n!}\sum_{\sigma \in S_n} R(\sigma)$ maps any state in $\C S_n$ to a state proportional to $|\psi\rangle$.

In this section we provide evidence for encoded universality of the ball permuting model, in the case where initial states other than $|123\ldots n\rangle$ are allowed. This requires three considerations: first, we need to find a subspace $V \subseteq \mathcal{H}$, that is invariant under ball permuting operators, and that also scales exponentially in $n$ in dimension. Secondly, we need to find a way of composing the $X$ operators to act densely in $SU(V)$. Thirdly, we need to find a way to sample from the output states in $V$, in the color basis, to extract nontrivial information about these states. In order to achieve the second goal, we find a reduction from a model that is already known to be $\BQP$ universal. See section .

\subsection{Theory of Decoherence Free Subspaces}
\label{dfs}

The theory of decoherence free subspaces was originally motivated by the following problem \cite{kempe2001theory}. Let $\mathcal{H}$ be a Hilbert space, and let $N$ be the set of Hamiltonians, as the unwanted noise interactions. The objective is to find a large enough subspace $V\subseteq\mathcal{H}$ that is unaffected by the noise operators. Such a subspace is called a decoherencce free subspace. It is tempting to find a set of local and feasible quantum Hamiltonians, $E$, which commute with $N$ and affect the decoherence free subspace only. Then universal quantum computation is possible if $E$ acts as a universal gate set on $V$.

Ideally, the Hilbert space is decomposable into two separate subsystems $\mathcal{H}=\mathcal{H}_1 \otimes \mathcal{H}_2$, and the action of $N$ affects $\mathcal{H}_2$ only, and acts trivially on $\mathcal{H}_1$. Then, any universal gate set acting on $\mathcal{H}_1$ can reliably do universal quantum computing. However, in general, $N$ can mix all local degrees of freedom of the available subsystems at the same time. Thereby, instead of subsystems, one can think about subspaces which mimic the structure of decoupled subsystems. Intuitively, this can be interpreted as subspaces of a Hilbert space simulating decoupled subsystems. 

In order to find such a decomposition, it is helpful to consider a larger structure, $A$, as the matrix algebra generated by operators of $N$, by matrix composition and scalar linear combinations. This is a vector space of matrices, with matrix multiplication as the vector on vector action. Then, under certain conditions, $A$ is isomorphic to the decomposition of smaller irreducible matrix algebras:

$$
A\cong\bigoplus^D_{\lambda=1} \bigoplus_{j=1}^{n_\lambda} M(d_\lambda)\cong \bigoplus_{\lambda} I(n_\lambda)\otimes M(d_\lambda).
$$

\noindent Here $\lambda$ enumerates the type of the matrix blocks. $M(d_\lambda)$ denotes the algebra of matrices with dimension $d_\lambda$, and $\bigoplus_{j=1}^{n_\lambda} M(d_\lambda)$ is the same algebra repeated for $d_\lambda$ times, and $n_\lambda$ is therefore the multiplicity of this block. This means that there is a change of basis, on which the element of $A$ acts as the block diagonal structure $M_1\times M_1 \times \ldots \times M_1 \times \ldots \times M_D \times M_D \times \ldots \times M_D$; each matrix $M_j$ is repeated for $n_j$ times in the product series. These basis states are indeed the desired subspaces, and the Hilbert space also decomposes accordingly:

$$
\mathcal{H}=\bigoplus_{\lambda} n_\lambda  X(d_\lambda) \cong \bigoplus_{\lambda} V(n_\lambda)\otimes X(d_\lambda).
$$

\noindent Here $n_\lambda X_{d_\lambda}$ means $n_\lambda$ isomorphic subspaces each with dimension $d_\lambda$, and therefore $\dim{\mathcal{H}}=\sum_{\lambda} n_\lambda d_\lambda$. Given these decompositions, the operators of each block $\lambda$ in the decomposition of $A$ acts non-trivially on $X(\lambda)$ only, and leaves the subspace $V(n_\lambda)$ unaffected. It remains to find operators that act only on the $V$ species and leave the $X$ parts unaffected. The most general such structure is the matrix algebra $B$ which commutes with all of $A$. $B$ has the unique decomposition according to:

$$
B\cong \bigoplus_{\lambda}  M(n_\lambda)\otimes I(d_\lambda).
$$

\noindent Let $E$ be a gate set in $B$. Then, universal quantum computation on a decoherence free subspace $V(n_\lambda)$ is translated to first zooming into a subspace $V(n_\lambda)\tensor |\psi\rangle$, for some $|\psi\rangle \in X(d_\lambda)$, and then denseness of $E$ in $SU(V(n_\lambda))$, and finally zooming out from $V(n_\lambda)$, by sampling bits of information from the output of computation.

\subsection{Representation Theory of the Symmetric Group}

Most of the mathematical review is borrowed from \cite{james1981representation} We are interested in two mathematical structures, the group algebra of the symmetric group $\C S_n$, and the unitary regular representation of the symmetric group. As it turns out, the two structures are closely related to each other, and also to the group generated by the ball permuting gates. Group algebra is an extension of a group to an algebra, by viewing the members of the group as linearly independent basis of a vector space over the field $\C$. Therefore, in addition to the group action an action of $\C$ on $S_n$ is needed, by the map $(\alpha, \sigma) \mapsto \alpha S_n$, and also addition of vectors in the usual sense. Therefore, a group algebra consists of all elements that can ever be generated by vector on vector composition and linear combination of vectors over $\C$.  Any element of $\C S_n$ can be uniquely written as $\sum_{\sigma \in S_n} \alpha_\sigma \sigma$, with $\C$ coefficients $\alpha_\sigma$. If we add a conjugation convolution $^\dagger$ with maps $\sigma^\dagger=\sigma^{-1}$, and $\alpha^\dagger = \alpha^\star$, then for any element $v\in \C S_n$, $v^\dagger v=0$, if and only if, $v=0$. In order to see this, let $v=\sum_{\sigma \in S_n} \alpha_\sigma \sigma$. Then, $v^\dagger v= \sum_\sigma |\alpha_\sigma|^2 e+ \ldots=0$. A zero on the right hand side implies zeroth of all the vector components, including the component along $e$, which implies $\alpha_\sigma=0$ for all $\sigma \in S_n$, and therefore $v=0$. Let $e$ be the identity element of $S_n$, consider an element $p\in \C S_n$ to be a projector if it has the property $p^2=p$. Two projectors $p$ and $q$ are called orthogonal if $p. q =0$. Then $(e-p)^2=e-p$ is also a projector, and also $p (e-p) =0$ are orthogonal projectors. $0$ is trivially a projector. Therefore, the group algebra decomposes as:

$$
\C S_n = \C S_n e = \C S_n (e-p) + p = \C S_n (e-p) \oplus \C S_n p.
$$

\noindent A projector is called minimal if it cannot be written as the sum of any two others projectors other than $0$ and itself. Let $p^\mu$ be a list of minimal projectors summing $\sum_\mu p^\mu=e$, then the decomposition of the group algebra into minimal parts is according to:

$$
\C S_n = \bigoplus_\mu \C S_n p^\mu.
$$

\noindent $p^\mu$ are known as Young symmetrizers, and we are going to mention them later.

A (finite) representation $\rho$ of a group $G$ is a homomorphism from $G$ to the group of isomorphisms of a linear space $: G \rightarrow GL(V, \C)$, for some vector space $V$. Let $g$ be any element of $G$, with its inverse $g^{-1}$, and $e$ and $1$ as the identity elements of $G$ and $GL(V,\C)$, respectively. Given the definition, $\rho (g^{-1})= \rho(g)^{-1}$, and $\rho(e)=1$ are immediate. One can observe that $\rho : G \rightarrow \{I\in GL(V,\C)\}$, is immediately a representation, and is called the trivial representation of $V$. A dual representation of $G$ is a homomorphism from $G$ into the group of linear maps $: V \rightarrow \C$. As we discussed before, this is called the dual space $V^\star$, and $V$ is viewed as the space of column vectors, then its dual space is a row space. For any vector spaces $V$ and $W$, the two can be combined into a larger linear structure, $V\otimes W^\star$, as the set of linear maps from $W$ to $V$. Let $M_1$ and $M_2$ be two elements of $GL(V,\C)$ and $GL(W,\C)$, respectively. Then, viewing $V\otimes W^\star$ as a vector space, the object $(M_1, M_2)$ acts on $x \in V\otimes W^\star$ with $M_1 x M^{-1}_2$. Then, if $M_1$ and $M_2$ are two representations of $G$ on $V$ and $W$, then $(M_1, M_2)$ is a representation of $G$ on $V\otimes W^\star$, as a vector space. Notice that the inverse on $M_2$ is needed in order to have $(M_1, M_2)$ act as a homomorphism. The dual representation $M$ of $V$ is then the representation on $\C\otimes V^\star$, when $M_2=M$, and $M_1$ is the one dimensional trivial representation. This is just saying that the dual representation $M^\star$ of $M$ on $V^\star$, maps $\langle \psi |$ to $\langle \psi| M(g^{-1})$, if we view the dual space as the usual row space. If we define an inner product as the action of the dual of a vector on itself, then $G$, as a representation, sends orthonormal basis to orthonormal basis. This suggests that every representation of a finite group is isomorphic to a unitary representation. That is, any non-unitary representation becomes unitary after a change of basis. Let $M$ be a representation on $V$. Then, we say $W\subseteq V$ is called stable under $M$, if for any $x\in W$, $M x \in W$. Then, $M$ restricted to $W$ is called a sub-representation.  A representation $M$ on $V$ is called an irreducible representation (irrep), if it has no stable subspaces other than $0$ and $V$. Two representations $M_1$ and $M_2$ on $V_1$ and $V_2$ are isomorphic if $M_1$ resembles $M_2$ after a suitable change of basis within $V_1$. Then, if $V$ is reducible, it can be decomposed as $V_1 \oplus V_2 \oplus \ldots \oplus V_n$, for $n>1$. Some of the sub-representations can be isomorphic, and the multiplicity of a sub-representation is the number of sub-representations isomorphic to it. Then, the isomorphic subspaces can be grouped together to $V\cong m_1 V_1 \oplus m_2 V_2 \oplus \ldots m_k V_k$. Then $\dim V= \sum_j m_j \dim V_j$. The structure of such decomposition is isomorphic to $\bigoplus_j V_j \otimes X_j$, where $X_j$ is the multiplicity space of $V_j$ and is a vector space of dimension $m_j$. Decomposition of a representation onto the irreducible ones is unique up to isomorphism and multiplicities and dimensionality of irreducible representations do not depend on the decomposition. Canonical ways to find a decomposition are also known.

The regular representation of $S_n$, also denoted by $\C S_n$, is the unitary representation of $S_n$ onto the usual Hilbert space $\C S_n$ spanned by the orthonormal basis $\{|\sigma\rangle : \sigma \in S_n\}$. It is well known that for any regular representation, the dimension of each irrep is equal to the multiplicity of the irrep, and therefore $\C S_n$ decomposes into irreducible representations of the following form:

$$
\C S_n \cong \bigoplus_\lambda V_\lambda \otimes X_\lambda,
$$

\noindent with $\dim X_\lambda = \dim V_\lambda =: m_\lambda$, and indeed $\sum_\lambda m^2_\lambda= n!$. Here $X_\lambda$ is again the multiplicity space, and $V_\lambda$ corresponds to each irrep. It is tempting to make a connection between the group algebra and regular representation of the symmetric group. As described earlier, $S_n$ can act on the Hilbert space $\C S_n$ in two ways; the left and right, $L, R: S_n \rightarrow U(\C S_n)$, unitary regular representation, with the maps $L(\sigma) |\tau\rangle = |\sigma \circ \tau\rangle$ and $R(\sigma) |\tau\rangle = |\tau\circ \sigma^{-1}\rangle$. Also, similar left and right structure can be added to the group algebra. Clearly, $L$ and $R$ representations commute, and it can be shown that the algebra generated by $L$ is the entire commutant of the algebra generated by $R$. Putting everything together, inspired by the theory of decoherence free subspaces, and the defined structures, one can show that the left ($A$) and right ($B$) algebras and the Hilbert space $\C S_n$ decompose according to:

$$
A\cong \bigoplus_\lambda M(m_\lambda)\otimes I(m_\lambda),
$$

$$
B\cong \bigoplus_\lambda I(m_\lambda) \otimes M(m_\lambda),
$$

and,

$$
\C S\cong \bigoplus_\lambda V(m_\lambda) \otimes X(m_\lambda).
$$

\noindent This is indeed a nice and symmetric structure. Indeed each irrep $V_\lambda$ is an invariant subspace of the $X$ operators, and it cannot be reduced further. It remains to demonstrate the structure of the irreps $\lambda$, and to study the action of $X$ operators on these subspaces.

The irreducible representations of the symmetric group $S_n$ are marked by the partitions of $n$. Remember that a partition of $n$ is a sequence of non-ascending positive numbers $\lambda_1 \geq \lambda_2 \geq \lambda_3 \geq \ldots \lambda_k$ summing to $n$, \i.e., $\sum_j \lambda_j = n$. The number of partitions of $n$ grows like $\exp \Theta(\sqrt{n})$. Each as described earlier each partition $\lambda= (\lambda_1, \lambda_2, \ldots, \lambda_k)$ is related to a diagram, called the Young diagram, which consists of $k$ horizontal rows of square boxes $r_1, r_2, \ldots, r_k$. The Young diagram is then created by paving the left-most box of $r_1$ to the left-most box of $r_2$, and so on. For a Young diagram $\lambda$, the dual diagram $\tilde{\lambda}$, is another Young diagram, whose rows are the columns of $\lambda$. A Young tableau $t^\lambda$ with the shape $\lambda$, is a way of bijective assigning of the numbers in $[n]$ to the boxes of $\lambda$. We will use $t^\lambda$ and simply $t$ with the shape $\lambda$ interchangeably. A permutation $\pi \in S_n$ can act on a Young tableau $t^\lambda$ by just replacing the content of each box to the its image under $\pi$, \i.e., if a box contains $j$, after the action of $\pi$ it will be replaced with $\pi(j)$. A tableau is called standard, if the numbers in each row and column are all in ascending orders. The number of standard tableau for each partition of shape $\lambda$ is denoted by $f^\lambda$.

Let $t$ be a tableau with shape $\lambda$. Define $P(t)$ and $Q(t)\subseteq S_n$ to be sets of permutations that leave each row and column invariant, respectively. Then the projectors of the $\C S_n$ group algebra are according to the Young symmetrizers, one for each standard tableau:

$$
p^t = \frac{1}{f^\lambda} \sum_{\pi \in C(t)}\sum_{\sigma \in R(t)} sgn(\pi) \pi \circ \sigma.
$$

These subspaces correspond to all of the irreducible invariant subspaces of $S_n$. The dimension for each of these subspaces is the number of standard tableaus of each partition, and it is computable using the hook lengths. The hook of each box in a partition of shape $\lambda$ is consists of the box itself along with all boxes below and at the right of the box. The hook length of each box is the number of boxes contained in that hook, and the hook length $h^\lambda$ of the shape $\lambda$ is the multiplication of these numbers for each box. Then, the dimension of the irrep corresponding to $\lambda$ is according to $f^\lambda=n!/ h^\lambda$. 

\subsection{The Young-Yamanouchi Basis}

In order to talk about quantum operations orthonormal basis for the discussed subspaces are needed. It would be nice if we have a lucid description of the basis, in a way that the action of $X$ operators on these subspaces is clear. Moreover, we seek for an inductive structure for the orthonormal basis of the irreps that is adapted to the nested subgroups $S_1\subset S_2 \subset \ldots \subset S_n$. By that we mean states that are marked with quantum numbers like $|j_1, j_2, j_3, \ldots, j_k\rangle$, such that while elements of $S_n$ affect all the quantum numbers, for any $m_1< n$, elements of $S_n$ restricted to the first $m_1$ labels affects the first $k_1$ quantum numbers only, and act trivially on the rest of the labels. Also, for any $m_2< m_1<n$, the elements of $S_n$ restricted to the first $m_2$ labels affect the first $j_2< j_1 < k$ quantum numbers only, and so on.

Fortunately, such a bases exist, and are known as the subgroup adapted Young-Yamanouchi (YY) bases \cite{james1981representation}. These bases are both intuitive and easy to describe: for any partition of shape $\lambda$, mark an orthonormal basis with the standard Young Tableaus of shape. Agree on a lexicographic ordering of the standard tableaus, and denote these basis corresponding to the partition $\lambda$, by a $\{|\lambda_j\rangle\}_{j=1}^{f^\lambda}$. Denote the action of a swap $(i, j)$ on $|\lambda_l\rangle$ by $|(i,j). \lambda_l\rangle$, to be the basis of a tableau that is resulted by exchanging location of $i$ and $j$ in the boxes. Suppose that for such tableau $t$, the number $j$ ($i$) is located at the $r_j$ and $c_j$ ($r_i$ and $c_i$) row and column of $t$, respectively. Then, define the axial distance $d_{ij}$ of the label $i$ from label $j$ of on each tableau to be $(c_j-c_i)-(r_j-r_i)$. Or in other words, starting with the box containing $i$ walk on the boxes to get to the box $j$. Whenever step up or right is taken add a $-1$, and whenever for a step down or left add a $1$. Starting with the number $0$, the resulting number in the end of the walk is the desired distance. Given this background, the action of $L_{(k,k+1)}$ on the state $|\lambda_i\rangle$, is according to:

$$
L_{(k,k+1)} |\lambda_i\rangle = \dfrac{1}{d_{k+1, k}} |\lambda_i\rangle + \sqrt{1-\dfrac{1}{d^2_{k+1, k}}}|(k,k+1).\lambda_i\rangle
$$

\noindent Three situations can occur: either $k$ and $k+1$ are in the same column or row, or they are not. If they are in the same row, since the tableau is standard, $k$ must come before $k+1$, then the axial distance is $d_{k+1, k}=1$, and the action of $L_{(k,k+1)}$ is merely:

$$
L_{(k,k+1)} |\lambda_i\rangle = |\lambda_i\rangle.
$$

\noindent If the numbers are not in the same column, $k$ must appear right at the top of $k+1$, and the action is:

$$
L_{(k,k+1)} |\lambda_i\rangle = -|\lambda_i\rangle.
$$

\noindent Finally, if neither of these happen, and the two labels are not in the same row or column, then the tableau is placed in the superposition of itself, and the tableau wherein $k$ and $k+1$ are exchanged. Notice that if the tableau $|\lambda_i\rangle$ is standard the exchanged tableau $|\lambda_i\rangle$ is also standard. This can be verified by checking the columns and rows containing $k$ and $k+1$. For example, in the row containing $k$, all the numbers at the left of $k$ are less than $k$, then if we replace $k$ with $k+1$, again all the numbers on the left of $k+1$ are still less than $k+1$. Similar tests for the different parts in the two rows and columns will verify $(k,k+1)\lambda_i$, as a standard tableau. The action of $L_{k,k+1}$ in this case is also an involution. This is obvious for the two cases where $k$ and $k+1$ are in the same row or column. Also, in the third case if the action of $L_{(k,k+1)}$ maps $|\lambda\rangle$ to $\dfrac{1}{d} |\lambda\rangle + \sqrt{1-\dfrac{1}{d^2}}|t\circ \lambda\rangle$ then a second action maps $|t\circ \lambda\rangle$ to $\dfrac{-1}{d} |t\circ \lambda\rangle + \sqrt{1-\dfrac{1}{d^2}}|\lambda\rangle$, and therefore:

$$
L^2_{(k,k+1)}|\lambda\rangle = \dfrac{1}{d}(\dfrac{1}{d} |\lambda\rangle + \sqrt{1-\dfrac{1}{d^2}}|t\circ \lambda\rangle)+\sqrt{1-\dfrac{1}{d^2}}(\dfrac{-1}{d} |t\circ \lambda\rangle + \sqrt{1-\dfrac{1}{d^2}}|\lambda\rangle)=|\lambda\rangle.
$$

Given this description of the invariant subspaces, we wish to provide a partial classification of the image of the ball permuting gates on each of these irreps. The hope is to find denseness in $\prod_\lambda SU(V_\lambda)$, on each of the irreps $V_\lambda$, with an independent action on each block. In this setting, two blocks $\lambda$ and $\mu$ are called dependent, if the action on $\lambda$ is a function of the action on $\mu$, \i.e., the action on the joint block $V_\lambda\oplus V_\mu$ resembles $U\times f(U)$, for some function $f$. Then, independence is translated to decoupled actions like $I \times U$ and $U \times I$.

Throughout, the $\lambda \vdash n$, means that $\lambda$ is a partition of $n$. We say $\mu\vdash n+1$ is constructible by $\lambda \vdash n$, if there is a way of adding a box to $\lambda$ to get $\mu$. We say a partition $\mu \vdash m$ is contained in $\lambda \vdash n$, for $m<n$, if there is a sequence of partitions $\mu_1 \vdash m+1$, $\mu_2 \vdash m+2, \ldots, \mu_{n-m-1} \vdash n-1$, such that $\mu_1$ is constructible by $\mu$, $\lambda$ is constructible by $\mu_{n-m-1}$, and finally for each $j \in [n-m-2]$, $\mu_{j+1}$ is constructible by $\mu_j$. We also call $\mu$ a sub-partition of $\lambda$. A box in a partition $\lambda$ is called removable, if by removing the box the resulting structure is still a partition. Also, define a box to be addable if by adding the box the resulting structure is a partition.

\begin{theorem}
The Young-Yamanouchi bases for partitions of $n$ are adapted to the chain of subgroups $\{e\}=S_1 \subset S_2 \subset \ldots \subset S_n$. 
\end{theorem}

\begin{proof}
Let $\lambda\vdash n$, and $t$ be any standard tableau of shape $\lambda$.  We construct some enumeration of states in the Young-Yamanouchi basis of $\lambda$ which is adapted to the action of subgroups. For any $m<n$, since $t$ is a standard tableau, the numbers $1,2,3,\ldots, m$, are all contained in a sub-partition $\mu \vdash m$ of $\lambda$. This must be true, since otherwise the locus of numbers $1, 2, 3, \ldots, m$ do not shape as a sub-partition of $\lambda$. Let $\nu$ be the smallest sub-partition of $n$ that contains these numbers. Clearly, $|\nu|>m$. The pigeonhole principle implies that, there is a number $k>m$ contained somewhere in $\nu$. The box containing $k$ is not removable from $\nu$, since otherwise you can just remove it to obtain a sub-partition smaller than $\nu$ that contains all of the numbers in $[m]$. Therefore, if $k$ is in the bulk of $\nu$, then both the row and column containing $k$ are not in the standard order. If $k$ is on a vertical (horizontal) boundary, then the column (row) of the box containing $k$ is not standard.

Let $\lambda_k$ be the smallest sub-partition of $\lambda$ that contains $[k]$. Then the enumeration of the basis is according to $|\lambda_1, \lambda_2, \ldots, \lambda_n \rangle$. Here, $\lambda_n=\lambda$, and $\lambda_1$ is a single box. From before, for any $j<n$, $\lambda_{j+1}$ is constructible by $\lambda_j$. For $m<n$, let $S_m$ be the subgroup of $S_n$, that stabilizes the numbers $m+1, m+2, \ldots, n$. For any $k\leq m$, $L_{(k,k+1)}$ just exchanges the content of boxes withing $\lambda_m$, and therefore leaves the quantum numbers $\lambda_{m+1}, \lambda_{m+2}, \ldots, \lambda_n$ invariant. Moreover, the box containing $m$ is somewhere among the removable boxes of $\lambda_m$, since otherwise, as described in the last paragraph, the tableau $\lambda_m$ is not standard. The box containing $m-1$ is either right above or on the left side of $m$, or it is also a removable box. In the first two cases, the action of $L_{(m-1,m)}$ is diagonal, and the quantum numbers are intact. In the third case, the only quantum numbers that are changed are $\lambda_{m-1}$ and $\lambda_m$.
\end{proof}

Consider now the action of $S_{n-1}$ on an element $|\lambda_1, \lambda_2, \ldots , \lambda_n=\lambda\rangle$. In any case $\lambda$ is constructible by $\lambda_{n-1}$, and the construction is by adding an addable box to $\lambda_{n-1}$. In other words, $\lambda_{n-1}$ can be any partition $\vdash n-1$, that is obtained by removing a removable box from $\lambda$. These observations, all together, lead to a neat tool:

\begin{lemma}
(Branching.) Under the action of $S_{n-1}$, $V_\lambda \cong \bigoplus_{\substack{\mu\vdash n-1 \\ \mu\subset \lambda}}V_{\mu}$.
\end{lemma}

\begin{proof}
This is an unusual proof, since it is based on the structure of the YY bases. We would like to emphasize that the multiplicity free branching rule of the symmetric group is manifest in the structure of the YY bases. For a more rational proof see \cite{james1981representation}.

Choose an orthonormal basis according to YY. Enumerate the removable boxes of $\lambda$ by $1,2,\ldots, p$. Clearly, in any standard tableau of $\lambda$, the box containing $n$ is a removable one. Group the tableaus according to the location of $n$. Clearly, each subspace corresponds to a partition $\mu \vdash n-1 \subset \lambda$. Call these partitions $\mu_1, \mu_2, \ldots, \mu_p$, according to the enumeration of removable boxes. Also denote the space $V_{\mu_j}$ correspondingly. For any $\mu_j$, any element of $S_{n-1}$, acted on $V_{\mu_j}$, generates a vector within $V_{\mu_j}$. In other words, these subspaces are stable under $S_{n-1}$.
\end{proof}

\subsection{Partial Classification of Quantum Computation on Arbitrary Initial States}
\label{classification}

The author would like to thank Greg Kuperberg for very useful discussions and notes on the representation theory of the symmetric group which motivated the approach of this section.

In the following, it is proved that the ball permuting gates act densely on invariant subspaces corresponding to Young tableaus with two rows or two columns. The proof is based on the bridge lemma and decoupling lemma of reference \cite{aharonov2011bqp}. As we discuss, conditioned on the existence of a bridge operator, and decoupled dense action on two orthogonal subspace of different dimensionality, the bridge lemma glues the two subspaces into a larger subspace with dense action on it. Also the decoupling lemma decouples action on two orthogonal subspaces of different dimensionality, given dense action on each of them. Greg Kuperberg brought this into our attention that consulting with \cite{kuperberg2011denseness}, it is conceivable that these two lemmas have natural generalizations to more than two subspaces and subspaces that have equal dimensionality. We conjecture that using these tools one can prove that the action of ball permuting gates is dense on all invariant subspaces of the symmetric group, even for those which correspond to Young diagrams of more than two rows/columns. We leave this result to further work. 

In this section, the Lie algebra and the unitary Lie group generated by $X$ operators are used interchangeably. For an intuitive introduction to the Lie Algebra see appendix \ref{LieAlgebra}. As described, the Hilbert space $\C S_n$ has the decomposition:

$$
\C S_n \cong \bigoplus_{\lambda \vdash n} V_\lambda \otimes X_\lambda
$$

Let $G$ be the unitary group generated by these $X(\theta,k)=\exp (i\theta L_{(k,k+1)}).$ operators. As described earlier, the space tangent to the identity element of $G$ is a Lie algebra, $g$, which contains $L_{(k,k+1)}$ for all $k\in [n-1]$, and is close under linear combination over $\mathbb{R}$, and the Lie commutator $i [\cdot, \cdot]$. The objective is to show that for any $\lambda\vdash n$ with two rows or two columns, and any element $U$ of $SU(V_\lambda)$, there is an element of $G$ that is arbitrarily close to $U$.

The proof is presented inductively. First of all, for any $n$, the irreps $V_n$ and $V_{1,1,1,\ldots, n}$ are one dimensional, and the action of $x \in G$ is to add an overall phase. However, observing the structure of YY basis for these irreps, the action of $G$ on the joint blocks $V_n \oplus V_{(1,1,1, \ldots, 1)}$ cannot be decoupled, and the projection of $G$ onto these subspaces is diagonal, and moreover isomorphic to the group $e^{i\theta}\times e^{-i \theta}: \theta \in \mathbb{R}$. Intuitively, these are Bosonic  and Fermionic subspaces, where an exchange $L_{(k,k+1)}$ of particles results in a $+1$ and $-1$ overall phase, respectively. 

For $n=2$, the only invariant subspaces are $V_2$ and $V_{(1,1)}$, and we know the structure of these irreps from the last paragraph:

$$
\C S_2\cong V_2 \oplus V_{(1,1)}, \hspace{1cm} G \twoheadrightarrow e^{i\theta}\times e^{-i \theta}: \theta \in \mathbb{R}.
$$

For $n=3$, the decomposition is according to:

$$
\C S_3 \cong V_{3} \oplus V_{(1,1,1)} \oplus V_{(2,1)}\otimes X(2).
$$

Here, $X(2)$ is a two dimensional multiplicity space. There are two standard $(2,1)$ tableaus and therefore $V_{(2,1)}$ is also two dimensional. Observing the YY basis the two generators $L_{(1,2)}$ and $L_{(2,3)}$ take the matrix forms:

$$
L_{(1,2)} = 
 \begin{pmatrix}
  1 &0 \\
  0&-1
 \end{pmatrix},
$$

and,

$$
L_{(2,3)} = 
 \begin{pmatrix}
  -1/2 &\sqrt{3}/2 \\
  \sqrt{3}/2&1/2
 \end{pmatrix}.
$$

\noindent The basis of the matrix are marked with the two standard Young tableaus of shape $(2,1)$. The first basis corresponds to the numbering $(1,2; 3)$ and the second one corresponds to $(1,3;2)$. Here, the rows are separated by semicolons. The following elements of the Lie algebra $g$ generate $\su(V_{(2,1)})$ and annihilate the two Bosonic and Fermionic subspaces:

$$
\dfrac{1}{2\sqrt{3}}[L_{(1,2)}, [L_{(1,2)},L_{(2,3)}]]= 0 \oplus 0 \oplus \sigma_x\otimes I,
$$

$$
\dfrac{i}{\sqrt{3}}[L_{(1,2)}, L_{(2,3)}]= 0 \oplus 0 \oplus \sigma_y\otimes I,
$$

and,

$$
\dfrac{1}{6}[[L_{(1,2)},[L_{(1,2)}, L_{(2,3)}]],[L_{(1,2)}, L_{(2,3)}]]= 0 \oplus 0 \oplus \sigma_z \otimes I.
$$

This implies the denseness of $G$ in $1 \times 1 \times SU(V_{(2,1)})$. Therefore, we obtain a qubit coupled to the multiplicity space, placed in a superposition of the one dimensional Bosonic and Fermionic subspaces. So, projecting onto a subspace like $V_{(2,1)}\otimes |\psi\rangle$, for $|\psi\rangle \in X(2)$, we obtain a qubit.

We use this result as the seed of an induction. The upshot is to add boxes to $(2,1)$ one by one, in a way that the partitions remain with two rows or two columns. At each step, we use the branching rule to combine the blocks together to larger and larger special unitary groups. In the course of this process, we use two important tools, called the bridge lemma, and decoupling lemma:

\begin{lemma}
(Aharonov-Arad \cite{aharonov2011bqp}) let $A$ and $B$ be two orthogonal subspaces, with \textit{non-equal} dimensions, $\dim A < \dim B$:

\begin{itemize}
\item (Bridge) if there is some state $|\psi\rangle \in A$, and a (bridge) operator $V \in SU (A \oplus B)$, such that the projection of $V |\psi\rangle$ on $B$ is nonzero, then the combination of $SU(A)$, $SU(B)$, and $V$ is dense in $SU(A \oplus B)$.

\item (Decoupling\footnote{We wrote an alternative formulation of Aharonov-Arad's original lemma that is consistent with the ball permuting group $G$.}) suppose for any elements $x\in SU(A)$ and $y \in SU(B)$, there are two corresponding sequences $I_x$ and $I_y$ in $G$, arbitrarily close to $x$ and $y$, respectively, then the action of $G$ on $A \oplus B$ is decoupled, \i.e., $SU(A)\times SU(B)\subseteq G$. 

\end{itemize}
\end{lemma}

See \cite{aharonov2009polynomial,aharonov2007polynomial} for more similar results. Intuitively, what bridge lemma says is that given two subspaces, with one of them larger than the other, dense action each, along with a bridge between them, implies denseness on the combined subspace. That is a bridge glues them to a larger special group. The condition of different dimensions is a crucial requirement for the application of this lemma. The decoupling lemma, on the other hand, states that given dense action on two subspaces, as long as they have different dimensionality, there is way of acting on the two subspaces independently. Again, in this case non-equal dimensionality is important. For example, suppose that $\dim A=\dim B$, then the action $x \times \bar{x}: x \in SU(A)$, cannot be decoupled. Here $\bar{x}$ is the complex conjugate of $x$, \i.e., entries $\bar{x}$ as a matrix are complex conjugates of corresponding entries of matrix $x$. In order to see this, just notice that after finite compositions, the general form of elements generated in this way is $(x_1 x_2 \ldots x_n) \times \overline{(x_1 x_2 \ldots x_n)}$, and an identity action on the left part forces identity action on the right part of the Cartesian product.

Next, we show that the lemma along with the branching rule, force denseness on all irreps corresponding to partitions of two rows or two columns. We will take care of the case with two rows. The situation with two columns is similar. As a way of induction, suppose that, for any $m<n$, for any  $\lambda = (\lambda_1\geq \lambda_2) \vdash m$, the projection of $G$ on $\lambda$ is dense in $SU(V_\lambda)$. The objective is to prove denseness for any partition $\mu \vdash n$.

This is true for $(2,1)$, as showed above. For the sake of illustration, we prove this for $n=4$. The partitions $(4)$ is immediate, because this is one dimensional. Also, the partition $(2,2)$ is immediate, since the branching rule, under the action of $S_3$ is:

$$
V_{(2,2)} \cong V_{(2,1)},
$$

\noindent That is the only removable box from $(2,2)$ is the last box, and in the YY basis for $(2,2)$, this last box can contain the symbol $4$ only. So, the same operators of $S_3$ act densely on this subspace.

The situation with the partition $(3,1)$ is a little different. Analyzing the hook lengths, $V_{(3,1)}$ has dimension $3$, and the branching rule involves the direct sum of partitions $(2,1)$ and $(3)$:

$$
V_{(3,1)} \cong V_{(2,1)} \oplus V_{(3)}.
$$

\noindent Where, $V_{(2,1)}$ is two dimensional, and $V_{(3)}$ is one dimensional, and therefore, they have non-equal dimensions, and also their direct sum adds up to dimension $3$. From, the analysis of $S_3$ we know that independent $SU(2)$, and $SU(1)=\{1\}$ is possible on these irreps. It suffices to find a bridge operator in $SU(V_{(2,1)} \oplus V_{(3)})$. In the first glance, the operator $L_{3,4} \in g$ sounds like a suitable choice. However, there is a problem with this: the restriction of $L_{3,4}$ on $V_{(3,1)}$ is not traceless, and therefore the image under exponentiation does not have unit determinant. Therefore, a wise choice for a bridge operator is $i [L_{(2,3)}, L_{(3,4)}]$. Looking at the actual matrices, restricted to the YY basis of $(3,1)$, one finds $i [L_{(2,3)}, L_{(3,4)}]$, as a suitable bridge, that is nice and traceless:

$$
i 
\begin{pmatrix}
0&\sqrt{2}&-\sqrt{\dfrac{2}{3}}\\
-\sqrt{2}&0 &\sqrt{\dfrac{1}{3}}\\
\sqrt{\dfrac{2}{3}} &-\sqrt{\dfrac{1}{3}}&0\\
\end{pmatrix}.
$$

\noindent Here the matrix is written in the basis corresponding to the tableaus $(1,2,3;4) , (1,2,4; 3)$  and $(1,3,4; 2)$. The bridging is between the $(1,2)$ and $(2,1)$ elements of the matrix. Thereby, the bridge lemma implies the desired denseness.

For general $n$, two situations can happen, either the partition under analysis is of the form $(\nu, \nu)=(n/2, n/2)$ (for even $n$ of course), or not. In the first case, the situation is similar to the partition $(2,2)$ of $n=4$. Thereby, restricted to $S_{n-1}$:

$$
V_{(\nu, \nu)}\cong V_{(\nu, \nu-1)},
$$

\noindent and based on the induction hypothesis the image of $G$ is already dense in the subspace. In the second case, also two cases can happen: either the partition has the form $\mu=(\nu+1, \nu)$, with $2 \nu+1 =n$, or not. In the first case, the branching rule is according to:

$$
V_{(\nu+1, \nu)}\cong V_{(\nu, \nu)}\oplus V_{(\nu+1, \nu-1)}
$$

\noindent The space $V_{(\nu, \nu)}$ corresponds to all YY basis corresponding to tableaus, wherein the index $n$ is located in the last box of the first row. Therefore, the index $n-1$ in all of the tableaus of $(\nu, \nu)$ is located in the last box of the second column, because this is the only removable box available. For simplicity, let's call this space $V_1$. The YY bases of $V_{(\nu+1, n-1)}$ correspond to all the tableaus of $(\nu+1, \nu)$, where the index $n$ is located in the last box of the second row. In this space, the location of the index $n-1$ is either in the last box in the first row or in the box right at the left of the last box in the second row. A coarser stratification of the states in $V_{(\nu+1, \nu-1)}$ is by grouping the YY basis according to the location of $n-1$. Let $V_2$ be the first one, and $V_3$ the second one. Therefore, $YY$ bases of $V_{(\nu+1, \nu)}$ can be grouped in three ways, $V_1, V_2, V_3$, corresponding to all the ways that one can remove two boxes from the original $V_{(\nu+1, \nu)}$. Again, a neat candidate for a bridge is $L_{(n-1,n)}$. Taking a closer look at the operator $L_{(n-1,n)}$, it can be decomposed according to:

$$
L_{(n-1,n)}=\sum_{|j\rangle \in V_3} |j\rangle \langle j|+ \dfrac{1}{2} \sum_{\substack{k' : k\\|k\rangle \in V_1\\|k'\rangle \in V_2}} |k\rangle\langle k| - |k' \rangle \langle k'|+\sqrt{\dfrac{3}{2}} \sum_{\substack{k' : k\\|k\rangle \in V_1\\|k'\rangle \in V_2}} |k\rangle\langle k'| + |k' \rangle \langle k|
$$

\noindent $|j\rangle$, $|k\rangle$, and $|k'\rangle$, of $V_1$, $V_2$, and $V_3$ are the corresponding orthonormal basis in the spaces. Notice that the space $V_1$ is isomorphic to $V_2$, and $k : k'$, refers to this isomorphism. Clearly, the restriction of $L_{(n-1,n)}$ to this block is not traceless, and indeed $tr_{V_{(\nu+1,\nu)}}=\dim V_3 = \dim V_{(\nu+1, \nu-2)}$.

Now, we use the decoupling lemma of Aharonov-Arad. $V_{(\nu+1, \nu)}$ and $V_{(\nu,\nu)}$ have different dimensionality, and also, due to the induction hypothesis the operators can act as the special unitary group on each of them. Thereby, there is a way to act as $x\oplus 0$ on the joint space $V_{(\nu,\nu)}\oplus V_{(\nu+1, \nu-1)}$, for some traceless element $x \in \su(V_{(\nu, \nu)})$. Therefore, $x |j\rangle =0$ and $x |k'\rangle$, for all $|j\rangle \in V_3$, $|k'\rangle \in V_2$. And denote $|x k\rangle := x |k\rangle$, for $|k\rangle \in V_1$. Taking the commutator $i [x, L_{(n-1,n)}]$:

$$
i [x, L_{(n-1,n)}]= \dfrac{i}{2} \sum_{\substack{k' : k\\|k\rangle \in V_1\\|k'\rangle \in V_2}} |x k\rangle\langle k|-|k\rangle\langle x k|+ i \sqrt{\dfrac{3}{2}} \sum_{\substack{k' : k\\|k\rangle \in V_1\\|k'\rangle \in V_2}} |x k\rangle\langle k'| - |k' \rangle \langle x k|.
$$

\noindent Clearly, this operator is traceless, Hermitian, and also one can choose $x$ in such a way that the bridging term in the second sum is nonzero.

Given the above proof for the case $V_{(\nu+1, \nu)}$, we will use a similar technique to take care of the situation $V_{(p, q)}$, where $p> q+1$, and $p+q=n$. Again, the branching rule is:

$$
V_{(p, q)} = V_{(p, q-1)}\oplus V_{(p-1, q)}.
$$

\noindent The space $V_{(p,q-1)}$ corresponds to all YY bases that correspond to the tableaus where the index $n$ is located at the last box of the first row. In this space, the index $n-1$ is either located at the left side of the box containing $n$, or it is located in the last box of the second row. Call the space corresponding to the first (second) one $V_1$ ($V_3$).  $V_{(p-1,q)}$ corresponds to all YY bases of tableaus with index $n$ is located at the last box of the second row. In this space, the index $n-1$ is either located at the left side of the box containing $n$, or it is located in the last box of the first row. Call the first space $V_2$ and the second one $V_4$. Again, write the decomposition of $L_{(n-1, n)}$, accordingly:

$$
L_{(n-1,n)}=\sum_{|j\rangle \in V_1} |j\rangle \langle j|+\sum_{|j\rangle \in V_2} |j\rangle \langle j|+ \alpha(p,q) \sum_{\substack{k' : k\\|k\rangle \in V_3\\|k'\rangle \in V_4}} |k\rangle\langle k| - |k' \rangle \langle k'|+\beta(p,q) \sum_{\substack{k' : k\\|k\rangle \in V_3\\|k'\rangle \in V_4}} |k\rangle\langle k'| + |k' \rangle \langle k|
$$

\noindent Here:

$$
\alpha(p,q) = \dfrac{1}{p-q+1}
$$

\noindent and,

$$
\beta(p,q)=\sqrt{1-\dfrac{1}{(p-q+1)^2}}.
$$

Once again, $V_2$ is isomorphic to $V_3$, and $k: k'$ denotes the correspondence between elements of the two spaces. Once again, we use the decoupling lemma, which asserts the existence of elements like $X:=x\oplus 0$, and $Y:=0 \oplus y$, on $V_{(p,q-1)}\oplus V_{(p-1,q)}$, for every $x\in \su(V_{(p,q-1)})$ and $y \in \su (V_{(p,q-1)})$. A bridge between $V_3$ and $V_4$ is needed, in such a way that the bridge annihilates both $V_1$ and $V_2$. A candidate for a bridge is $[Y,[X, L_{(n-1,n)}]]$. However, it can be easily shown that the element $i[X, L_{(n-1,n)}]$ will also work. The operator $X$ annihilates everything in $V_2$ and $V_4$. Therefor, taking the commutator, the second sum is annihilated, and also, all the remaining terms are traceless and one can find $x$ in such a way that the bridge part is nonzero. All the above results also apply to the tableaus with two columns.

\subsection{Reduction from Exchange Interactions}
\label{bqpuniversality}

In the last section, we demonstrated denseness of $G$ in special unitary groups over some specific invariant subspaces of $\C S_n$. However, this is a nonconstructive statement, and it is desirable to have an explicit description of $\BQP$ simulations in this model. Suppose that we can prepare an arbitrary initial state. For example, suppose that we can initialize the ball permuting model in one of the YY bases of an irrep $V_{(n,n)}$. We can view $V_{(n,n)}$ as a single giant qudit of exponential size. Analyzing the hook lengths, indeed one obtains the dimensionality:

$$
\dim V_{(n,n)}=\dfrac{(2n)!}{(n+1)! n!}=2^{\Omega(n)},
$$

\noindent which is exponential, and can hold $\Omega(n)$ bits of quantum information in it. Even so, it is not clear how to efficiently program the states of $V_{(n,n)}$, using a polynomial time Turing preprocessor. Moreover, in the way that we described the model, the final measurements can be done in the ball labels basis only, and given an output state in $V_{(n,n)}$, it is not obvious how one can sample from it to extract bits of information. Given this motivation, in this section we show how to use arbitrary initial states to obtain a programmable $\BQP$ universal model. This is done by demonstrating a reduction from the exchange interaction model of quantum computation which is already known to be $\BQP$ universal.

Here, we first review the exchange interaction model \cite{kempe2001theory, bacon2001encoded,kempe2001encoded}, and then describe how to do a reduction from the computation in this model to the ball permuting model of computing on arbitrary initial states. Next, we sketch the proof of universality for the exchange interaction model, which in turn results in $\BQP$ universality of ball permuting model on arbitrary initial states. Consider the Hilbert space $(\C^2)^{\tensor n}=:\C \{0,1\}^n$, with binary strings of length $n$, $\mathcal{X}_n:=\{|x_1\rangle \tensor |x_2\rangle \tensor \ldots \tensor |x_n\rangle: x_j \in \{0,1\} \}$, as the orthonormal computational basis. we are interested in the group generated by the unitary gates $T (\theta, i, j)=\exp (i \theta E_{(i,j)})=\cos\theta I + i \sin \theta E_{(i, j)}$, where the operator, $E_{(i,j)}$, called the exchange operator, acts as:

$$
E= \dfrac{1}{2}(I + \sigma_x \tensor \sigma_x + \sigma_y \tensor \sigma_y + \sigma_z \tensor \sigma_z)
$$

\noindent on the $i, j$ slots of the tensor product, and acts as identity on the other parts. More specifically, $E$ is the map:

\begin{eqnarray*}
&|00\rangle&\rightarrow \hspace{3mm}|00\rangle,\\
&|01\rangle&\rightarrow \hspace{3mm}|10\rangle,\\
&|10\rangle&\rightarrow \hspace{3mm}|01\rangle,\\
&|11\rangle&\rightarrow \hspace{3mm}|11\rangle.
\end{eqnarray*}

The action of $E_{i j}$ is very similar to the permuting operator $L_{(i, j)}$, except that $E$ operates on bits rather than the arbitrary labels of $[n]$. These operators are also known as the Heisenberg couplings, related to the Heisenberg Hamiltonian for spin-spin interactions:

$$
H(t)=\sum_{i< j \in [n]} a^x_{i j}(t) \sigma^i_x \tensor \sigma^j_x + a^y_{i j}(t) \sigma^i_y \tensor \sigma^j_y+ a^z_{i j}(t) \sigma^i_z \tensor \sigma^j_z.
$$

\noindent Here $a^{l}_{ij}$ are time dependent real valued couplings. The Heisenberg Hamiltonian models pairwise interaction of spin $\dfrac{1}{2}$ particles, on a network. If one considers zero couplings for the nonadjacent locations with $|i-j|>1$, what the model describes is a chain of spins, with spin-spin interaction of the particles located on a line, or circle if the boundary condition of $a_{i,i+1} = a_{j, j+1}$ is considered for $i= j \mod n$. These are both one dimensional geometries. An isotropic Heisenberg Hamiltonian is the one for which the coefficients satisfy $a^x_{i j}(t)=a^y_{i j}(t)=a^z_{i j}(t)=a(t)_{i j}$, at all times. The isotropic Hamiltonian is expressible by the exchange operators according to:

$$
H(t) =\sum_{i< j \in [n]} a(t)_{ij} (2 E_{(i,j)}-1).
$$

\noindent If one further restricts the $a(t)_{ij}$ couplings to be piecewise constant in time, and that at most one nonzero coupling at a time, in the summation above, $H(t)$ imposes a general form of unitary evolution, according to:

\begin{equation}
U=T (\theta_m, i_m, j_m)\ldots T (\theta_2, i_2, j_2) T (\theta_1, i_1, j_1).
\label{exch}
\end{equation}

\noindent Given $(\theta_1, i_1, j_1), (\theta_2, i_2, j_2),\ldots, (\theta_m, i_m, j_m)$ as the description of $U$, and special initial states $|\psi \rangle \in (\C^2)^{\tensor n}$, we will demonstrate how to use $X$ operators with arbitrary initial state to simulate $U$, and then sample from the output distribution of $U |\psi\rangle$. After that, we consult a previously known result, which asserts that exchange interactions are sufficient for universal quantum computing.

\begin{definition}
Define $\mathcal{X}^k_n :=\{|x\rangle: x\in \{0,1\}^n, |x|_H = k\}$ to be the subset of $\mathcal{X}_n$, containing strings of Hamming distance $k\leq n$. Here, $|.|_H$ is the Hamming distance, which is the number of $1$'s in a string. Also, let $\C \mathcal{X}^k_n$ \footnote{Usually $\C G$ refers to a group algebra, however, here we just use $\C \mathcal{X}^k_n$ just for the simplicity of notations.} be the corresponding Hilbert space spanned by these basis.
\end{definition}

\begin{theorem}
Given a description of $U$ (in equation ~\ref{exch}), and an initial state $|\psi\rangle \in \C \mathcal{X}^k_n$, there exists an initial $|\psi'\rangle\in \C S_n$, and a ball permuting circuit, with $X$ operators, that can sample from the output of $U|\psi'\rangle$, exactly.
\label{hrd}
\end{theorem}

\begin{proof}
We show how to encode any state of $\C \mathcal{X}^k_n$ with states of $\C S_n$. Let $S_{k, n-k}$ be the subgroup of $S_n$ according to the cycles $\{1,2,\ldots, k\}$ and $\{k+1, k+2, \ldots, n\}$, and denote $|\phi_0\rangle = \dfrac{1}{\sqrt{k! (n-k)!}}\sum_{\sigma \in S_{k, n-k}} R(\sigma) |123\ldots n\rangle$ be an encoding of the state $|1^k 0^{n-k}\rangle$. Here, $1^k$ means $1$'s repeated for $k$ times. This is indeed a quantum state that is symmetric on each the labels of $\{1,2,\ldots, k\}$ and $\{k+1,k+2,\ldots, n\}$, separately.  Any string of Hamming distance $k$ can be obtained by permuting the string $0^k 1^{n-k}$. For any such string $x$ let $\pi_x$ be such a permutation, and encode $|x\rangle $ with $|\phi(x)\rangle := L_{\pi_x} |\phi_0\rangle$. Therefore, given any initial state $|\psi\rangle := \sum_{x \in \mathcal{X}^k_n} \alpha_x |x\rangle$, pick an initial state $|\psi'\rangle := \sum_{x \in \mathcal{X}^k_n} \alpha_x |\phi(x)\rangle$ in $\C S_n$. Now, given any unitary $U=T (\theta_m, i_m, j_m)\ldots T (\theta_2, i_2, j_2) T (\theta_1, i_1, j_1)$ with $T$ operators, pick a corresponding ball permuting circuit $U'=X (\theta_m, i_m, j_m)\ldots X (\theta_2, i_2, j_2) X (\theta_1, i_1, j_1)$. It can be confirmed that for any $i< j \in [n]$ if $E_{(i, j)} |x\rangle = |x'\rangle$, then $E_{(i, j)} |\phi(x)\rangle = |\phi(x')\rangle$. From this, if $U|\psi\rangle=\sum_{x \in \mathcal{X}^k_n} \beta_x |x\rangle$, then $U'|\psi'\rangle=\sum_{x \in \mathcal{X}^k_n} \beta_x |\phi(x)\rangle$.

 It remains to show that given access to the output of $U'|\psi'\rangle$, one can efficiently sample from $U|\psi\rangle$. Suppose that $U'|\psi'\rangle$ is measured in the end, and one obtains the permutation $\sigma= (\sigma(1),\sigma(2),\ldots, \sigma (n))$. Then, by outputting a string $x$ by replacing all the labels of $\{1,2,\ldots, k\}$ in $\sigma$ with ones and the other labels with zeros the reduction is complete. The probability of obtaining any string $x$ with this protocol is exactly equal to $| \langle x |U|\psi\rangle |^2$.
\end{proof}

Indeed, in this simulation, the space is going to be projected onto a subspace of $\C S_n$ that is invariant under $X$ operators. Moreover, this subspace is isomorphic to the strings of bits with certain Hamming distances. One can formally extend this idea to other similar subspaces. As before, let $\lambda= (\lambda_1, \lambda_2, \ldots, \lambda_t)$ be a partition of $n$, and $S_\lambda\cong S_{\lambda_1}\times S_{\lambda_2}\times\ldots S_{\lambda_t}$, be the subgroup of $S_n$ as the set of permutations with cycles $\{1,2,\ldots, \lambda_1\}, \{\lambda_1 +1, \ldots, \lambda_2\}, \{\lambda_{t-1}+1, \ldots, \lambda_t\}$. Then it can be seen that $P(\lambda):=\dfrac{1}{\lambda!}\sum_{\sigma\in S_\lambda} R(\sigma)$ is a projection, \i.e., it is Hermitian and also $P(\lambda)^2=P(\lambda)$. $P(\lambda)$ is Hermitian, since $R(\cdot)$ is a Hermitian operator. $S_\lambda$ is a group, and $R$ is a homomorphism, therefore for any $\tau \in S_\lambda$, $R(\tau) \sum_{\sigma \in S_\lambda} R(\sigma)=\sum_{\sigma \in S_\lambda} R(\sigma)$, which is implied by the closure of $R(S_\lambda)$ as a group. Therefore:

$$
P(\lambda)^2=\dfrac{1}{\lambda!^2}\sum_{\tau \in S_\lambda} R(\tau) \sum_{\sigma \in S_\lambda} R(\sigma)=\dfrac{|S_\lambda|}{\lambda !^2}\sum_{\sigma \in S_\lambda} R(\sigma)=P(\lambda).
$$

\noindent Indeed, looking at the Young symmetrizers, it can be confirmed that the subspace $V_{\lambda}$ is contained in the space resulted under this projection, and the subspace is further reducible. We need to show that for any partition $\lambda$, $P(\lambda) \C S_n$ is a subspace that is invariant under the $X$ operators. Let $W_\lambda:= \{|\psi\rangle \in \C S_n: (I- P(\lambda))|\psi\rangle=0 \}$ be the subspace obtained by this projection, and $W'_\lambda$ as its complement in $\C S_n$. Choose any $|\psi\rangle \in W_\lambda$, we claim that for any operator $X$, $X|\psi\rangle \in W_\lambda$. This is true, because the projection $P(\lambda)$ commutes with $X$, and $(I-P(\lambda))X|\psi\rangle=X (I-P(\lambda))|\psi\rangle=0$. However, as it is going to be mentioned in a later section, these subspaces are further reducible. More specifically, in the proof of theorem ~\ref{hrd}, we used the partition $\lambda=(k, n-k)$, and constructed the subspace $W_\lambda:= P(\lambda) \C S_n \subset \C S_n$ as the encoding of $\C \mathcal{X}^k_n$; in other words $W_{(k, n-k)}\cong \C \mathcal{X}^k_n$.

For any $k\in [n]$, $\C \mathcal{X}^k_n$ is an invariant subspace of the group, $G_T$, generated by $T$ operators. This is because the exchange operators do not change the Hamming distance of the computational basis. So the decomposition $\C \{0,1\}^n \cong \bigoplus_{k} \C \mathcal{X}^k_n$ is immediate. Consider the standard total $Z$ direction angular momentum operator:

$$
J_Z:= \dfrac{1}{2} (\sigma^1_z+\sigma^2_z+\ldots+\sigma^n_z).
$$

\noindent
Then $[J_Z, E_{(i,j)}]=0$ for all $i$ and $j$. Here, the superscript $j$ in $A ^j$ for operator $A$ means $I\tensor I \tensor \ldots \tensor \overset{\overset{j}{\downarrow}}{A} \tensor \ldots \tensor I$, the action of the operator on the $j$'th slot of the tensor product. $J_Z$ indeed counts the Hamming distance of a string, and more precisely, for any $|\psi\rangle \in \C \mathcal{X}^k_n$, $J_Z|\psi\rangle = (\dfrac{n}{2}-k)|\psi\rangle$. Therefore, the eigenspace corresponding to each eigenvalue of $J_Z$ is an invariant subspace of $G_T$. For each eigenvalue $n/2 - k$ the multiplicity of this space is ${n}\choose {k}$, the number of $n$ bit strings of Hamming distance $k$. One can also define the $X$ and $Y$ direction total angular momentum operators in the same way:

$$
J_X:= \dfrac{1}{2} (\sigma^1_x+\sigma^2_x+\ldots+\sigma^n_x),
$$

and,

$$
J_Y:= \dfrac{1}{2} (\sigma^1_y+\sigma^2_y+\ldots+\sigma^n_y).
$$

\noindent Indeed, consulting the decoherence free subspaces theory of the exchange operators, the algebra generated by the operators $J_X, J_Y$ and $J_Z$, is the unique commutant of the exchange operators, and vice versa. Indeed, for any positive algebra that is closed under the conjugation map, the commutant relation is an involution \cite{james1981representation}, i.e., the commutant of the commutant of any such algebra is the algebra itself.

The decomposition of $\C \{0,1\}^n$, of $n$ spin $\dfrac{1}{2}$ particles, is well known, and can be characterized by total angular momentum, and the $Z$ direction of the total angular momentum. The total angular momentum operator is:

$$
J^2 = (\sum_{j\in [n]}\dfrac{1}{2}\sigma^j_x)^2+(\sum_{j\in [n]} \dfrac{1}{2}\sigma^j_y)^2+(\sum_{j\in [n]}\dfrac{1}{2}\sigma^j_z)^2.
$$

\noindent Indeed, using a minimal calculation one can rewrite $J^2$ as:

$$
J^2= n(n-1/4) + \sum_{i<j \in [n]} E_{(i,j)},
$$

\noindent and it can be confirmed that for all $k<l \in [n]$, $[E_{(k,l)}, J^2]=0$. In other words, the exchange operators do not change the total and $Z$ direction angular momentum of the a system of spin $1/2$ particles. The decomposition of $\C \{0,1\}^n$ can be written down according to these quantum numbers. Let $V(s)\subset \C \{0,1\}^n$, be the set of states $|\psi\rangle$ in $\C \{0,1\}^n$ such that $J^2|\psi\rangle = s(s+1/2)|\psi\rangle$, and $V(s,m)\subset V(s)\subset \C \{0,1\}^n$, as the subspace with states $|\phi\rangle$ such that $J_Z |\phi\rangle = m/2 |\phi\rangle$. If $n$ is even, $\C \{0,1\}^n$ decomposes according to:

$$
\C \{0,1\}^n \cong V(0) \oplus V(1) \oplus \ldots \oplus V(n/2),
$$

\noindent
and each of these subspaces further decomposes to:

$$
V(s)\cong V(s,-s)\oplus V(s,-s+1) \oplus \ldots \oplus V( s, s).
$$

\noindent For odd $n$ the only difference is in the decomposition $\C \{0,1\}^n \cong V(\dfrac{1}{2}) \oplus V(\dfrac{3}{2}) \oplus \ldots \oplus V(n/2)$. From what is described in the context of decoherence free subspaces theory, the exchange interaction can affect the multiplicity space of each subspace $V(s,m)$. We are interested in the subspaces of the form $V(s,0)$ for even $n$, and $V(s, \pm \dfrac{1}{2})$, for odd $n$, which correspond to the decomposition of $\mathcal{X}^{n/2}_n$, and $\mathcal{X}^{(n \pm 1)/{2}}_n$, based on the total angular momentum, respectively.

There is a neat connection between the multiplicity space of these subspaces, and the subgroup adapted YY bases. For $k\in [n]$, define the following series of operators:

$$
J_k^2 = k(k-1/4) + \sum_{i<j \in [k]} E_{(i,j)}.
$$

\noindent Clearly, $J_k^2=J^2$. These are indeed the total angular momentum measured by just looking at the first $k$ particles. Using a minimal calculation one gets $[J^2_k , J^2_l]=0$ for all $k,l$. That is they are all commuting, and they can be mutually diagonalized. For $x_j \in [n]$, let $|x_1, x_2, \ldots, x_n\rangle$, be such basis with $J^2_j |x_1, x_2, \ldots, x_n\rangle=x_j (x_j + \dfrac{1}{2})|x_1, x_2, \ldots, x_n\rangle$. These are appropriate candidates as a basis for the multiplicity space of $V(s,0)$($V(s,1/2)$ for odd $n$). Then, $x_n=s$. Analyzing these operators more carefully, it is realized that for each $l<n$, either $x_{l+1}= x_l + 1/2$ or $x_{l+1}= x_l - 1/2$. Intuitively, this is saying that adding a new spin $1/2$ particle $Q=\C^2$ to $V(x_j)$:

$$
V(x_j)\tensor Q \cong V(x_j+\dfrac{1}{2})\oplus V(x_j-\dfrac{1}{2}),
$$

\noindent for $x_j>0$, and otherwise:

$$
V(0)\tensor Q \cong V(\dfrac{1}{2}),
$$

\noindent This is similar to the branching rule of the symmetric group representation theory. The second form is directly related to the branching rule of $V_{(n,n)}\cong V_{(n,n-1)}$. For simplicity, here we consider the twice of the $J$ operators instead, so that the branching rule takes the form:

$$
V(x_j)\tensor Q \cong V(x_j+1)\oplus V(x_j-1),
$$

\noindent for $x_j>0$ and,

$$
V(0)\tensor Q \cong V(1),
$$


\begin{figure}[tp]
\centering
\includegraphics[height=3.0in]{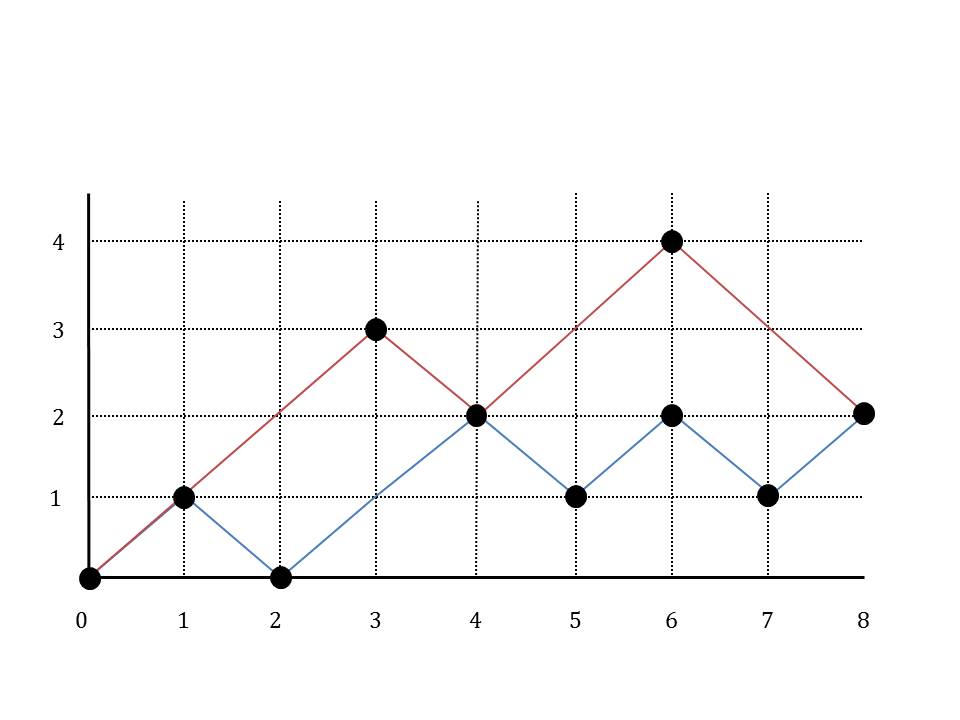}
\caption[Example of a path model]{An example of a path model. The blue and red paths start out of $(0, 0)$ and end up with the point $(8,2)$. YY basis corresponding to tableaus that two rows are closely related to the path model.
}
\label{Path}
\end{figure}

\noindent otherwise. Applying this rule recursively, the path model is obtained. See Figure \ref{Path} for an example. A path model $P_s$ is the set of paths between two points $(0,0)$ and $(n, s)$ in a two dimensional discrete Cartesian plane $\{0,1,2,3,\ldots, n\}^2$, where no path is allowed to cross the $(x,0)$ line, and at each step the path will move either one step up or one step down. Path up/down from the point $(x,y)$ is the connection from this point to $(x+1, y+1)/(x+1, y-1)$. See Figure \ref{Path} for an example of a path model. Therefore, the Hilbert space $V(s)$ corresponds to the orthonormal basis labeled by the Paths to the point $(n,2s)$. 

Indeed, we could agree on a path model for the YY basis of the tableaus with two rows. Let $\lambda$ be the Young diagram of shape $\lambda= (n,m)$. The path model is constructed in the following way: map each tableau $t$ to a symbol, $M_t= y_1 y_2 \ldots y_{n+m}$, where $y_j$ is the row index of the box containing the number $j$. Starting at the point $(0, 0)$ then the path corresponding to $t$ is constructed by taking a step up, whenever a $1$ is read in $M$, and a step down otherwise. Thereby, $P_0$ corresponds to $(n,n)$, and $P_{2n}$ corresponds to the partition $(2n)$.

Given this background, universal quantum computing is possible by encoding a qubit using three spin $1/2$ particles. Suppose that the following initial states are given in $\C \mathcal{X}^1_3$:

$$
|0_L\rangle := \dfrac{|010\rangle-|100\rangle}{\sqrt{2}}
$$

and,

$$
|1_L\rangle :=\dfrac{2|001\rangle-|010\rangle-|100\rangle}{\sqrt{6}},
$$

\noindent
as some logical encoding of a qubit using three quantum digits. We claim that there is a way to distinguish $|0_L\rangle$ from $|1_L\rangle$ with perfect soundness. These mark the multiplicity space of the space with half $Z$ direction angular momentum and half total angular momentum. First, we should find a way to distinguish between these two states using measurement in the computational basis. Suppose that we have access to $k$ copies of an unknown quantum state, and we have the promise that it is either $|0_L\rangle$ or $|1_L\rangle$, and we want to see which one is the case. The idea is to simply measure the third bit of each copy, and announce it to be $0_L$ if the results of the $k$ measurements are all $0$ bits. If the state has been $|0_L\rangle$, the probability of error in this decision is zero, because $|0_L\rangle= \dfrac{|01\rangle-|10\rangle}{\sqrt{2}}\tensor |0\rangle$. Otherwise, we will make a wrong decision with probability at most $(1/3)^k$, which is exponentially small. This is because the probability of reading a $0$ in the third bit of $|1_L\rangle$ is $1/3$. 

\begin{theorem}
There is a way of acting as encoded $SU(2)$ on the span of $\{|0_L\rangle,|1_L\rangle \}$, and also $SU(4)$ on the concatenation of two encoded qubits. 
\end{theorem}

\begin{proof}(Sketch) according to the analysis of \cite{divincenzo2000universal, kempe2001theory}, one can look at the Lie algebra of the exchange operators to find encoded $\su(2)$ algebra on the encoded qubit. Also, we need to take enough commutations such that the action of the designed operators annihilates the two one dimensional spaces spanned by $|000\rangle$ and $|111\rangle$. The authors of \cite{kempe2001theory} prove that there is a way to act as $SU(V(s,m))$ on each invariant subspace $V(s,m)$. Moreover, they prove that the action on two subspaces $V(s_1, m_1)$ and $V(s_2, m_2)$ can be decoupled, unless $s_1=s_2$, and $m_2=-m_1$, where the two subspaces are isomorphic. It is almost enough to prove that the state $|0_L\rangle \otimes |0_L\rangle$ is contained in non-isomorphic invariant subspaces. However, this is also true, since $|0_L\rangle \otimes |0_L\rangle$ is completely contained in subspaces with $m=2$.
\end{proof}

See \cite{bauer2014universality,kempe2002exact,wu2002power} for similar models with encoded universality. Therefore, this is a nonconstructive proof for the existence of an encoded entangling quantum gate; CNOT for example. Indeed, the actual construction of a CNOT is given in \cite{divincenzo2000universal} . Notice that for a decision problem, one can formulate quantum computation in such a way that only one qubit needs to be measured in the end, and this can be done by distinguishing $|0_L\rangle$ and $|1_L\rangle$ using measurement in the computational basis. The probability of success in distinguishing between the two bits can also be amplified by just repeating the computation for polynomial number of times, and taking the majority of votes. Also, taking the majority of votes can be done with encoded CNOTs and single qubits gates on a larger circuit, and without loss of generality we can assume that one single measurement on one single qubit is sufficient.

\section{The Scattering Quantum Computer with Demolition Intermediate Measurements}
\label{post}
In section \ref{Upper-bounds} of this chapter we showed that the group of unitaries generated by Yang-Baxter circuits constitute the union of small dimensional manifolds. Since this model is the generalization of scattering models in section \ref{models}, we conclude:

\begin{corollary}
(Of theorem ~\ref{YBnonu}) let $S$ be the set of unitary scattering operators generated by $n$ particle scattering with any of the models in section 2. Then $S$ corresponds to a manifold of dimension at most $n$.
\end{corollary}

\begin{proof}
The proof of theorem ~\ref{YBnonu} still works here, with a tiny difference. In the scattering problem, the signature of the corresponding planar YB quantum circuit is fixed by the velocities, therefore, for any set of velocities, the model can generate at most one unitary, and hence the points in the set of scattering matrices combined together is parameterized with $n$ real numbers, as velocities. 
\end{proof}

\noindent Moreover:

\begin{corollary}
(Of theorem \ref{YBnonu1}) let $S$ be the set of unitary scattering operators generated by $n$ particle scattering with any of the models in section 2, with and without postselection in the particle label basis in the end of computation. Then $S$ corresponds to $n!^{O(1)}$ manifolds each with dimension at most $n$.
\end{corollary}

\begin{proof}
Following the proof of \ref{YBnonu1}, if the velocity parameters are fixed, then the model can generate exactly one unitary operator, and there are (discrete) $n!^{O(1)}$ ways to do postselection on the output of this unitary matrix. Therefore, given a set of particles initialized with fixed velocities, the model can generate at most $n!^{O(1)}$ number of unitary scattering matrices. 
\end{proof}

Given this observation, we find out that probably a proof for postselected $\BQP$ universality of particle scattering will probably fail, if we postselect in the particle label basis in the end of computation. For this reason we modify the model and allow intermediate measurement. The result, as it is going to be established in section \ref{phcol}, with some complexity theoretic assumptions this new model is hard to simulate on a classical computer.

There are two ways to do an intermediate measurement. The first of these is to measure, intermediately, in the particle label basis, in a way that the outcome of the measurement is the post-measurement quantum state. In this way, because of the planarity of circuits in one spatial dimension, we have to use the post-measurement state over and over. The second scheme of intermediate measurement is demolition measurement due to a particle detector. This measurement reveals the classical output of measurement but not the post-measurement quantum state. In this case the measured particle will not participate in further scatterings. The second model of measurement is more realistically connected to scattering models of section \ref{models}; however, we also establish a similar result for the non-demolition model, which is more related to the general $HQBALL$ complexity class. We consider non-adaptive measurements in both cases.

\subsection{Programming the Scattering Amplitudes with Intermediate Measurements}

The goal is to come up with a quantum algorithm based on particle scattering in $1+1$ dimension, which takes the description $\langle C \rangle$ of a general $X$ quantum ball permuting model as an input, and outputs the description of a sequence of particle scatterings and a sequence of intermediate non-adaptive demolition particle measurements, in a way that the overall process efficiently samples from the output of $C$. The construction of this section is very similar to the nondeterministic gates of \cite{knill2001scheme,knill2000efficient}. For a review of quantum computing with intermediate measurements see \cite{leung2004quantum, terhal2002adaptive,briegel2009measurement}.

Consider the $X$ ball-permuting gate of Figure \ref{fig1}, where we let the two input  wires interact with arbitrary amplitudes, and in the end we measure the label locations of $A$ and $B$. The objective is to have a particle scattering gadget that can simulate the output distribution of this circuit. Therefore, we can use the four particle gadget of Figure \ref{fig2}. The left and right rectangles demonstrate demolition measurements and the final superposition is created at the locations $A$ and $B$. The overall scattering process acts as a nondeterministic gate, in the sense that the gadget succeeds its simulation, only if the left detector measures label $a$ and the right detector measures label $b$, and an experimenter can verify this in the end. The velocities $v_1, v_2, v_a$ and $v_b$ can be tuned in such a way that the desired swap is obtained. The probability of success, thereby, depends on these velocity parameters. More precisely, conditioned on a successful simulation, the overall action of the scattering gadget is the gate $X(\tan^{-1 }z_{eff}, 1)$, where:

$$
\tan^{-1} z_{eff}= \tan^{-1} z_{1}+\tan^{-1} z_{2}
$$

\noindent with $z_1=v_1-v_2$ and $z_2= v_a-v_b$. As a result of this, the left and right black output particles will have velocities $v_b$ and $v_a$, respectively.   Notice that all of these results still hold if the black particles start out of arbitrary initial superpositions. However, one should make sure that the state of the ancilla particles are separable from the black ones.

Moreover, as described, in this model of scattering the particles move on straight line in time-space place, and they do not naturally change their directions. We thereby can use a two particle gadget of Figure \ref{fig3} to navigate the particles' trajectories. The two particles collide from left to right, and the left particle is measured in the end. Conditioned on the detector measuring the label $a$, the navigation is successful, and the outcome of this process is particle with its original label $1$ moves to the right direction with velocity $v_a$. One can match $v_a=v_1$, so that the overall action of the nondeterministic gadget is a change of direction. The success probability, then depends on $v_1$ and $v_a$.

As another example consider the $X$ quantum ball permuting circuit of Figure \ref{fig4}. This circuit consists of $X$ gates, $1, 2$ and $3$, and they permute labels of the four input wires. In the end we measure the output wires $A, B, C$ and $D$, in the particle label basis. We use the particle scattering sequences of Figure \ref{fig5} to simulate this circuit. Again the blue particles are ancilla, and the black particles correspond to the wires, and the labels $1, 2$ and $3$, correspond to the simulation of gates $1, 2$ and $3$ in Figure \ref{fig4}, respectively. Each of the detectors measure in the particle label basis, and in the end the experimenter measures the particle locations $A, B, C$ and $D$, corresponding to the output wires $A, B, C$ and $D$, in Figure \ref{fig4}, respectively. The overall scattering process succeeds in its simulation only if the detectors measure the ancilla particles with their initial labels. That is, conditioned on all blue particles successfully pass through their intermediate interactions and bouncing off the last interaction, the scattering process simulates the circuit successfully. This is true, because the particles move on straight lines, and the only event corresponding to detecting an ancilla particle with its original label is the one where it never bounces off in its intermediate interactions, and bounces off its final collision before moving to the detector. For an example of a larger simulation see the simulation of the $X$ quantum circuit of Figure \ref{fig6} with the scattering process of Figure \ref{fig7}. This example specially, demonstrates that during the scattering, blue (ancilla) particles can experience many intermediate interactions, and the number of these interactions can scale linearly in the number of particles being used. Therefore, the event corresponding to a successful simulation can have exponentially small probability.

It is important to mention that because of the Yang-Baxter equation, braiding of two particles is impossible. Braiding means that two particles can interact with each other over and over, however, because of the expression of unitarity, $H(u) H(-u) = I$, two successive collisions is equivalent to no collision. The role of the intermediate measurements is to allow two particles interact over and over without ending up with identity. 

\subsection{Stationary Programming of Particle Scatterings with Intermediate Demolition Measurements}

The simulations of last section are both intuitive and instructive. However, they have a drawback. The slope of the lines corresponding to particle trajectories, depend on the velocities of the particles. So for large simulations, we need to keep the track of the architecture of collisions, and the amplitudes of interactions at the same time, and this can be both messy and difficult. In this section, we try to present a better simulation scheme where one only needs to keep track of amplitudes, and the architecture of collisions can be tuned easily. The philosophy is to have steady particles, in the beginning, and whenever we want a ball permuting gate, a number of ancilla particles are fired to the target steady particles. Then the intermediate detections are used, and then postselections on their outcomes enables the model to simulate an arbitrary $X$ quantum ball permutation. By stationary particle we mean a particle that is not moving. In order to fulfill this purpose, we use the stationary gadget of Figure \ref{fig8}. The objective is to impose a desired permutation on the input black particles. And we want the black particles to stay stationary in the end of the simulation. In order to do this, two other stationary ancillas are put at the left and right of the black particles. Then, two other ancilla particles, the desired velocities, are fired from left and right, and postselection is made on them bouncing off from the black particles. Then the two black particles interact and exchange momenta, and then they collide with the two stationary ancillas. In the end, we measure and postselect on the ancilla particles bouncing off the black particles. Therefore, in the end of the process, the stationary black particles are left stationary, and the desired superposition is obtained. In order to see an example for the implementation of the stationary programming in larger circuits, see the simulation of $\XQBALL$ circuit of Figure \ref{fig9} with stationary particle programming of Figure \ref{fig10}. 

\begin{figure}[tp]
\centering
\begin{subfigure}{.5\textwidth}
  \centering
{\includegraphics[height=2.0in]{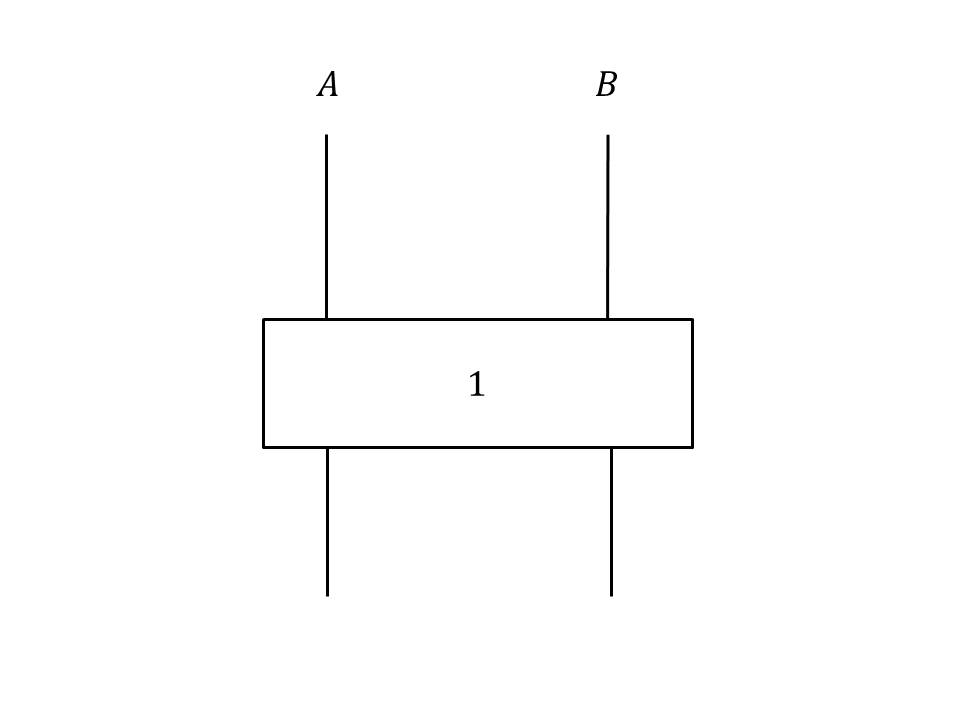}}
  \caption{}
  \label{fig1}
\end{subfigure}%
\hspace{-1cm}
\begin{subfigure}{.5\textwidth}
  \centering
{\includegraphics[height=2.5in]{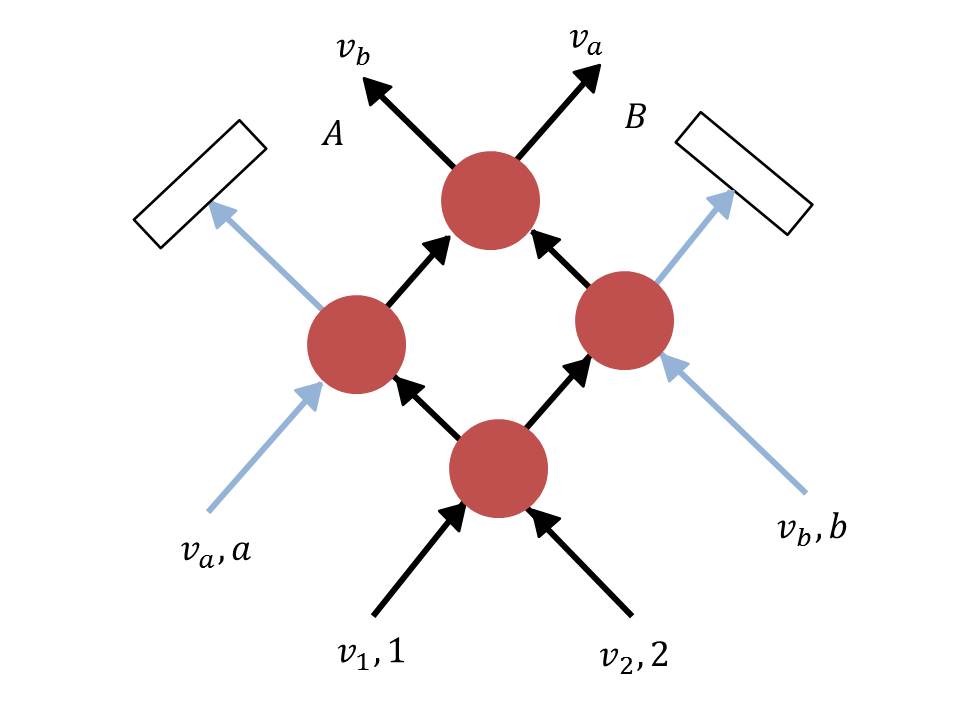}}
  \caption{}
  \label{fig2}
\end{subfigure}
\caption[Nondeterministic four-particle gadget which allows braiding of the balls]{(a) The representation of an $X$ operator. The gate permutes the input labels, and in the end we measure the labels of output wires $A$ and $B$. (b) Four-particle scattering gadget to simulate the $X$ rotation. Lines represent the trajectories of particles, red circles demonstrate interactions, and white rectangles are detectors. Blue particles are ancillas which mediate computation, and black particles are the particles that we wish to implement the actual quantum swap on. The gate is nondeterministic in the sense that it succeeds in producing the desired superposition on labels $|1\rangle$ and $|2\rangle$ only if the left and right detectors detect $|a\rangle$ and $|b\rangle$ labels in the particle label basis, respectively. Conditioned on both ancilla particles bounce off the black particles, the gate operates successfully. The probability of success, thereby, depends on the velocities.
}
\label{fig1-2}
\end{figure}

\begin{figure}[tp]
\centering
\includegraphics[height=3.0in]{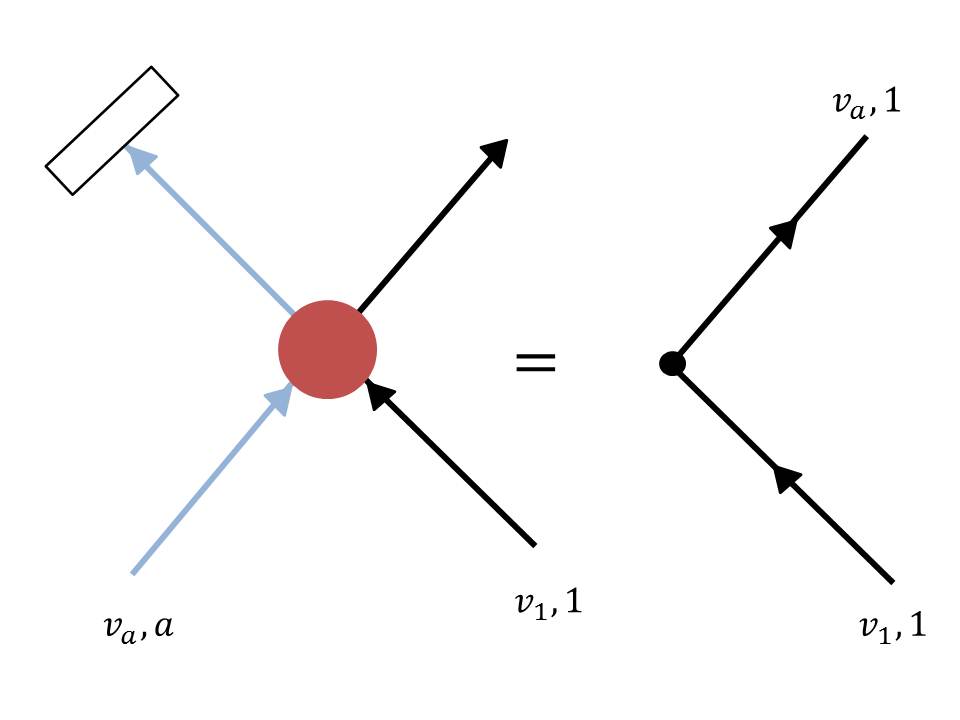}
\caption[Nondeterministic two-particle gadget to navigate the particles]{Two-particle gadget to navigate the trajectory of a single particle. Since in the model we consider the particles move on straight lines, we use this nondeterministic gadget to change the particle's trajectory. The particle that is moving left with velocity $v_1$ non-deterministically changes direction to the right with velocity $v_a$, and this event succeeds only if the detector on the left detects label $|a\rangle$. If the velocities match, $v_a=v_1$, the overall action is a change of direction.
}
\label{fig3}
\end{figure}

\begin{figure}[tp]
\centering
\begin{subfigure}{.5\textwidth}
  \centering
{\includegraphics[height=2.0in]{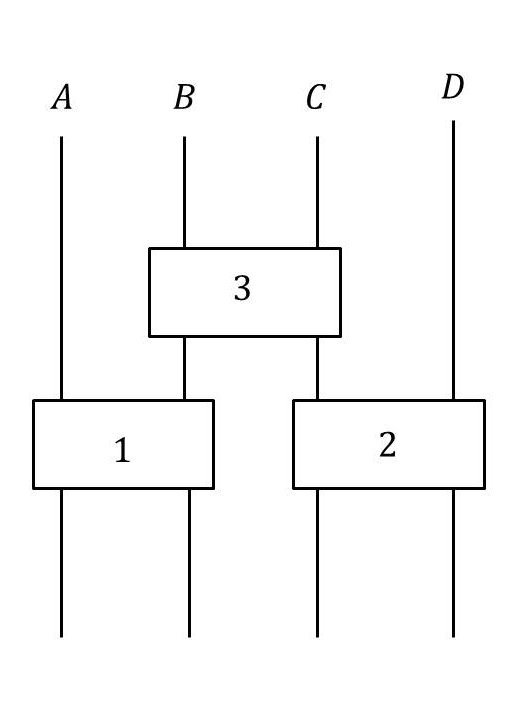}}
  \caption{}
  \label{fig4}
\end{subfigure}%
\hspace{-1cm}
\begin{subfigure}{.5\textwidth}
  \centering
{\includegraphics[height=2.5in]{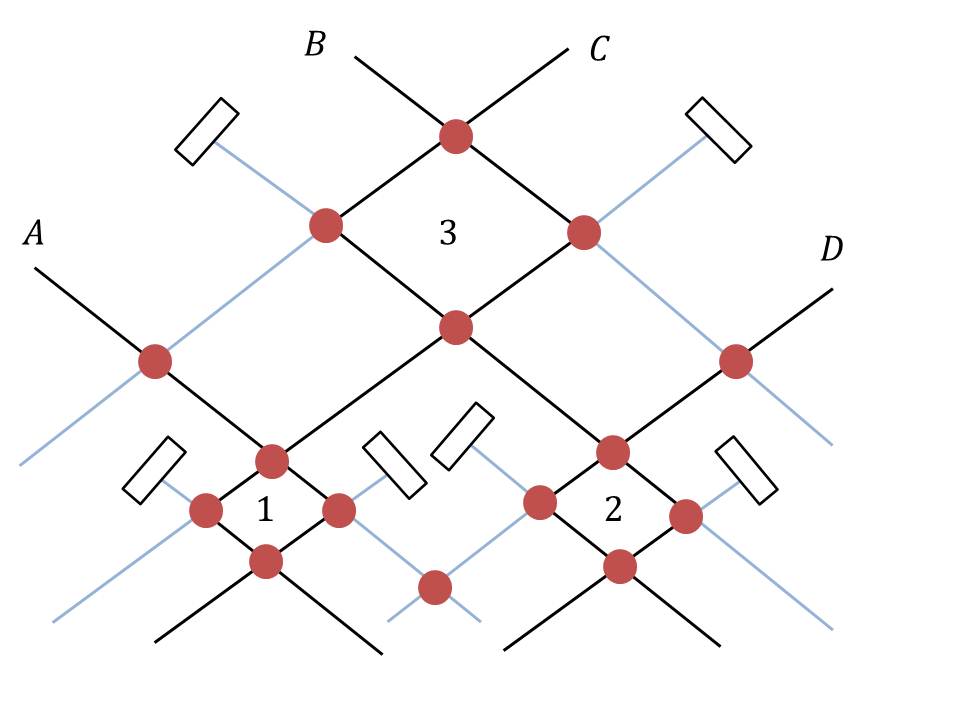}}
  \caption{}
  \label{fig5}
\end{subfigure}
\caption[An example for the simulation of a ball permuting circuit with nondeterministic ball scattering gadgets]{(a) Example of a combination of $X$ operators forming a circuit. The circuit consists of three gates, $1$, $2$, and $3$, and in the end the wires $A$, $B$, $C$, and $D$ are measured in the label basis. (b) An architecture of quantum ball permuting circuit based on particle scattering and intermediate particle measurements to simulate quantum ball permuting circuit of Figure (a). The circuit consists of six ancilla particles which mediate the computation and are detected intermediately with detectors. The labels $1$, $2$, and $3$ demonstrate the simulation of gates $1$, $2$, and $3$, of Figure (a), respectively. In the end we measure the particle locations $A$, $B$, $C$, and $D$. Conditioned on all ancilla particles succeed in passing through all of the intermediate interactions and bouncing off the last interaction, the overall scattering process succeeds in its simulation.}
\label{fig4-5}
\end{figure}


\begin{figure}[tp]
\centering
\begin{subfigure}{.5\textwidth}
  \centering
{\includegraphics[height=2.0in]{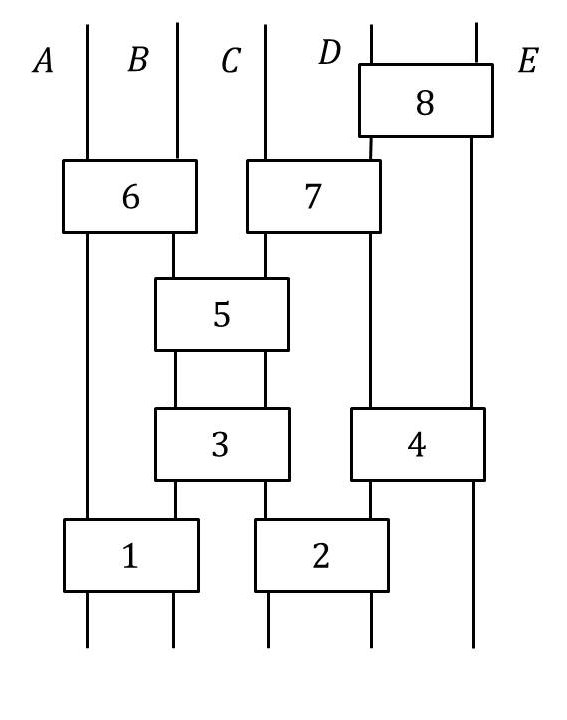}}
  \caption{}
  \label{fig6}
\end{subfigure}%
\hspace{-1cm}
\begin{subfigure}{.5\textwidth}
  \centering
{\includegraphics[height=3.0in]{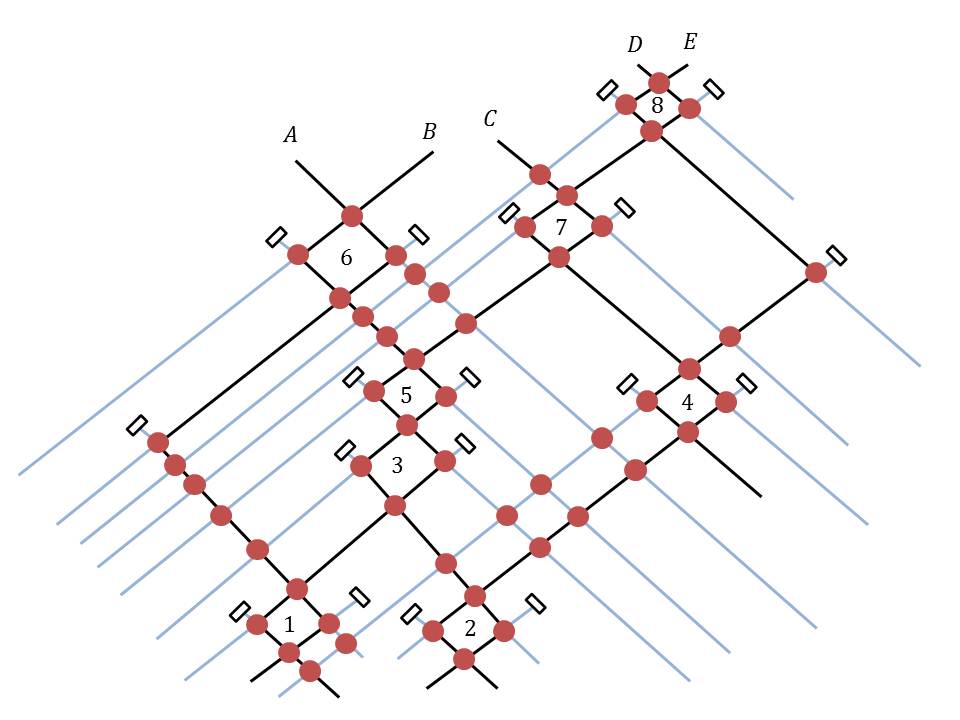}}
  \caption{}
  \label{fig7}
\end{subfigure}
\caption[Another example of ball permuting simulation in nondeterministic ball scattering]{(a) Another example of a quantum circuit with $X$ gates, $1$, $2$, \ldots, $8$, on five labels. In the end we measure the wires $A, B, C, D$ and $E$, in the particle label basis. (b) Programming of particle scattering with intermediate measurements to simulate the $X$ quantum ball permuting circuit of Figure (a) nondeterministically. The labels $1, 2, \ldots, 8$ correspond to the simulation of gates $1, 2, \ldots, 8$ in Figure (a), respectively. Notice that in this example the ancilla particles can experience many intermediate interactions. This example demonstrates that the overall process succeeds in successful simulation, only with small probability, and in general simulations, the probability of success can be exponentially small in the number of particles being used. Therefore, postselecting on the measurement outcomes, one can successfully simulate any $X$ ball permuting quantum circuit. In the end all the particle locations $A, B, C, D$ and $E$ are measured. A drawback in this model of simulation is that it is hard to set the velocities as we proceed to higher layers of the quantum circuit, and we might need to use particles with higher and higher velocity, as we proceed to the top of the circuit.}
\label{fig6-7}
\end{figure}

\begin{figure}[tp]
\centering
\includegraphics[height=3.0in]{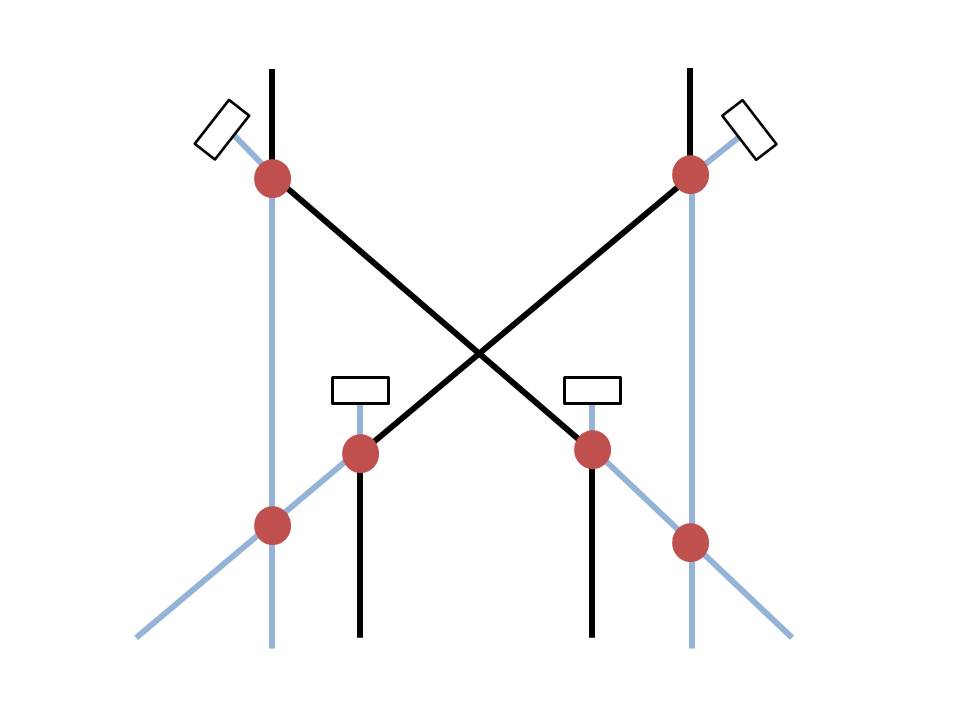}
\caption[Nondeterministic stationary ball scattering gadget to simulate an $X$ operator]{Nondeterministic four-particle gadgets for stationary programming of particle scattering with particle collisions and intermediate measurements. The overall gadget simulates the two label permutation of Figure \ref{fig1}. The objective is to produce superpositions on stationary black particles. Here a stationary particle means a particle that does not move. Initially, two black particles are stationary in the beginning, and we put two more stationary ancilla particles next to them. Then we shoot two ancilla particles from left and right and measure and postselect on them being bounced off from the black particles. Then the two black particles collide with the two stationary ancilla particles and we measure and postselect on the ancilla particles being bounced off in the end. In this scheme it is easier to set the particle scatterings.  
}
\label{fig8}
\end{figure}


\begin{figure}[tp]
\centering
\begin{subfigure}{.5\textwidth}
  \centering
{\includegraphics[height=2.0in]{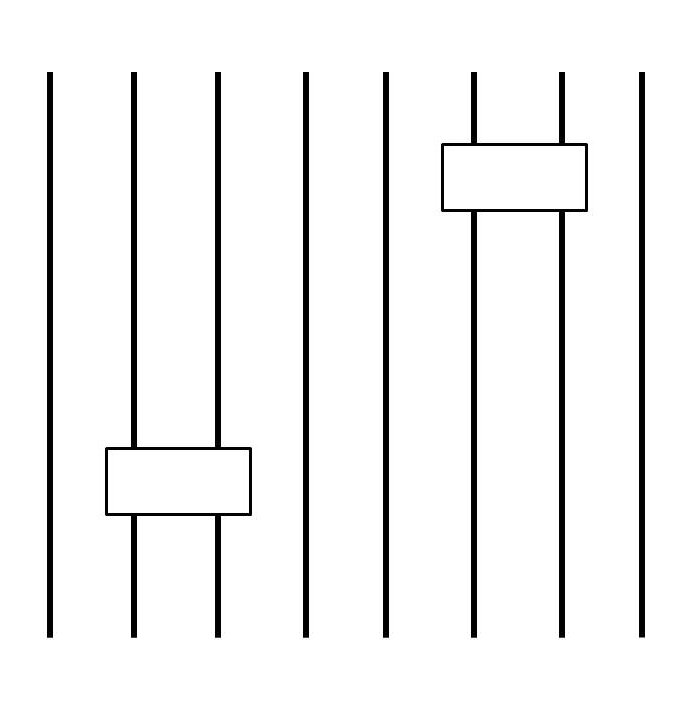}}
  \caption{}
  \label{fig9}
\end{subfigure}%
\hspace{-1cm}
\begin{subfigure}{.5\textwidth}
  \centering
{\includegraphics[height=3.0in]{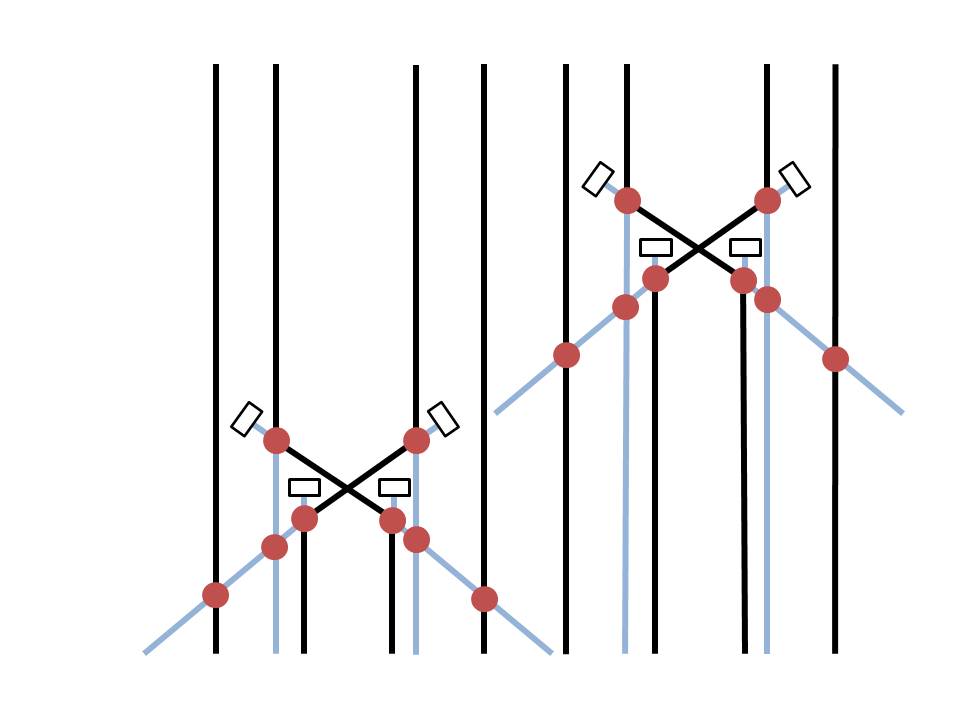}}
  \caption{}
  \label{fig10}
\end{subfigure}
\caption[An example for stationary simulation of an $X$ circuit with ball scattering]{Stationary programming of particle scattering. (a) An example of an $X$ ball permuting circuit on two gates and eight labels. (b) Stationary nondeterministic simulation of the circuit in Figure (a) with ball scattering and intermediate measurements. Each gate in Figure (a) is simulate by a gadget of Figure \ref{fig8}. Except for intermediate interactions, the black particles remain stationary at all the times.}
\label{fig9-10}
\end{figure}

\begin{figure}[tp]
\centering
\includegraphics[height=3.5in]{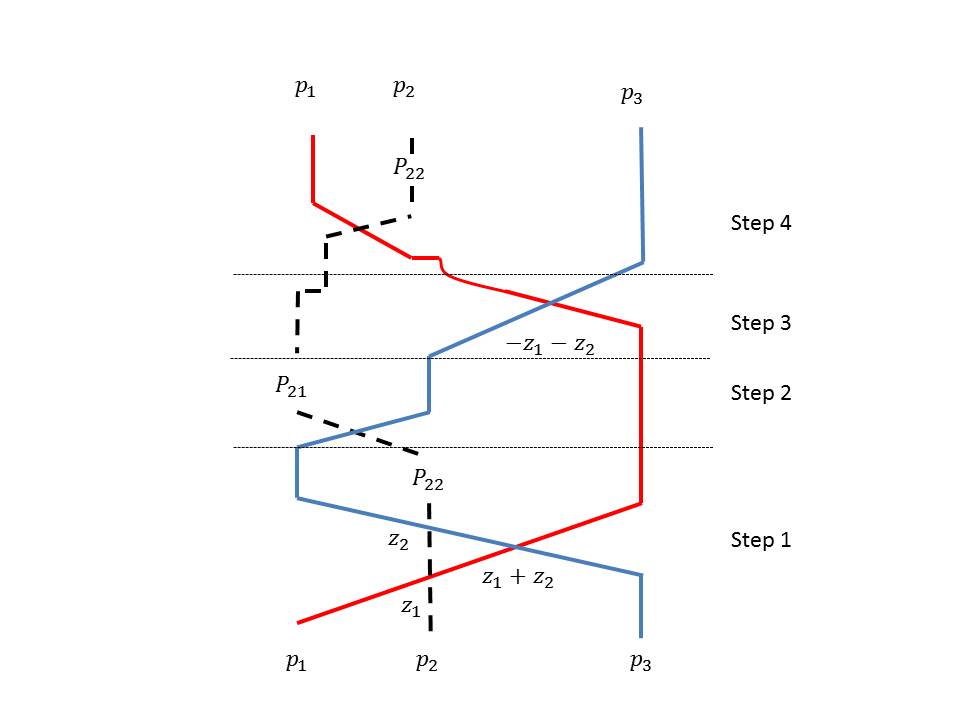}
\caption[Nondeterministic three particle gadget to simulate an $X$ operator with intermediate \textit{non-demolition} measurements]{Nondeterministic three-particle gadget to simulate an $X$ operator with \textit{non-demolition} measurements. This model motivates nondeterministic simulation of ball permuting gates with the model where the one where amplitudes are selected according to the Yang-Baxter equation, and we can do non-demolition intermediate measurements. In such a measurement the outcome of the measurement is the post-measurement quantum state, and the measured labels are being used over and over in this computing scheme.. In this model the braiding of the particles is not necessarily according to straight lines, but without intermediate measurements the obtainable unitary matrices correspond to discrete number of low-dimensional manifolds.

\indent The gadget operates on three labels, and simulates arbitrary rotations on the left (red) and right (blue) wires, non-deterministically. $P_{ij}$ means postselection of label $i$ measured in the location $j$. The $z$ parameters are the rapidities. The construction is done in three steps. First we let three wires to interact with three intersections. Then we measure the middle wire and postselect on measuring label $|2\rangle$ in there. Then in step $2$ we let the two left wires to interact and postselect on the left most wire being label $2$. Then in step $3)$, we let the right most wires interact, and finally in step $4)$ the two left wires have an interaction and we postselect on measuring the label $|2\rangle$ on the middle wire. The aim of steps 2-4 is to reconfigure the momenta back in their original configurations.
}
\label{fig11}
\end{figure}

\subsection{Three-particle Gadgets with Non-demolition Measurements}

In this part we give a three particle gadget with intermediate non-demolition measurements to simulate an arbitrary $X$ rotations, nondeterministically. This model establishes the grounds to understand the complexity of the general Yang-Baxter circuits when we allow intermediate non-adaptive measurements in the particle label basis.  Like before, let the rapidities $z_1= v_1- v_2$ and $z_2 = v_2-v_3$, and consider a Yang-Baxter circuit with  permutation signature $(1 3)$, the permutation which maps $(1, 2,  3)\rightarrow (3,2,1)$:

$$
C(v_1,  v_2, v_3)=H(z_2, 1). H(z_1+z_2, 2).H(z_1, 1)=: C(z_1, z_2)
$$

The corresponding unitary operation amounts to:

\small

$$
\hspace{-5mm}
\dfrac{(1-z_1 z_2)+ i (z_1+z_2) (L_1+L_2) - (z_1+z_2)(z_1 L_2 L_1+ z_2 L_1 L_2)-i z_1 z_2 (z_1+z_2) L_1 L_2 L_1}{\sqrt{(1+z^2_1)(1+z^2_2)(1+(z_1+z_2)^2)}}
$$

\normalsize

There are $3\times 3=9$ choices for a postselection. Let $P_{ij}$ to be the corresponding postselection of label $i$ being located at the $j$'th location. Modulo a normalization factor, this is indeed the projection:

$$
P_{ij}=\sum_{\sigma \in S_3: \sigma(j)=i} |\sigma \rangle \langle \sigma|,
$$

\noindent followed by an appropriate normalization. The following is the list of the normalized output of these measurements if we postselect on the desired outcome. we will drop the overall phases throughout:

\begin{itemize}

\item $P_{11} C(z_1 , z_2)|123\rangle= \dfrac{(1-z_1 z_2) |123\rangle + i (z_1+z_2) |132\rangle}{\sqrt{(1-z_1 z_2)^2 + (z_1+z_2)^2}}$

\item $P_{12}C(z_1 , z_2)|123\rangle= \dfrac{|213\rangle + i z_2 |312\rangle}{\sqrt{1+z_2^2}}$

\item $P_{13} C (z_1 , z_2)|123\rangle= \dfrac{|231\rangle + i z_2 |321\rangle}{\sqrt{1+z_2^2}}$

\item $P_{21}C(z_1 , z_2)|123\rangle= \dfrac{|213\rangle + i z_1 |231\rangle}{\sqrt{1+z_1^2}}$

\item $P_{22}C(z_1 , z_2)|123\rangle= \dfrac{(1-z_1 z_2) |123\rangle - i z_1 z_2 (z_1+z_2) |321\rangle}{\sqrt{(1-z_1 z_2)^2 +( z_1 z_2 (z_1+z_2))^2 }}$

\item $P_{23}C(z_1 , z_2)|123\rangle=\dfrac{|132\rangle + i z_2 |312\rangle}{{\sqrt{1+z_2^2}}}$

\item $P_{31}C(z_1 , z_2)|123\rangle= \dfrac{|312\rangle + i z_1 |321\rangle}{\sqrt{1+z_1^2}}$

\item $P_{32}C(z_1 , z_2)|123\rangle= \dfrac{|132\rangle + i z_1 |231\rangle}{\sqrt{1+z_1^2}}$

\item $P_{33}C(z_1 , z_2)|123\rangle=\dfrac{ (1-z_1 z_2) |123\rangle + i (z_1+z_2) |213\rangle}{(1-z_1 z_2)^2+(z_1+z_2)^2}$

\end{itemize}

We use a three particle gadget to simulate general rotations on two labels. We claim that the three particle gadget of Figure \ref{fig11} does this task. The gadget consists of two circuits, the first of which is a single $P_{22} C(z_1, z_2)$ iteration and induces a rotation on the first and the third labels, and the second circuit makes sure that the velocities are arranged back to their primary locations. Let $x$ and $y$ be the labels of the first and the third locations, respectively. We add a third ancilla color, with a new label $2$, and through each step of the protocol we will make sure that the ancilla label does not superimpose with the other labels.

We go through the steps of the protocol one by one. For simplicity, we drop the normalization factor for the intermediate steps and the states are normalized in the end. Let $z_1=v_1-v_2$ and $z_2= v_2-v_3$ and $z_3= v_1- v_3$ also let the label of the middle particle be $2$:

\begin{itemize}
\item Step 1: $|x2y\rangle \rightarrow (1-z_1 . z_2) |x2y\rangle - i z_1 . z_2 z_3|y2x\rangle$. And the configuration of velocities in the end is: $v_3, v_2, v_1$.

\item Step 2: let the first and second label locations interact, and after that, postselect on the first particle to have the label $2$. Then:

$$
|x2y\rangle \rightarrow (1-z_1 . z_2) |2 x y\rangle - i z_1 . z_2 z_3|2 y x\rangle,
$$

\noindent and the configuration of velocities is $v_2, v_3, v_1$.

\item Step 3: let the second and the third label locations interact, and that is going to be with rapidity $-z_1-z_2$. Then:

$$
|x2y\rangle \rightarrow (1-z_1 . z_2 -z_1 z_2 z_3^2) |2 x y\rangle - i  z_3 |2 y x\rangle,
$$

\noindent and the configuration of the velocities is $v_2, v_1, v_3$.

\item Step 4: finally, let the first and the second label locations interact, and after that postselect on label $2$ to be at the second location, which maps:

$$
|123\rangle \rightarrow (1-z_1 . z_2 -z_1 z_2 (z_1+z_2)^2) |x2y\rangle - i  (z_1+z_2) |y 2 x\rangle,
$$

\noindent and the configuration of velocities is now back to $v_1, v_2, v_3$.

\end{itemize}

The overall action of the protocol is:

$$
|x2y\rangle \rightarrow \cos (\phi_{z_1,z_2}) |x2y\rangle + i \sin (\phi_{z_1, z_2})|y 2 x\rangle
$$

With:

$$
\phi_{z_1,z_2}=\tan^{-1}\left(\dfrac{-(z_1+z_2)}{1-z_1 . z_2 -z_1 . z_2 (z_1+z_2)^2}\right)
$$

\noindent Now it is easy to check that the output state after $t$ iterations of this gadgets is going to be:

$$
|x2y\rangle \rightarrow \cos (t . \phi_{z_1,z_2}) |x2y\rangle + i \sin (t . \phi_{z_1, z_2})|y 2 x\rangle.
$$

\noindent The rapidity $v_2$ is a free parameter and can be set in a way that the angle $\phi_{z_1,z_2}$ is an irrational multiple of $2 \pi$, so using $O(1/\epsilon)$ iterations, one can simulate any of the $X$ rotations with accuracy $\epsilon$.

Given these three particle gadgets, we wish to prove that any model with $X$ operators can be efficiently approximated by a postselected instance of planar Yang-Baxter $H$ operators. We therefore use the discussed three particle gadgets to create arbitrary rotations on two labels, whenever we need them.


\begin{theorem}
Any language in $\XQBALL$ can be efficiently decided by YB quantum circuits with intermediate non-demolition post-selections.
\end{theorem}

\begin{proof}
We establish the proof for the case where the swaps are adjacent ones, the proof for nonadjacent swaps is immediate since adjacent swaps can simulate non-adjacent ones. Let $C$ be any $X$ ball permuting quantum circuit. For each gate in $C$ add a three particle gadget like in Figure \ref{fig11}. We just need to choose the velocity of the ancilla particle in a way that the post-selected gate acts like the desired $X$ operator. The only issue is that since a YB quantum circuit is planar, the ancilla particle might intersect with other particles before and after arriving to its desired gadget. To take care of this issue, we just need to post select on the ancilla particle passing through its intermediate intersections. 
\end{proof}

\section{Evidence for the Hardness of Classical Simulation of the Scattering Quantum Computer}
\label{phcol}

In this section, we combine some of the results from last sections with known facts in complexity theory to give substantial evidence that it is hard to sample from the output probability distribution of the ball scattering model when we allow intermediate particle detections and arbitrary initial states. The approach is to demonstrate that the existence of a feasible sampling scheme results in falsification of statements that are believed to be true. These are statements that have not been proved, but yet no counter examples are known to them. An example is the well-known $\P$ versus $\NP$ question. Most of the researchers in computer science believe that these two objects are not equal. However, a proof of equivalence or a separation does not exist yet, and it might be the case that this problem is an undecidable problem itself. Another example of this kind is the problem of deciding if the polynomial hierarchy is finite or infinite. Indeed, in this section we show that efficient sampling from the output distribution of the ball scattering problem directly implies finiteness of the polynomial hierarchy ($\PH$). As described, the polynomial hierarchy is an extension of nondeterministic polynomial time, $\NP$, to a tower of complexity classes with the form $\NP^{\NP^{\cdots^{\NP}}}$; more precisely $\PH$ is the union of $\Sigma^j_{\P}$ for $j\geq 1$, where $\Sigma^1_{\P}={\P}$, $\Sigma^2_{\P}={\NP}$, and $\Sigma^{j+1}_{\P}= {\Sigma^j_{\P}}^{\NP}$, for $j \geq 2$. Remember that $A^B$ is the machine of class $A$ with oracle access to $B$. It is widely believed that the polynomial hierarchy is infinite, and recently, it has been proved that relative to a random oracle $\PH$ is infinite.

The objective of this section is to demonstrate that it is hard to sample from the output distribution of the ball scattering model, unless the polynomial hierarchy collapses to its third level. Similar proof techniques already exist, for example see \cite{morimae2014hardness, aaronson2011computational}. In section \ref{post} we showed that one can use intermediate demolition measurements on the ball scattering problem of section \ref{models} to come up with a sequence of nondeterministic gates that are able simulate quantum circuits of $\XQBALL$. The gates are nondeterministic, in the sense that they will succeed in their simulation, only if certain measurement outcomes are obtained, and this can happen with exponentially small probability. Also, the proof still holds on arbitrary initial states. After that, in section 5.2, we proved that on arbitrary initial states, the model $\XQBALL$ is equal to the standard $\BQP$. Moreover, we discussed that in order to simulate a standard quantum circuit model $C$ of $\BQP$ in $\XQBALL$, the form of the desired initial state depends on the number $n$ of qubits in $C$ only. We denote this initial state by $|\psi^\star_n\rangle$, or just simply by $|\psi^\star\rangle$. Putting these results together, we observe that the model of ball scattering with intermediate demolition postselection is equal to $\BQP$, if we allow arbitrary initial states. For the sake of clarity, we define $\Post\XQBALL$ with the following definition to capture the discussed ingredients in the model of ball scattering:

\begin{definition}
Let $\Post \HQBALL$ be the class of decision problems that are efficiently solvable using the following resources: 

\begin{itemize}
\item Yang-Baxter ball collision circuits, similar to section 4.2,

\item initial states of the form $|\psi^\star\rangle \tensor |c_1, c_2, \ldots, c_m\rangle$, 

\item postselection on demolition intermediate measurement outcomes. 
\end{itemize}

\noindent Here, $|\psi^\star\rangle$ is the discussed special inital, and $c_1, c_2, \ldots, c_m$ are distinct colors that are also distinct from color species of $|\psi^\star\rangle$. More precisely, this is the class of languages $L \subset \{0,1\}^\star$, for which there is a polynomial time Turing machine $M$ that on any input $x\in \{0,1\}^\star$ outputs the description of a ball scattering setup like in section 4.2, and a polynomial size set of ball colors $\tilde{c}$, along with a special output ball register $c'_0$, with the following properties:

\begin{itemize} 
\item[] 1)  for all $x \in \{0,1\}^\star$, $\operatorname*{Pr}[c'_1 = c_1, c'_2 = c_2, \ldots, c'_m = c_m]>0$, where $c'_l$ is the outcome of the $l$'th intermediate measurement. Denote this event by $C$.

\item[] 2) if $x\in L$, $\operatorname*{Pr} [c'_0 \in \tilde {c} | C]\geq 2/3$. Here $c'_j \in \tilde {c}$ is the event where the color of the ball measured in location $j$, in the end of scattering, is among the colors of the set $\tilde{c}$.

\item[] 3) if $x\notin L$, $\operatorname*{Pr} [c'_0 \in \tilde {c} | C]\leq 1/3$.

\end{itemize}
\label{postH}
\end{definition}

Condition $1$ states that the probability of the event that the classical outcomes of the intermediate measurements match the guessed outcomes $c_1, c_2, c_3, \ldots, c_m$ is nonzero. Notice that as discussed in the section \ref{post}, this probability can be exponentially small, but since we are dealing with conditional probabilities, a nonzero probability is sufficient. Conditions $2$ and $3$ state that the probability of error is bounded. Here $\tilde{c}$ is a set of colors among the colors of $|\psi^\star\rangle$, which correspond to the accepting ball colors. The ball register $c'_0$, is a special location of a ball in the output of ball scattering; $c'_0$ can be viewed as an answer register.

A major observation is the following theorem, stating that although the ball scattering model might be strictly weaker than $\BQP$, the postselected version is equal to $\Post\BQP$, and to $\PP$ because of Aaronson's  $\Post\BQP=\PP$.

\begin{theorem}
$\Post \HQBALL=\Post \BQP= \PP$.
\label{thmxqbll}
\end{theorem}

\begin{proof}
$\Post \BQP = \PP$ is given by the result of Aaronson. Also, $\Post \HQBALL\subseteq \Post \BQP$. In order to see this, observe that the $\BQP$ machine first prepares the initial state $|\psi^\star\rangle \tensor |c_1, c_2, \ldots, c_m\rangle$, encoded with binary strings. Then, whenever an intermediate measurement is done, it just leaves the state along and postpones the measurement to the end of computation. This might give rise to a non-planar quantum circuit, but it is fine, since we are working with $\BQP$. The $\BQP$ measurements are done in a proper basis that encodes the ball color basis. For example, if we encode ball colors with binary representations, then it is sufficient to measure in binary basis and confirm if the digital representation of the ball color (number) is correct. Notice that in the simulation, we are not going to use the balls that have already been measured intermediately again. All the swap gates are applied accordingly. Then in the end we postselect on the desired demolition measurements and in the end we measure the encoded location of the $j$'th ball and confirm if it is among $\tilde {c}$. We can also use CNOT gates to shrink the number of postselections down to one.

In order to see the more interesting direction $\Post \HQBALL\supseteq \Post \BQP$, just we follow the postselected universality of the combined result of sections \ref{post} and \ref{bqpuniversality}, to simulate any computation in $\BQP$. Then notice that any $\Post \BQP$ computation can be deformed in a way that one only needs to postselect on one qubit, and also measure one qubit in the end. $\Post \HQBALL$  uses this deformed $\Post \BQP$ protocol, instead, and uses one of its demolition measurements in the end of computation to simulate postselection of the actual $\Post \BQP$ circuit.
\end{proof}

In section \ref{qcomplexity}, we defined $\HQBALL$ as a variation of the ball permuting model with Yang-Baxter circuits. Later in section \ref{post}, we demonstrated how to use intermediate non-demolition measurements to come up with nondeterministic three particle gadgets that simulate the two particle gates of $\XQBALL$. Thereby, we define the formal model with postselection:

\begin{definition}
Let $\Post \HQBALL^\star$ to be the class of decision problems that are efficiently solvable using initial states of the form $|\psi^\star\rangle$, and non-adaptive planar Yang-Baxter circuits with postselection on non-demolition intermediate measurements. The details of the definition is similar to definition ~\ref{postH}.
\end{definition}

This model does not immediately have a corresponding physical example, but it is interesting that the same result of theorem ~\ref{thmxqbll} is applicable to it:

\begin{theorem}
$\Post\HQBALL^\star=\Post\BQP=\PP$.
\end{theorem}

\begin{proof}
The proof of both directions is similar to the proof theorem ~\ref{thmxqbll}, except that now in the direction $\Post\BQP \supseteq \Post\HQBALL^\star$, the $\BQP$ simulation uses $CNOT$ gates to postpone all intermediate measurements to the end.
\end{proof}

 The equivalence of these quantum models and $\PP$ is already and interesting connection, however, we are two steps away from the major results, that is the connection to the collapse of polynomial hierarchy. In order to achieve this goal, we summarize amazing facts from complexity theory in the following theorem:

\begin{theorem} The following relationships are true for the complexity classes $\PP, \Post\BPP, \NP, \PH,$ $ \Post\BQP$, and $\Sigma^3_\P$:

\begin{itemize}
\item $\Post\BPP \subseteq \BPP^\NP \subseteq \Sigma^3_\P\subseteq \PH$ \cite{arora2009computational}.
\item $\P^{\# \P}=\P^{\PP}$ \cite{arora2009computational}
\item (Toda\cite{toda1991pp}) $\PH \subseteq \P^{\#\P}$ 
\item $\Post \BPP \subseteq \Post \BQP$
\item (Aaronson) $\Post\BQP =\PP$\cite{aaronson2005quantum}
\end{itemize}
\label{amazingtools}
\end{theorem}

\noindent A direct corollary to these containment relations is the following:
\begin{corollary}
$\Post \HQBALL\subseteq \Post \BPP$ implies the collapse of $\PH$ to the third level $\Sigma^3_{\P}$.
\end{corollary}

\begin{proof}
Putting the relations in theorem ~\ref{amazingtools}, along with the result of theorem ~\ref{thmxqbll} together we obtain:

$$
\Post\BPP \subseteq \BPP^\NP \subseteq \Sigma^3_\P\subseteq \PH\subseteq \P^{\# \P}=\P^{\PP}= \P^{\Post\BQP} = \P^{\Post \HQBALL}
$$

\noindent If $\Post\HQBALL \subseteq \Post\BPP$, then $\P^{\Post\HQBALL} \subseteq \P^{\Post\BPP}\subseteq \P^{\Sigma^3_{\P}}$, and thereby $\PH \subseteq \Sigma^3_{\P}$, which results in the collapse of $\PH$ to the third level.
\end{proof}

Given this corollary, the existence of a randomized classical procedure to \textit{exactly} sample from the output of the ball scattering problem, immediately results in the collapse of $\PH$. However, we can even proceed further to come up with a somehow stronger result. Here, we borrow definitions from the notions of randomized simulation from computational complexity theory:

\begin{definition}
Let $P$ be the output a computational (or physical) model, as a probability distribution on $n$ variables $x:=(x_1, x_2, \ldots, x_n)$, we say a randomized algorithm $R$ simulates $P$ within multiplicative constant error, if $R$ produces a probability distribution $\tilde{P}$ on the variables, with the property that there is a constant number $\alpha>1$, which for all $x$:

$$
\dfrac{1}{\alpha}P(x) < \tilde{P}(x)< \alpha P(x).
$$
\end{definition}

\noindent The following theorem states that the existence of a randomized simulation of ball scattering with multiplicative error immediately results in $\Post\HQBALL \subseteq \Post\BPP$:

\begin{theorem}
The existence of a $\BPP$ algorithm to create a probability distribution withing multiplicative error to the actual distribution on $c_1, c_2, \ldots, c_m$ and $c_0$ of definition ~\ref{postH} implies $\Post \HQBALL\subseteq \Post \BPP$.
\end{theorem}

\begin{proof}
The proof is similar to the proof of theorem $2$ in \cite{bremner2010classical}. Suppose that there is a procedure which outputs the numbers $x_1, x_2, \ldots, x_m, y$ such that:

\begin{eqnarray*}
1/\alpha \operatorname*{Pr}[c'_0=c_0 ; c'_1=c_1, c'_2=c_2, \ldots, c'_m=c_m]&<&\\
\operatorname*{Pr}[y=c_0; x_1=c_1, x_2=c_2, \ldots, x_m=c_m]&<&\alpha \operatorname*{Pr}[c'_0=c_0 ; c'_1=c_1, c'_2=c_2, \ldots, c'_m=c_m],
\end{eqnarray*}

\noindent for every list of colors $c_1, c_2, \ldots, c_m$ and $c_0$. Notice that if this is true, it should also be true for all marginal probability distributions. Denote the vector $(c'_1, c'_2, \ldots, c'_m,  c'_0)$ by $(\tilde{c'}, c'_0)$, and $(x_1, x_2, \ldots, x_m,  y)$ by $(\tilde{x}, y)$. Then the conditional probabilities also satisfy:

$$
1/{\alpha^2}\operatorname*{Pr}[c'_0=c_0|\tilde{c'}=\tilde{c}] <\operatorname*{Pr}[y=c_0|\tilde{x}=\tilde{c}]<\alpha^2 \operatorname*{Pr}[c'_0=c_0|\tilde{c'}=\tilde{c}].
$$

\noindent Now let $L$ be any language in $\Post\HQBALL$. Then if $x\in L$, $\operatorname*{Pr}[d'=d| C]\geq 2/3$ and otherwise $\leq 1/3$. Suppose that $x\in L$, then in order for $\operatorname*{Pr}[y=d|\tilde{x}=\tilde{c}]>1/2$, it is required that:

$$
1/2 < 1/{\alpha^2} (2/3)
$$

\noindent which means if $1 <\alpha < \sqrt{4/3}-O(1)$, then the $\Post\BPP$ algorithm recognizes $x$ with bounded probability of error. The probability of success can be increased by using the majority of votes' technique.

\end{proof}

\noindent Putting it all together, we finally mention the result of this section in the following corollary:

\begin{corollary}
There exists no polynomial time randomized algorithm to simulate ball scattering models of section 4.4.2 and 4.4.3 within multiplicative constant error, unless $\PH$ collapses to its third level.
\end{corollary}

\section{Open Problems and Further Directions }
\label{openproblems}

\begin{enumerate}
\item The actual model of particle scattering corresponds to a unitary scattering matrix with at most $n$ degrees of freedom. Therefore, as a manifold, the dimension of the unitary group that is generated by the Yang-Baxter scatterings have small dimension. Although we proved that with intermediate measurements the structure of the group will expand to exponential dimension, it is still unknown if classical simulation is allowed for the problem without intermediate measurements. There examples \cite{jordan2009permutational} of models that generate discrete number of unitary operators, and still they are believed to be able to solve problems that are hard for classical computation. Moreover, the structure of the described manifold, as a mathematical object is interesting.

\item No physical example for the model $\ZQBALL$ is mentioned. Since the model is proved to be equivalent to $\BQP$ it is interesting to see if there are actual physical models that capture the dynamics of $\ZQBALL$. This might be a model for molecular dynamics with exchange interactions for which the coupling terms depend on the type of molecules being exchanged.

\item While we hope to find a lower bound of $\prod_{\lambda \vdash n} SU(V_\lambda)$ for $G$, on $\bigoplus_{\lambda\vdash n} V_\lambda$, there is some evidence supporting correlated action of $G$ on pairs of subspaces $V_\lambda$ and $V_{\lambda^\star}$, whenever $\lambda$ is dual to $\lambda^\star$. As a first input, the subspaces of dual partitions have equal dimensions, therefore, even dense action on each, separately, does not immediately imply decoupling of the two subspaces.

If we enumerate the $YY$ basis of $\lambda$, by $|i_1\rangle, |i_2\rangle, \ldots, |i_d\rangle$, where $d$ is the dimension of $V_\lambda$. Each element of these bases, $|i_j\rangle$, have a corresponding tableau, and the transpose of that tableau is a tableau of $\lambda^\star$, and it is bijectively related to a $YY$ base element, $|i'_j\rangle$, of $V_{\lambda^\star}$. Let $X$ be an element of the Lie algebra of $G$, whose action on $V_\lambda$ has the form:

$$
X\overset{\lambda}{\twoheadrightarrow} \sum_{j\in d} \alpha_j |j\rangle\langle j| + \sum_{i<j} \beta_{i j}( |i\rangle \langle j| + |i\rangle \langle j| )+i \delta_{i j}( |i\rangle \langle j| - |i\rangle \langle j| ).
$$

\noindent By $\overset{\lambda}{\twoheadrightarrow}$, we mean the restriction of $X$ on $V_\lambda$, by projecting out everything else from other subspaces. It is believable that the similar action on the dual block is according to:

\begin{claim}
The projection of $X$ on $\lambda^\star$ is according to:

$$
X\overset{\lambda^\star}{\twoheadrightarrow} \sum_{j\in d} \pm\alpha_j |j'\rangle\langle j'| + \sum_{i'<j'} \pm\beta_{i j}( |i'\rangle \langle j'| + |i'\rangle \langle j'| )\pm i \delta_{i j}( |i'\rangle \langle j'| - |i'\rangle \langle j'| ).
$$
\end{claim}

\noindent If this is the case, then it immediately implies the mentioned coupling. In order to see this, notice that $ X\overset{\lambda}{\twoheadrightarrow} 0$ implies $X\overset{\lambda^\star}{\twoheadrightarrow}0$, and vice versa. Checking the situation for the dual partitions $(2,1,1)$ and $(3,1)$ for four labels, one can confirm that this is true. However, even in this case, quantum efficient computation is not ruled out. For example, consider the situation where the initial state is according to $\dfrac{1}{\sqrt{2}}(|x\rangle +|x'\rangle)$. Then a coupled action of the form $U \oplus U^\star$ maps $\dfrac{1}{\sqrt{2}}(|x\rangle +|x'\rangle)$ to $\dfrac{1}{\sqrt{2}}(U|x\rangle +U^\star|x'\rangle)$. Then finalizing with the same state we get $\dfrac{1}{2} (\langle x | U | x \rangle + \langle x' | U^\star | x' \rangle)= \Re \langle x | U |x \rangle$. However, the problem of reading an entry of a quantum circuit is already known to be $\BQP$-complete.

\item It is not clear if there is a way to act as $SU(V_\lambda)$ on all of the partitions of $n$. The bridge lemma works if two orthogonal subspaces are being joined together. Therefore, it is interesting to extend the bridge lemma to more subspaces. Moreover, there are cases where two subspaces of equal dimensionality take part in a single branching rule. Therefore, the bridge lemma is not applicable. For example consider the partition $(3,2,1)$. The branching rule for this partition involves decomposition into the direct sum of three subspaces corresponding to $(2,2,1)$, $(3,1,1)$, and $(3,2)$. However, not only the subspaces corresponding partitions $(2,2,1)$ and $(3,2)$ have the same dimensions, but also they are dual partitions, and it might be the case that even the action of $G$ cannot be decoupled from the two. 

\item We proved that single amplitudes can be efficiently approximated within additive error by $\DQC 1$ computation. It is unkown if an efficient $\DQC 1$ (approximate) sampling scheme exists for the ball permuting model on separable initial states. Such a sampling exists within the $\BQP$ class. Finding efficient sampling within any other class that is believed to be below $\BQP$ is interesting. Also a lower-bound for the ball permuting model in this case is unknown.

\item It remains open to generalize the result of section \ref{trace} to arbitrary quantum models with Hilbert spaces based on group algebras. As discussed, such a generalization relies on the ability to do the following. Given the description of the generators of a group, is it possible to encode the elements of an arbitrary element with binary strings with nearly $\log |G|$ bits so that the action of group elements on each other is implementable by reversible circuits that affect around $\log \log |G|$ bits at a time? We could establish this for the symmetric group using the factorial number basis. Are there other groups with similar property?
\end{enumerate}

\section{Conclusion}

\noindent We applied some of the ideas from complexity theory and quantum complexity theory to a small regime of theoretical physics, the problem of particle scattering in integrable theories of $1+1$ dimensions. We found that the complexity of the model essentially depends on the initial superpositions that the particles start out from. Single amplitudes of the theory can be approximated with polynomial size window of additive approximation within the one-clean-qubit if no initial superposition is allowed. We defined the ball permuting model as a natural generalization of the particle scattering model, which has less restrictions on the amplitudes of the theory. We proved that if special initial superpositions are allowed, the ball permuting model can simulate and efficiently sample from the output distribution of arbitrary quantum computers. We used this result to prove that if the particle scattering model is equipped with demolition intermediate measurements, then it is hard to sample from the output distributions of this model on a classical computer, unless the polynomial hierarchy collapses to the third level. In all of these models we consider that the particles are distinguishable. 

\section{Acknowledgment}

I am especially indebted to my advisor Scott Aaronson for his unlimited wisdom and support throughout this research. I would like to thank Greg Kuperberg for insightful discussions and notes. I am thankful Adam Bouland, Robin Kothari, Hamed Pakatchi, Usman Naseer and Hossein Esfandiary for valuable discussions. This research was funded by NSF Waterman of Scott Aaronson.

\begin{appendices}

\chapter {A Word on Special Relativity}
\label{Lorentz}

Lorentz transformations are those which preserve length with respect to the metric $\eta= \text{diag} (-1,1,1,1)$. The basis is $(t,\textbf{x})=(x^0,x^1,x^2,x^3)$; time and three space coordinates. The Greek letters ($\mu, \nu, \ldots$) run from 0 to 3, the ordinary letters ($a, b, \ldots$) run from 1 to 3. An infinitesimal length should be conserved according to this metric, and therefore $\eta_{\mu\nu}dx^\mu dx^\nu=\eta_{\mu\nu}dx'^\mu dx'^\nu$. If we define the matrix elements of the Lorentz transform as $L^{\mu} \hspace{0.05cm}{_\nu} = \partial x'^\mu/ \partial x^\nu$. So the general transforms are of the form $x'^\mu=L^{\mu} \hspace{0.05cm}{_\nu} x^\nu + a^\mu$, and the transforms must satisfy:

$$
\eta_{\mu\nu}L^{\mu} \hspace{0.05cm}{_\alpha}L^{\nu} \hspace{0.05cm}{_\beta}=\eta_{\alpha\beta}
$$

\noindent or equivalently using the matrix form:

$$
L^T \eta L = \eta.
$$

\noindent This equation tells us that $Det(L)=\pm1$. In general, these construct the\textit{ Poincar\'e} group. Any representation $U$ of the general Poincar\'e group for quantum states should satisfy:

\begin{equation}
U(L_1,a_1)U(L_2,a_2)= U(L_1 L_2, L_1 a_2 + a_1)
\end{equation}

\chapter{The Scattering Matrix}
\label{scattering}

This work analyzes the computational complexity of computing scattering amplitudes in certain two dimensional quantum theories. In order to obtain a point of reference, in this part, we briefly describe how these quantities are obtained as a function of physical interactions.

Let $\mathcal{H}= \bigotimes^n_{j=1} \mathcal{H}_j$ be a Hilbert space composed of different subsystems $\mathcal{H}_j$. A free Hamiltoinian for this system consists of Hermitian operators each affecting one of the subsystems, and acting as identity on the other parts. The time evolution of a free Hamiltonian is factorized over the subsystems of the Hilbert space. Any Hamiltonian operator on this Hilbert space that is not a free Hamiltonian is an interaction Hamiltonian. An interaction couples different subsystems by terms like $A\tensor B$, where $A$ is some Hermitian operator on some subsystem and $B$ is another Hermitian operator on another subsystem. Let $H_0$ and $V$ be any such free and interaction terms. The Hamiltonian describing the overall time evolution is then $H=H_0+V$. All of the operators are time-independent. Let $\Psi_t$ be the time dependent quantum state in the Hilbert space; then the unitary time evolution is according to $U(t,\tau)=\exp (-i H(t-\tau))$, which maps $|\Psi_\tau\rangle \mapsto |\Psi_t\rangle$. Define $|\Phi_t\rangle: = \exp (i H_0 t) |\Psi_t\rangle$, as the interaction picture of $|\Psi_t\rangle$, which factors out the free term. The evolution of $|\Phi_\tau\rangle$ is then given by the unitary operator:

\begin {equation}
\tilde {U} (t,\tau)=\exp(iH_0 t)\exp (- i H (t-\tau))\exp(-iH_0 \tau),
\end {equation}

\noindent and $|\Phi_t\rangle= \tilde {U} (t,\tau)|\Phi_\tau\rangle$. Scattering matrix is a unitary operator which relates $\Phi_{-\infty}$ to $\Phi_{+\infty}$, and is given according to $S=\tilde{U} (+\infty, -\infty)$. The scattering matrix relates the states in far past to the states in far future. The interaction picture for the observables is defined as $V(t)=\exp(iH_0t)V\exp(-iH_0t)$. The scattering matrix is then:

\begin {equation} 
S= \mathrm{T} \exp \left(-i \int_{-\infty}^{+\infty} V(t)dt\right),
\label {S}
\end {equation}

\noindent where $\mathrm{T}$ is the time ordering operator sorting the operators in decreasing order from left to right, given by:

$$
S=1+\sum^\infty_{n=1} \dfrac{(-i)^n}{n!}\int_{-\infty}^\infty dt_1 dt_2 \ldots dt_n \theta (t_n, \ldots,t_2, t_1) V(t_1) V(t_2)\ldots V(t_n)
$$

\noindent where $\theta: \mathbb{R}^n \rightarrow \{0,1\}$, is the indicator of the event $t_1> t_2 > \ldots > t_n$. Or we can just write:

$$
S=1+\sum^\infty_{n=1} \dfrac{(-i)^n}{n!}\int_{-\infty}^\infty dt_1 dt_2 \ldots dt_n \mathrm{T}\Big ( V(t_1) V(t_2)\ldots V(t_n)\Big )
$$

This expansion is known as the Dyson series. There is also a similar expression for the scattering matrix of quantum field theory. Here, we sketch the logic given in \cite{weinberg1996quantum}. We need to first revisit the Hilbert space slightly. Consider the decomposition of a Hilbert space according to the number of particles $\bigoplus^\infty_{i=0} \mathcal{H}_i$. The zeroth term $\mathcal{H}_0$ consists of a single vector $|0\rangle$, which is called the vacuum state. This corresponds to the minimum energy state. The other terms $\mathcal{H}_i$ for $i\geq 1$, corresponds to $i$ particles, each with a four momentum $p_j= (p^0, \pmb{p})$ and $A_j$ a set of discrete internal degrees of freedom, for $j \in [i]$. The discrete degrees of freedom are internal degrees of freedom like spin. Denote the vectors in the combined Hilbert space with Greek letter $\alpha, \beta, \ldots$. The amplitude of scattering for state with label $\alpha$ to state with label $\beta$ is then $S_{\beta, \alpha}=\langle \Psi_\beta| S |\Psi_\alpha\rangle$. A creation operator $a^\dagger_\alpha$ is naturally defined as the operator which maps vacuum state to the state with $\alpha$ ingredients in its corresponding block of the Hilbert space, \i.e., $|\Psi_\alpha\rangle= a^\dagger _\alpha |0\rangle$.

We need a theory with locality, unitarity and Lorentz symmetry. Unitarity is simply captured by choosing Hermitian interactions. 
We want the scattering amplitudes to be Lorentz invariant. The Lorentz group should be a symmetry of this theory, and we can hardly imagine a way of defining operators that only depend on time and are also invariant under Lorentz's space-time transformations. For a brief introduction to the Lorentz group see appendix \ref{Lorentz}. Therefore, a necessary condition is that there exist a well-defined interaction density in spatial coordinates, in the sense that $V(t)=\int d^3 x \mathcal {H}(\textbf{x},t)$. Cluster decomposition in quantum field theory relates the structure of the interaction density to a seemingly obvious logic: if two physical processes occur independently the amplitude of the combined process products of separate amplitudes in each subsystem.  This is something like statistical independence in probability theory. Then cluster decomposition forces the interaction densities to be composed of integration over creation and annihilation (the Hermitian conjugate of the creation operator) operators, in such a way that all the creation operators appear on the left hand side of annihilation operators, and the arguments in these integrals contain at most one Dirac delta like singularity. Therefore, in short, the interactions should be specific combinations of the annihilation and creation operators in such a way that they are local and transform according to Lorentz symmetry. That is:

$$
U(L,a) \mathcal{H}(x) U^{-1} (L, a) = \mathcal{H}(L x + a)
$$

and locality,

$$
[\mathcal{H}(0, \pmb{x}), \mathcal{H}(0, 0)]=0 \hspace{1cm} |\pmb{x}|>0
$$

If we take a close look at the Greek letter $\alpha$ we see that these are momentum, and other discrete quantum numbers. So the interaction density depends on several (creation and annihilation) operators which they themselves depend on a finite number of momentum species. How can we impose Lorentz scalars when your interaction density depends only on finite numbers of momenta. The solution is the quantum field, an operator that depends on space time and carries all possible creation and annihilation operators in a suitable superposition:

$$
\Psi_i^\dagger (x) = \sum_\sigma \int d^3 p A_i (p, x, \sigma) a^\dagger (p, \sigma)
$$

\noindent where $A_l(\cdot, \cdot, \cdot) $ is a suitable function, that can be a vector or scalar and so on. A quantum field is then a proper combination of annihilation and creation operators in such a way that they respect the Lorentz transformation as operators:

$$
U(L,a) \Psi^\dagger_i (x) U^{-1} (L, a) = \sum_i M_{i j} \Psi_j^\dagger (L x + a)
$$

Therefore, the quantum fields are building blocks of the quantum field theory, operators that each are Lorentz invariant independently. The legitimate interaction densities then correspond to all the ways that we can combine quantum fields to a Hermitian operator. Represent points of space-time with $x= (t, \pmb{x})$. Expanding equation ~\ref{S} we get:

\begin {equation}
\hspace{-1cm}
S_{\beta,\alpha}=\sum_{N=0}^\infty \dfrac{(-i)^N}{N!}\int d^4x_1 d^4 x_2 \ldots d^4x_N \langle 0| a_\beta T\{\mathcal {H}(x_1)\ldots \mathcal{H}(x_N)\}a^\dagger_\alpha|0\rangle
\end {equation}

The zeroth term in the summation is $1$, by convention. Each argument in the integral is a combination of creation and annihilation operators and quantum fields which are integrations over creation and annihilation operators themselves, and these along with their (anti-) commutation relations pose a combinatorial structure and correspond to a set of Feynman diagrams, which we will not delve into. Thereby, given the interaction density of a theory, we construct sequences of Feynman diagrams each amounting to a complex number, and then the summation over them gives the overall amplitude of an interaction.

\chapter{Some Remarks on Lie Algebras}
\label{LieAlgebra}

Let $a_1, a_2, \ldots, a_n$ be Hermitian operators. We are interested in the Lie group generated by operators $A_k(t):=\exp(i a_k t)$, for $t \in \mathbb{R}$. $G$ is a Lie group, and is fairly complicated, therefore it is worth looking at the vector space tangent to the identity element. Such a vector space is an alternative definition of a Lie algebra. Denote this Lie algebra by $g$. Since the exponential map is an isomorphism, the dimension of a Lie algebra, as a vector space, matches the dimension of the manifold $G$. In addition, for some vector space $V$, the containment of $\su(V)$ in $g$, implies denseness of $G$ in $SU(V)$. Starting with $A_1 (t), A_2 (t), \ldots, A_n (t)$, it is well-known that the corresponding Lie algebra consists of all the elements that can be ever generated by taking the Lie commutators $i [\cdot , \cdot ]$ of the elements, and Linear combination of the operators over $\mathbb{R}$.

Here we prrovide some intuitiion about $g$. The following remarks are followed from \cite{deutsch1995universality,lloyd1995almost, divincenzo1995two}. Notice that this is not a proof. The Trotter formula states that for any two operators $H$ and $V$:

\begin{equation}
\exp \left(i (H+V) \right) = \lim_{m\rightarrow \infty} \left(\exp(i H/m) \exp (i V/m) \right)^m,
\end{equation}

\noindent and the rate of convergence is according to:

$$
\exp\left(i (H+V)\right)= \left(\exp(i H/m) \exp (i V/ m )\right)^m+ \mathrm{O}(1/m).
$$

\noindent Therefore, this implies that if two elements $H$ and $V$ in $g$ can be approximated within arbitrary accuracy, then $H+V$ can also be approximated efficiently. Next, observe the following identity:

\begin{equation}
\exp ([H,V]) = \lim_{m\rightarrow \infty} \left( \exp (i H/\sqrt {m}) \exp (i V/\sqrt {m})\exp (-i H/\sqrt {m})\exp (-i V/\sqrt {m})\right)^m,
\end{equation}

\noindent with the rate of convergence:

$$
\left( \exp (i H/\sqrt {m}) \exp (i V/\sqrt {m})\exp (-i H/\sqrt {m})\exp (-i V/\sqrt {m})\right)^m= exp ([U,V])+ \mathrm{O}(1/\sqrt{m})
$$

\noindent Therefore this explains the commutator. Lastly, we need to take care of the scalar multiplications. Notice that if an element $\exp (i H)$ is approximated in $G$, then for any $k\geq 1$, $\exp (i H k)$ can be approximated similarly by repeating the process for $k$ times. Now assume that the sequence $\prod_i A_i (t_i)$ approximates the element with the desired accuracy, then the process $(\prod_i A_i (t_i/k m))^m$ can approximate $\exp (i H/k)$ in any way that we want by tuning $m$.

\chapter{$\P^{\# \P}$ Algorithm for the Classical Ball Permuting Problems}
\label{sharpp}

Here we give a brief justification for the containment of the three problems mentioned section \ref{classicalball} in $\P^{\#\P}$. All of the mentioned problems (1, 2, and 3) can be solved within polynomial amount of space: first notice that if one can exactly compute the probability of the target permutation, then deciding if the probability is large or small (the first problem), or the problem of deciding if the probability is zero or not (the second problem) can be solved immediately. In order to solve the problem in polynomial space, just notice that probability of observing the target permutation, $p_\pi$, can be computed as:

$$
p_{\pi}= \sum_{\substack{S\subseteq [m]\\ \pi_S= \pi}} \prod_{j\in S} p_j \prod_{j\in [n]/S} 1-p_j.
$$

\noindent Where for any subset $S \subseteq [m]$, $\pi_S$ is the permutation that is obtained by multiplying the swaps $\prod_{l\in S}(i_l , j_l)$ in order. All such terms can be enumerated and added within polynomial amount of space. This can also be viewed with weighted counting: if all of the $p_j$ probabilities in the list are set to $\dfrac{1}{2}$, then $p_\pi = \dfrac{|\{S\subset [m]: \pi_S=\pi \}|}{2^m}$, and then the problem is reduced to the computation of $|\{S\subset [m]: \pi_S=\pi \}|$ which can be done within $\#\P$. For general probabilities, the counting is weighted by $2^m \prod_{j\in S} p_j$.

We can go one step further, and prove that:

\begin{theorem}
If the probabilities $p_1, p_2, \ldots, p_m$ are rational numbers, given oracle access to $\# \P$, the computation of $p_\pi$ can be done in polynomial time.
\end{theorem}

\begin{proof}
The objective is to compute:

$$
p_\pi=\sum_{\substack{S\subseteq [m]\\ \pi_S= \pi}} \prod_{j\in S} p_j \prod_{j\in [n]/S} 1-p_j.
$$

\noindent Since $p_j$'s are rational numbers, we can write $p_j:= \dfrac{m_j}{M_j}$, where $m_j\leq M_j \in \mathbb{N}$. Define $M:= \prod_j m_j$. Then:
$$
p_\pi=\dfrac{1}{M}\sum_{\substack{S\subseteq [m]\\ \pi_S= \pi}} \prod_{j\in S} m_j \prod_{j\in [n]/S} M-m_j.
$$

\noindent Define $N_\pi := M p_\pi$ and $N_S:=\prod_{j\in S} m_j \prod_{j\in [m]/S} M-p_j$ for any $S\subseteq [m]$. Also define $Q:=\max_{S\subseteq [m]} \{N_S\}$, and $q:=\log_2 Q$. Then the claim is that:

$$
N_\pi=\sum_{\substack{S\subseteq [m]\\ \pi_S= \pi}} N_S
$$

\noindent can be computed in $\P^{\#\P}$. To see this, define the function $f: 2^{[m]}\times \{0,1\}^q\rightarrow \{0,1\}$ like below, where $2^{[m]}$ is the set of subsets of $[m]$:

$$
f(S, y)= 
\begin{cases}
    1 & \pi_S=\pi \hspace{1mm} \& \hspace{1mm}|y|\leq N_S \\
    0              & \text{otherwise}
\end{cases}.
$$

\noindent Clearly:

$$
\sum_{S\subseteq [m], y\in \{0,1\}^q} f(S,y)= N_\pi.
$$

\noindent There is an $\NP$ machine $M$ to decide if $f(S,y)$ is satisfiable. Then query the description of $M$ to the $\#\P$ oracle and find the number of accepting paths, which amounts to $N_\pi$.
\end{proof}

\end{appendices}

\include{biblio}
\bibliography{permut}{}
\bibliographystyle{plain}

\end{document}

%% file: header.tex
\newtheorem{theorem}{Theorem}[chapter]
\newtheorem{definition}{Definition}[chapter]
\newtheorem{proposition}[theorem]{Proposition}
\newtheorem{corollary}[theorem]{Corollary}
\newtheorem{claim}{Claim}
\newtheorem{lemma}[theorem]{Lemma}

\newcommand{\VDP}{\mathsf{VDP}}
\newcommand{\Rev}{\mathsf{Rev}}
\newcommand{\EDP}{\mathsf{EDP}}
\newcommand{\WPPP}{\mathsf{WPPP}}

\newcommand{\BPP}{\mathsf{BPP}}
\newcommand{\PP}{\mathsf{PP}}
\newcommand{\BPL}{\mathsf{BPL}}

\newcommand{\Post}{\mathsf{Post}}

\renewcommand{\P}[0]{\ensuremath{\mathsf{P}}}

\renewcommand{\L}[0]{\ensuremath{\mathsf{L}}}
\newcommand{\TIME}{\mathsf{TIME}}
\newcommand{\Trace}{\mathsf{Trace}}
\newcommand{\DQC}{\mathsf{DQC}}

\newcommand{\PH}{\mathsf{PH}}
\newcommand{\SPACE}{\mathsf{SPACE}}
\newcommand{\NTIME}{\mathsf{NTIME}}
\newcommand{\NSPACE}{\mathsf{NSPACE}}

\newcommand{\PSPACE}{\mathsf{PSPACE}}
\newcommand{\LOGSPACE}{\mathsf{LOGSPACE}}
\newcommand{\NC}{\mathsf{NC}}
\newcommand{\NL}{\mathsf{NL}}

\newcommand{\SAT}{\mathsf{SAT}}
\newcommand{\Almost}{\mathsf{Almost}}
\newcommand{\AC}{\mathsf{AC}}
\newcommand{\NP}{\mathsf{NP}}

\newcommand{\Ball}{\mathsf{Ball}}
\newcommand{\BALL}{\mathsf{BALL}}
\newcommand{\RBALL}{\mathsf{RBALL}}
\newcommand{\RBall}{\mathsf{RBall}}
\newcommand{\DBALL}{\mathsf{DBALL}}
\newcommand{\DBall}{\mathsf{DBall}}
\newcommand{\NBALL}{\mathsf{NBALL}}
\newcommand{\NBall}{\mathsf{NBall}}

\newcommand{\HQBALL}{\mathsf{HQBALL}}
\newcommand{\XQBALL}{\mathsf{XQBALL}}
\newcommand{\YQBALL}{\mathsf{YQBALL}}
\newcommand{\ZQBALL}{\mathsf{ZQBALL}}

\newcommand{\BQP}{\mathsf{BQP}}
\newcommand{\Poly}{\mathsf{poly}}

\newcommand{\C}{\mathbb{C}}

\newcommand{\tensor}{\otimes}

\newcommand{\su}{\mathfrak{su}}